\documentclass[11pt]{article}

\usepackage[
letterpaper,
top=1in,
bottom=1in,
left=1in,
right=1in]{geometry}
\usepackage{bbm}

\usepackage{amsmath, amssymb}
\usepackage{amsthm}
\usepackage{color}
\usepackage{mathtools,amsfonts,enumerate}
\usepackage{thmtools, thm-restate}
\usepackage[titletoc,title]{appendix}
\usepackage{hyperref,cleveref}
\newtheorem{theorem}{Theorem}[section]
\newtheorem{claim}[theorem]{Claim}
\newtheorem{lemma}[theorem]{Lemma}
\newtheorem{fact}[theorem]{Fact}
\newtheorem{definition}[theorem]{Definition}
\newtheorem{observation}[theorem]{Observation}
\newtheorem{assumption}[theorem]{Assumption}
\newtheorem{algorithm}[theorem]{Algorithm}
\usepackage{dsfont}

\newcommand{\vol}{{\hbox{\bf vol}}}
\newcommand{\current}{{\hbox{current}}}

\newcommand{\dtfs}{{\hbox{find-size}}}
\newcommand{\dtfr}{{\hbox{find-root}}}
\newcommand{\dtfv}{{\hbox{find-value}}}
\newcommand{\dtfm}{{\hbox{find-min}}}
\newcommand{\dtcv}{{\hbox{change-value}}}
\newcommand{\dtlink}{{\hbox{link}}}
\newcommand{\dtcut}{{\hbox{cut}}}

\newcommand{\del}{\tilde{\delta}}

\newcommand\myeq{\mathrel{\overset{\makebox[0pt]{\mbox{\normalfont\tiny\sffamily def}}}{=}}}
\newcommand{\polylog}{\operatorname{polylog}}
\newcommand{\tab}{.\hskip.1in}
\newcommand{\ex}{\operatorname{ex}}

\newcommand{\tempfix}[2]{#2}
\newcommand{\cutsize}{1.1\del}

\DeclarePairedDelimiter\floor{\lfloor}{\rfloor}

\usepackage[colorinlistoftodos,prependcaption,textsize=tiny]{todonotes}

\title{Deterministic Near-Linear Time Minimum Cut in Weighted Graphs}
\author
{
Monika Henzinger\thanks{
Institute of Science and Technology Austria (ISTA), Klosterneuburg, Austria.
See funding information in the 
acknowledgement  section.
email: \texttt{monika.henzinger@ista.ac.at}}
\and
Jason Li\thanks{Simons Institute, UC Berkeley. email: \texttt{jmli@alumni.cmu.edu}}
\and
Satish Rao\thanks{UC Berkeley. email: \texttt{satishr@berkeley.edu}}
\and
Di Wang\thanks{Google Research. email: \texttt{wadi@google.com}}
}
\begin{document}
\maketitle
 \sloppy 
 \thispagestyle{empty}

\begin{abstract}

  In 1996, Karger \cite{KargerStoc} gave a startling randomized algorithm that finds a
  minimum-cut in a (weighted) graph in time $O(m\log^3n)$ which he
  termed near-linear time meaning linear (in the size of the input) times a polylogarthmic
  factor. In this paper, we give the first deterministic algorithm
  which runs in near-linear time for weighted graphs. 

  Previously, the breakthrough results of Kawarabayashi and Thorup \cite{KawarabayashiT19}
  gave a near-linear time algorithm for simple graphs (which was
  improved to have running time $O(m \log^2 n \log \log n)$ in
  \cite{DBLP:journals/siamcomp/HenzingerRW20}.)  The main technique here is a clustering procedure that
  perfectly preserves minimum cuts. Recently, Li \cite{Li21} gave an
  $m^{1+o(1)}$ deterministic minimum-cut algorithm for weighted graphs;
  this form of running time has been termed ``almost-linear''.  Li
  uses almost-linear time deterministic expander decompositions which
  do not perfectly preserve minimum cuts, but he can use these
  clusterings to, in a sense, ``derandomize'' the methods of Karger.

  In terms of techniques, we provide a structural theorem that says there
  exists a sparse clustering that preserves minimum cuts in a weighted
  graph with $o(1)$ error. In addition, we construct it
  deterministically in near linear time.  This was done exactly for
  simple graphs in \cite{KawarabayashiT19,DBLP:journals/siamcomp/HenzingerRW20} and with polylogarithmic error for
  weighted graphs in \cite{Li21}. Extending the
  techniques in
  \cite{KawarabayashiT19,DBLP:journals/siamcomp/HenzingerRW20} to
  weighted graphs presents significant challenges,
  and moreover, the algorithm can
  only polylogarithmically approximately preserve minimum cuts. A
  remaining challenge is to reduce the polylogarithmic-approximate
  clusterings to $1+o(1/\log n)$-approximate so that they can be
  applied recursively as in~\cite{Li21}
  over $O(\log n)$ many levels.
  This is an additional challenge that requires building on
  properties of tree-packings in the presence of a wide range of edge
  weights to, for example, find sources for local flow computations
  which identify minimum cuts that cross clusters.
\end{abstract}
\break
\pagenumbering{arabic}

\section{Introduction}

  The minimum cut problem of finding the partition in a graph with $n$ nodes and $m$ edges which
  minimizes the number (or weight) of edges between the parts has a
  long history and many applications. Algorithms were initially
  derived from algorithms for the minimum $s-t$ cut problem, which
  finds a partition that separates $s$ and $t$ and minimizes the edge
  number or weight between the pieces; here one can simply use $n-1$
  instances of the minimum $s-t$ cut to find a minimum cut. It is
  interesting to remind the reader that the minimum $s-t$ cut problem
  has as a dual, the maximum $s-t$ flow problem, whose flow value equals
  the minimum cut value, and algorithms are inextricably linked for
  both. 

  For minimum cut itself, a directed version of a dual was developed
  in 1962 by Nash-Williams \cite{nash1961edge}, who showed that a
  packing of arborescences could be found in a directed graph whose
  size is equal to the size of a related cut.  This packing figures in
  modern algorithms. In particular, the packing
  algorithm of Gabow \cite{gabo1995matroid} makes this duality algorithmically
  effective.  We also use it in our algorithm.
  Approximate versions using the multiplicative weights
  methods \cite{plotkin1995fast} of these ideas are directly
  useful in our work.  

  A linear-time deterministic method for producing  a ``sparse'' graph
  with $O(\lambda n)$ edges that preserves all minimum cuts for a graph
  with minimum cut size $\lambda$ was given by Nagamochi and Ibaraki
  \cite{NagamochiI92}.  
  Another structural breakthrough was given by Karger
  \cite{karger1993global} who showed that the remarkably simple
  randomized algorithm of repeatedly choosing a random edge to
  collapse, returns a minimum cut with probability $1/n^2$. This
  yields a proof that the number of minimum cuts is at most $O(n^2)$
  which is matched by the number of minimum cuts on a simple
  cycle. The method can also be used to bound the number of
  near-minimum cuts.  This led to a sequence of randomized algorithms
  \cite{karger1996new,benczur1996approximating,karger2000minimum} with
  improved running times over the deterministic maximum-flow based
  approaches.

  Furthermore, Karger \cite{karger2000minimum} gave a randomized
  algorithm that ran in time $O(m \log^3 n)$, which yielded a minimum cut
  with high probability. The algorithm makes use of the structural
  insights on the number of cuts in a graph, on tree packings, and 
  on analyzing sampling methods to yield a weighted graph with few
  edges that almost preserved all minimum cuts. This,
  along with a packing of trees and a clever dynamic program yielded
  Karger's bounds.

  For deterministic algorithms, the state of the art remained largely
  tied to maximum flow based methods until the breakthrough work of
  Kawarabayashi and Thorup \cite{KawarabayashiT19} who developed graph decomposition
  techniques to give a $\tilde{O}(m)$ deterministic algorithm for
  simple graphs, which was improved by Henzinger, Rao, and Wang
  \cite{DBLP:journals/siamcomp/HenzingerRW20} to $O(m \log^2 n \log\log n)$\footnote{While the running time is better than Karger's
randomized algorithm, subsequent work has provided an $O(m\log^2n)$
time randomized algorithm for minimum cut in simple graphs \cite{gawrychowski2020minimum}.}. Interestingly, the
  decomposition method's closest progenitor was perhaps the work of
  Spielman and Teng \cite{spielman2004nearly} who gave {\em randomized} near-linear
  time algorithms for solving linear systems.  Moreover, these graph
  decomposition methods are a central ingredient in the recent {\em
    randomized} almost-linear time algorithms for maximum flow
  \cite{chen2022maximum}.  We note that a central tool in \cite{KawarabayashiT19} was the page-rank
  algorithm which was replaced by local flow algorithms in \cite{DBLP:journals/siamcomp/HenzingerRW20}.
  The usefulness of decompositions in both randomized and deterministic
  settings is notable. 
  
  The main weakness of the method of \cite{KawarabayashiT19,DBLP:journals/siamcomp/HenzingerRW20} is that it is limited to simple
  graphs due to several issues, one of which is the fact that any
  minimum cut of size $\lambda$ must consist of either a single vertex
  or must have at least $\lambda$ vertices. When $\lambda$ is
  large\footnote{When $\lambda$ is small, the minimum cut can be found
  using \cite{gabo1995matroid}.}, this allows the decomposition procedure of
 \cite{KawarabayashiT19,DBLP:journals/siamcomp/HenzingerRW20} to find clusterings which are ``perfect''
  in the  sense that they don't ``cross'' any cut whose size is close to a minimum cuts. This is one of several obstacles that we must overcome.

  Li \cite{Li21} produced an $m^{1+o(1)}$ deterministic algorithm for
  weighted graphs using deterministic constructions of expander
  decompositions \cite{li2021deterministic,chuzhoy2020deterministic}
  to, in a sense, derandomize Karger's \cite{karger2000minimum} approach to
  minimum cut. His approach is limited by both the time for
  constructing hierarchical expander decompositions and the ``inefficiency''  (relative to
  \cite{KawarabayashiT19}) of the resulting decompositions, in the sense that the minimum cut in the corresponding collapsed graph might be a factor of $\Omega(\log n)$ larger per level in the hierarchical decomposition than the minimum cut. This leads to a super-polylogarithmic 
  multiplicative factor in the size of the minimum cut over all levels. Fortunately, through the framework of Karger's randomized algorithm, this multiplicative approximation can be paid for in the running time to compute a so-called \emph{skeleton graph}, from which an exact minimum cut can be computed. 

  In this paper, we give an $\tilde{O}(m)$ deterministic algorithm for
  minimum cut. While we combine the approaches in \cite{Li21} and
  those in
  \cite{KawarabayashiT19,DBLP:journals/siamcomp/HenzingerRW20}, we
  have to address a number of issues in each by devising more
  efficient algorithms for weighted graphs as well as deal with the
  more complicated structure of minimum cuts in weighted graphs as
  compared to simple graphs.
  
  Of independent interest, we provide a structural theorem that there
  exists a sparse clustering, i.e., a partition of the vertices into set such that there are few edges between the sets, that preserves minimum cuts in weighted
  graphs with $1/\polylog(n)$ multiplicative error.  This extends the theorem of \cite{KawarabayashiT19}
  which gives a perfect clustering for simple graphs and the theorem
  of \cite{Li21} which gives a clustering with a logarithmic error.  
  
  \subsection{Our techniques and related work}\label{subsec:overview}
  
We discuss in more detail some of the techniques that
we build upon along with a description of our contributions.

A classical result of Nash and Williams \cite{nash1961edge} showed that any graph
$G$ with minimum cut $\lambda$ contains a packing of $\lambda$
arborescences (directed trees) in the directed graph where each edge
in $G$ is replaced by an edge in each direction. Gabow \cite{gabo1995matroid}
provided an efficient $\tilde{O}(\lambda m)$ algorithm that
produced the packing. There are randomized approximation algorithms that
produce fractional and approximate packings as well
\cite{plotkin1995fast}.

As noted previously, in 1996, Karger's famous paper \cite{KargerStoc}
gave an $O(m \log^3 n)$ randomized min-cut algorithm that takes a graph $G$
and produces a weighted graph with $O(n\log n)$ edges where the
minimum cuts are preserved to within a $(1+\epsilon)$
factor.\footnote{This work is inspired by a sequence of previous
papers on randomized approaches that remained far from near linear
time. \cite{karger1996new,benczur1996approximating,benczur2015randomized}.}
Then Karger builds a tree packing using \cite{gabo1995matroid} in
that weighted graph that consists of $O(\log n)$ trees.  In such a
tree packing, there must be a tree which intersects the minimum-cut at
most twice, and Karger provides an efficient dynamic programming
algorithm that given that tree produces a minimum cut.

Karger produces the sparse graph by sampling edges inversely to their
``strength'', a measure of the connectivity of the endpoints, in a
manner that preserves all small cuts and yields a sparse
graph. Specifically, the sampling introduces an edge of weight $1/p$ if
an edge is sampled with probability $p$, which preserves the edge on
expectation, and sampling by strength preserves all cuts with high
probability, which can be argued using the bounds on the number of
minimum-cuts and near minimum-cuts as follows from
\cite{karger1996new}. This sampling procedure was very much
randomized and an efficient deterministic construction has remained
elusive.
We note that Karger's algorithm works naturally for weighted graphs.

The deterministic algorithms for simple graphs in \cite{KawarabayashiT19,DBLP:journals/siamcomp/HenzingerRW20} instead generate a sparse graph that preserves all non-trivial minimum cuts \footnote{A trivial cut consists of just one
vertex.} by producing a
hierarchical clustering. At each level, a set of clusters is produced
such that no non-trivial minimum-cut separates any cluster and has few intercluster
edges. Thus, collapsing the graph and recursing ensures that all
non-trivial minimum cuts are preserved. The collapsing can be repeated until the
graph has $m' = \tilde{O}(m/\delta)$ edges, where $\delta$ is the minimum degree, and then simply run Gabow's
algorithm \cite{gabo1995matroid} to find a minimum cut in $\tilde{O}(\lambda\cdot m/\delta)=\tilde{O}(m)$
time. 

An example is useful here. First, note that a simple cycle is ``easy''
in the sense that it is sparse and one can simply run Gabow
\cite{gabo1995matroid} in linear time.  Replacing each vertex with a
clique on $\delta$ vertices and connecting adjacent cliques along the cycle with matchings produces a ``dense'' graph with minimum degree $\delta + 1$ and minimum cut
size $\delta$. Note that the cliques are 
not split by any non-trivial minimum cut.
The algorithms in
\cite{KawarabayashiT19,DBLP:journals/siamcomp/HenzingerRW20} produce a
final graph such that each node is created by collapsing a clique.  The number of edges,
including parallel edges, in the collapsed graph is $O(m/\delta)$ and
this graph contains all non-trivial minimum cuts in the original
graph. Note that each collapsed graph induces a vertex clustering or partition, where each node in the collapsed graph induces a \emph{piece} in the partition. Inspired by this example, we say a vertex clustering is
\emph{non-crossing} if there exists a minimum cut that does not cut any
of the pieces. Thus, the algorithms in
\cite{KawarabayashiT19,DBLP:journals/siamcomp/HenzingerRW20} produce a non-crossing vertex clustering.

To proceed it is useful to provide a little more detail.  The
algorithms of
\cite{KawarabayashiT19,DBLP:journals/siamcomp/HenzingerRW20} produce
the non-crossing clusterings \emph{incrementally}. That is, their algorithm
begins with a partition where any minimum cut only overlaps (i.e.,
crosses) any piece by some volume $s$, which is initially $m/2$, and
then proceeds to iteratively refine that clustering, in each iteration
decreasing $s$ by a factor of $2$.  Once $s$ is smaller than
$\Theta(\delta)$, where $\delta$ is the minimum vertex degree, they clean up the resulting partition to obtain the
property that a non-trivial minimum cut does not cross any of the
resulting clusters at all. This last step of obtaining a 
non-crossing clustering relies crucially on a local expansion
property of simple graphs: that any minimum cut in a simple graph
consists of either a single vertex or at least $\delta$ vertices and,
thus, any non-trivial minimum-cut is extremely sparse for simple
graphs. However, this does not hold for weighted graphs: In a weighted graph a non-trivial minimum cut can contain
just two vertices for arbitrary values of $\delta$.

Moreover, the recursive techniques face substantial issues for weighted graphs
in terms of complexity. The refinement process in
\cite{DBLP:journals/siamcomp/HenzingerRW20} performs local flow
computations where the progress is proportional to the total weight of the 
edges that are touched.  In simple graphs, the weight and the number of
edges are the same and, thus, the total work is near-linear in the number
of edges.  For weighted graphs, the number of edges and the total
weight differ and we, thus, need to devise new algorithms for
the refinement process.  Another issue in the refinement process is
the allocation of sources for the local flows, which also is more
complex as weighted graphs do not ``expand'' with respect to weights
as fast as simple graphs expand with respect to degrees. E.g., a
degree $\delta$ simple graph has $\delta$-neighbors, where a vertex of
weighted degree of $\delta$ may have only a single neighbor. This is problematic for the local flow computation.

This expansion issue creates an even bigger obstruction for going
from ``almost'' non-crossing as defined by $s$ to completely
non-crossing.  This barrier remains and we do not produce a clustering
such that every cluster is non-crossing in this work.

Fortunately, Li \cite{Li21} provides an ingenious approach for
utilizing ``approximately non-crossing'' clusterings to derandomize the
sampling approach of Karger \cite{karger1996new}. He first shows how
to use expander decompositions to produce approximately non-crossing
clusterings\footnote{The technical definitions of approximately
non-crossing are a bit different than those in
\cite{KawarabayashiT19,DBLP:journals/siamcomp/HenzingerRW20} and here,
but are similar in consequence.}, the deterministic constructions of
which take $m^{1+o(1)}$ time. He introduced another measure of
non-crossing of cuts for the partition in addition to the volume $s$
of overlap. In particular, his clusterings ensuring that making a cut
consistent with a clustering would only increase the size of the cut by a multiplicative
$n^{o(1)}$ factor. The idea is to ``shift'' the cut by adding or removing a small set of vertices across the cut. This is what he calls \emph{uncrossing} a cut. This is how he
produces a hierarchical clustering where every minimum cut
is well approximated by an arithmetic expression using a small linear
number of the clusters in the hierarchy.

He argued further that sampling a subset of edges that approximately
preserves the boundary size of the near-linear number of clusters
would then preserve approximately the minimum cuts as well.  This
contrasts with Karger's sampling which argued about an exponential
number of cuts.  The linear number of cuts allows one to derandomize
the sampling as one can explicitly track whether all the cuts are properly 
represented as one chooses edges to include in the sample.  One method to sample edges
is to satisfy a weighted sum of the near-linear number of cuts using a multiplicative weights method to ensure that all cuts are properly represented. This method was introduced in \cite{raghavan1988probabilistic}.

The obstruction for Li~\cite{Li21} to obtaining 
a $\tilde O(m)$ running time is both in the use of the expander decomposition and the fact
that even a constant multiplicative factor increase in boundary size, much less a polylogarithmic,
induced by uncrossing cuts limits the number of recursive levels to a
constant. Indeed, a recursion of depth $k = O(\log n)$ would give a
disastrous $\log n^{O(\log n)}$ multiplicative increase in the size of the minimum cut, while fewer recursive levels
are insufficient to sufficiently sparsify the graph.  We note that the use of
expander decompositions requires at least a $\Omega(\log n)$ multiplicative increase of the cut size per level of the recursion without a fundamental new understanding of finding
sparsest cuts.

We address these problems by using a modified version of the techniques from \cite{KawarabayashiT19,DBLP:journals/siamcomp/HenzingerRW20}  which runs faster and is better in that the uncrossing of
minimum cuts only introduces an error of $(1+\epsilon)$ at each level
where $\epsilon = 1/\polylog(n)$. Thus, a subset of the edges (sampled deterministically as sketched above) that
approximates the cuts in the resulting decomposition yields an
accurate representation of all minimum cuts.

As before, we produce these clusterings by first recursively reducing the
overlap, $s$, of any minimum cut to have weighted volume  $\tilde{O}(\lambda)$. 
Producing clustering with overlap of $s$ needs modification,
as does understanding the cost of uncrossing at the detail level.  The
recursion needs significant changes due to the running time analysis
having to depend on the number of edges rather than their weight as mentioned above.  At
this point, cuts that do not cross any of the clusters 
can increase the size of the minimum cut by a polylogarithmic factor.

We then need to refine these clusters so that the size of the minimum cut increases only by a $(1+ \epsilon)$ factor.  To do so, we appeal to several
structural properties of minimum cuts to deal with a spectrum of edge
weights. For ``small'' clusters (with a polylogarithmic number of
vertices), we can use bounds on the number of minimum cuts along with
our flow procedures to produce a better clustering and charging
arguments like those in
\cite{KawarabayashiT19,DBLP:journals/siamcomp/HenzingerRW20}.  For
large clusters, we have the property that minimum cuts do not overlap
them very much which implies that overlapping minimum cuts contains
edges which have large weights. This allows us to pack few forests
into these large weight edges which, in turn, allows us to
appropriately select small sets of sources for the local flow
procedures that find problematic overlapping minimum cuts. In the end,
we prune off overlapping small cuts from a large cluster until the
large cluster is certified to have no bad overlapping mininum cuts.
We note that the localized flow methods are similar to those used in
randomized constructions of expander decompositions
\cite{DBLP:conf/soda/SaranurakW19} in addition to those in
\cite{DBLP:journals/siamcomp/HenzingerRW20}.



\section{Preliminaries}
We are given an undirected, weighted graph $G = (V,E, w)$. Let  $\lambda$ be the minimum cut in $G$ and let $W$ be the sum of the weights of all edges. We assume that $W$ is polynomial in $n$. If this is not the case, the stated running time increases by a factor of $O(\log W)$.
We use \emph{component} to denote a connected component of $G \setminus E'$ for some  $E' \subseteq E$. Given a component $C$ we use $\partial C$ to denote the set of edges called \emph{boundary edges} with exactly one endpoint in $C$.

\paragraph{Weights.} Given two non-overlapping vertex sets $A, B$, we use $w_G(A,B)$ (or $w(A,B)$ if the graph we refer to clear from the text) to denote the sum of the weight of the edges between $A$ and $B$
and we use $w_G(A)$ (or $w(A)$) to denote the sum of the weighted edge degrees of all vertices in $A$. Furthermore, we adapt the notations to weighted graphs by explicitly using superscript $W$ to denote the weighted quantities, e.g. $d(v)$ is the number of edges incident to $v$, $d^W(v)$ is the total weight of edges incident to $v$, while $\vol(V)=\sum_{v\in V} d(v)$, $\vol^W(V)=\sum_{v\in V} d^W(v)$ for any vertex set $V$.

\paragraph{Subgraphs.} We use $G[A]$ to denote the subgraph of $G$ induced by the vertices in $A$, while $G\{A\}$ denotes the graph $G[A]$ together with,  for every node $u\in A$ with $w_G(u,G\setminus A)>0$, an added self-loop of $u$ with weight $w_G(u,G\setminus A)$. We will use the subgraph $G\{A\}$ by default when we work on a component $A$, and the self-loops make $d^W(v)$ the same in $G\{A\}$ as in $G$ for all $v\in A$, so there is no ambiguity when we talk about the weighted degree or volume in a subgraph.

\begin{definition}[\textbf{Strength}]
\label{def:strength}
Let $s,\del \ge 0$. A  component $C$ in a graph $G$ is \emph{$s$-strong for $\del$} if every cut $S$ with cut-size $w(S,G\setminus S)\leq \cutsize$ satisfies $\min\left(\vol^W_{G\{C\}}(S\cap C),\vol^W_{G\{C\}}(C\setminus S) \right)\leq s$. We call $s$ the \emph{strength} of a $s$-strong component. Note that  $s$-strong is always defined with respect to some parameter $\del$ even though we frequently drop $\del$ and just use $s$-strong when it is clear from the context.
\end{definition}

\begin{fact}\label{fact:strong}
    Let $s,\del \ge 0$.
    Any component that is a subset of an $s$-strong component is $s$-strong, where both use the same $\del$ in the definition of $s$-strong.
\end{fact}

\begin{definition}\label{def:boundary-sparse}
For a component $A\subset V$ and parameter $\beta\le1$, a set $U\subseteq A$ is $\beta$-boundary-sparse in $A$ iff
\[ w(U,A\setminus U) < \beta\min\{w(U,V\setminus A),w(A\setminus U,V\setminus A)\} .\]
Otherwise it is non-$\beta$-boundary-sparse in $A$.
\end{definition}
Note that if $U$ is $\beta$-boundary-sparse in $A$, then $A \setminus U$ is also $\beta$-boundary-sparse in $A$, and that both $\emptyset$ and
$A$ are non-$\beta$-boundary-sparse in $A$.


\paragraph{Crossing.}
Given two vertex sets $A,B$, we say that $A$ and $B$ \emph{cross} if the sets $A\cap B$, $A\cap(V\setminus B)$, $(V\setminus A)\cap B$, and $(V\setminus A)\cap(V\setminus B)$  are all non-empty. Informally, when we say we wish to \emph{uncross} a set $A$ with a set $B$, we mean replacing $A$ with another set $A'$ that does not cross $B$.

\section{Overview of the algorithm}
At a high level, we follow the approach of the $O(|E|^{1+o(1)})$ deterministic minimum cut algorithm by Li~\cite{Li21}, which first performs a hierarchical expander decomposition of the input graph $G$, and uses the structure uncovered by the decomposition to derandomize the sampling procedure from Karger's classical randomized near-linear time minimum cut algorithm~\cite{karger2000minimum} to construct a skeleton graph. Informally speaking, the skeleton graph approximately preserves the minimum cut from the original graph while being much sparser. Finally, one can use the (deterministic) tree packing and dynamic programming techniques from Karger's result to find the minimum cut. While all other steps take time $\tilde{O}(|E|)$,
the bottleneck of Li's algorithm is the deterministic expander decomposition  and we replace it with a customized graph decomposition inspired by ideas from minimum cut algorithms for unweighted graphs~\cite{KawarabayashiT19,DBLP:journals/siamcomp/HenzingerRW20}. Although our decomposition doesn't provide all the guarantees of an expander decomposition, it takes time $\tilde{O}(|E|)$ and we show that it captures all the essential properties required by the subsequent steps in Li's algorithm. This improves the time complexity of the whole algorithm to be nearly linear. 
Concretely, the goal of the decomposition is to partition the vertices into subsets, such that there exists a cut of size $(1 + \epsilon) \lambda$
that does not cross any of the subsets, where $\epsilon = 1/\polylog|V|$. In the rest of this section, we discuss our graph decomposition routines and outline the main steps of our algorithm. The remaining sections then fill in the details for these steps.

First we get a close approximation of $\tilde \lambda$ of $\lambda$ (i.e.\ $\lambda \le \tilde\lambda \le 1.01 \lambda$) by guessing the value of $\lambda$ and then start decomposing the graph using three different procedures.
The first decomposition routine is similar to \cite{KawarabayashiT19,DBLP:journals/siamcomp/HenzingerRW20} and decomposes the graph $G'$ into \emph{$s_0$-strong} components for parameters $\del = \tilde\lambda/1.01$ and  $s_0=O(\del \tau^2)$
with $\tau = \polylog|V|$. 
We intuitively think of $s_0$-strong components as having similar guarantees for near-minimum cuts as expanders in the expander decomposition.
 Our decomposition is summarized in the following lemma, whose running time is near-linear in the unweighted number of edges in the graph.
 
 \begin{restatable*}[$s_0$-strong partition]{lemma}{StrongPartition}\label{lem:strong-partition}
Given a weighted graph $G=(V,E,w)$, a parameter $\del$ such that $\del\leq \min_{v\in V}d^W(v)$ and a parameter $s_0=\Theta(\del\log^c |V|)$ for some $c\geq 2$, there is an algorithm that runs in $\tilde O(|E|)$ time and partitions the vertex set $V$ into $s_0$-strong components for $\del$ such that the total weight of the edges between distinct components is at most $O(\frac{\sqrt{\del}\log|V|}{\sqrt {s_0}})\cdot\vol^W(V)$. 
\end{restatable*}
To prove this lemma we need to extend the techniques from \cite{DBLP:journals/siamcomp/HenzingerRW20} to weighted graphs, which poses the challenges discussed in Subsection~\ref{subsec:overview}. Sections \ref{sec:unit-flow} and~\ref{sec:s0} shows how we overcome them.

In the algorithm we choose $\del = \tilde\lambda/1.01$. 
In the following we always use 
$$s_0 = 10^{11}\tilde\delta\tau^2 = 10^{11} \tilde\lambda \tau^2/1.01$$
 and we use $\tau$ to denote
 $\Theta( \log^{c/2} |V|)$  for some $c\geq 2$.
We will call an $s_0$-strong component a \emph{cluster}.
To emphasize that a cluster is $s_0$-strong, we will often write $s_0$-strong cluster, even though this is redundant. We then decompose these clusters further.
By Fact~\ref{fact:strong} every refinement of a cluster is a cluster so that the two subsequent decompositions are only applied to clusters.
These decompositions are novel and constitute our main technical contributions. They are presented in Section~\ref{sec:uncrossing}.
The first decomposes each  cluster into at most one large-volume cluster and a collection of small-volume clusters. It is summarized in the following lemma, which is proven in Subsection~\ref{sec:bigpieces}.

\begin{restatable*}[Large cluster decomposition]{lemma}{FinalCuts}\label{lem:final-cuts}
Let $A\subseteq V$ be a  cluster, let $0<\epsilon\le0.1$ be a parameter and let $\tilde\lambda\le1.01\lambda$ be an approximate lower bound to the minimum cut known to the algorithm. There is an algorithm that partitions $A$ into a (potentially ``large'') set $A_0$ and ``small'' sets $A_1,\ldots,A_r$ such that
 \begin{enumerate}
 \item For each $i\in\{1,2,\ldots,r\}$, we have $\textbf{\textup{vol}}^W_G(A_i)\le (\frac{10 s_0}{\epsilon \tilde \lambda})\cdot s_0$.\label{item:final-cuts-1}
 \item The total weight $w(A_0,A_1,\ldots,A_r)$ of inter-cluster edges is at most $O(\epsilon^{-1})\cdot\partial A$.\label{item:final-cuts-2}
 \item For any set $U\subseteq A$ with $\partial_{G[A]}U\le1.01\tilde\lambda$ and $\textbf{\textup{vol}}^W_G(U)\le s_0$, there exists $U^*\subseteq U$ with $U^*\cap A_0=\emptyset$ and $w(U\setminus U^*,V\setminus A)+\partial_{G[A]}U^*\le(1+\epsilon)\partial_{G[A]}U$.\label{item:final-cuts-3} 
 \end{enumerate}
The algorithm runs in time $\tilde O((\frac{s_0}{\epsilon\tilde\lambda})^{O(1)}|E(G[A])|)$.
\end{restatable*}
By our choice of $s_0$ the running time is $\tilde O(|E(G[A])|)$.

By Fact~\ref{fact:strong} every set $A_i$ is $s_0$-strong.
We now call a cluster  \emph{small} if it has volume $(\frac{10 s_0}{\epsilon \tilde \lambda})\cdot s_0$ and \emph{large} otherwise. As we will choose $\epsilon = O(1/\log^{1.1} |V|)$, a small cluster has volume $  O(\del \tau^4 \log^{1.1} |V|) = \tilde O(\lambda)$.
Since every vertex has degree at least $\lambda$, it follows that small clusters contain only $\tilde O(1)$ many vertices.
Note that the lemma states that the potentially large cluster $A_0$ has been uncrossed, i.e., every cut $U$ of size at most $1.01 \tilde \lambda$ can be shifted to a cut $U^*$ such that $U^*$ does not cross $A_0$ and such that the size of $U^*$ is only a factor of $(1 + \epsilon)$ larger than the size of $U$. 

Thus we are left with potentially further decomposing the small clusters, for which we use the second decomposition. It runs in time polynomial in the number of vertices of the cluster, which is why we need the clusters to have only $\polylog|V|$ vertices. We will not define the concept of \emph{$(1+\epsilon)$-uncrossable} in property~(1) in this section, but informally, for two sets $S$ and $A$, we say that $S\cap A$ is $(1+\epsilon)$-uncrossable with $A$ if either $\partial(S\setminus A)\approx\partial S$ or $\partial(S\cup A)\approx\partial S$. In other words, we can ``uncross'' $S$ with $A$ at little additional cost to the size of the boundary.
\begin{restatable*}[Small cluster decomposition]{lemma}{SmallCuts}\label{lem:small-clusters}
Let $A\subseteq V$ be a cluster, let $0<\epsilon\le0.01$ be a parameter, and let $\tilde\lambda\le1.01\lambda$ be an approximate lower bound to the minimum cut known to the algorithm. There is an algorithm that partitions $A$ into a disjoint union of clusters $A_1\cup A_2\cup\cdots\cup A_k=A$ such that
 \begin{enumerate}
 \item For any set $S\subseteq A$ of $G$ with $\partial_{G[A]}S\le1.01\tilde\lambda$, there is a partition $\mathcal P$ of the set $\{A_1,A_2,\ldots,A_k\}$ such that for each part $P=\{A_{i_1},\ldots,A_{i_\ell}\}\in\mathcal P$, the set $S\cap\bigcup_{A'\in P}A'$ is non-$(1-\epsilon)$-boundary-sparse in $\bigcup_{A'\in P}A'$. 
 \item The sum of boundaries $\sum_i\partial A_i$ is at most $O(\epsilon^{-3}(\log|A|)^{O(1)}\partial A)$.
 \end{enumerate}
 The running time of the algorithm is $O(|A|^9 \polylog|A|)$.
\end{restatable*}
Note that if $A$ is a small cluster, the running time is $\tilde O(1)$.

With these decomposition routines, we now outline the algorithm for a weighted graph $G=(V,E)$. The goal of Step $1$ and the outer loop of Step \ref{step:2}  is to find a good enough estimation $\tilde{\lambda}$ of the minimum cut value $\lambda$. Step 1 find a value $\lambda_0$ with $\lambda \le \lambda_0 \le 3 \lambda$.
As Step \ref{step:2} iterates over all powers of 1.01 between $\lambda_0/3$ and $1.01 \lambda_0$, the iteration  is executed $O(1)$ times and $\tilde {\lambda}$ can be seen as a ``guess'' of $\lambda$
that will  at least once assume a value in $[\lambda,1.01\lambda]$. If $\tilde \lambda$ does not belong to $[\lambda,1.01\lambda]$, an arbitrary cut will be returned, but as the algorithm returns at the end the smallest over all cuts found for all guesses, it suffices that the algorithm finds a correct minimum cut in the iteration when $\tilde\lambda \in [\lambda,1.01\lambda]$.
Our  graph decomposition procedures are featured in Step~\ref{item:decomposition}, and the subsequent Steps~\ref{step:skel} and \ref{step:karger} are mostly the same as the original $O(|E|^{1+o(1)})$ algorithm by Li~\cite{Li21} and a presented in Section~\ref{sec:sparsifier}.
\begin{enumerate}
\item We first run Matula's $(2+\epsilon)$-approximation algorithm~\cite{DBLP:conf/soda/Matula93} (which can be easily extended to the weighted setting) to get a value $\lambda_0$ such that $\lambda \le \lambda_0 \le 3\lambda$. 
\item \label{step:2} For all powers $\tilde \lambda$ of $1.01$ between $\lambda_0/3$ and $1.01 \lambda_0$\label{item:guess-lambda}
\begin{enumerate}
    \item Initialize $j \gets 0$ and $G_0 \gets G$ with vertex set $V_0$  and edge set $ E_0$
    \item While there are at least 2 nodes in $G_j=(V_j,E_j)$: \label{item:decomposition}
    \begin{enumerate}

    \item Partition $V_j$ into subsets called clusters that are $s_0$-strong using the algorithm of \Cref{lem:strong-partition}, for parameters $\del=\tilde\lambda/1.01$ and $s_0=10^{11}\del\tau^2$. Let $\mathcal C^{s_0}_j$ be the set of $s_0$-strong clusters.\label{item:s0-strong} 
    \item Let $\epsilon= 1/\log^{1.1}|V|$.
    \item For each resulting cluster $C$, we further decompose $C$ according to \Cref{lem:final-cuts} with values $\tilde\lambda$ and $\epsilon$. Note that all clusters except possibly $A_0$ (if it exists) is a small cluster.\label{item:first-cluster}
    \item For each small cluster, we further decompose it according to \Cref{lem:small-clusters} with values $\tilde\lambda$ and $\epsilon$.\label{item:cluster}
    \item Let ${\cal C}_j = \{ C_{j,1}, C_{j,2}, \dots \}$ be the resulting set of clusters. Collapse every cluster $C \in {\cal C}_j$ with $|C| > 1$ into one node  and call the resulting graph $G_{j+1} = (V_{j+1}, E_{j+1})$.\label{item:contraction}
    \item $j \gets j + 1$
    \end{enumerate}
    \item Build a skeleton graph $H$ using the clusters $\{{\cal C}_j\}_j$ using the algorithm from \Cref{sec:sparsifier}.\label{step:skel}
    \item Run Karger's tree packing and dynamic programming algorithm on the skeleton $H$ to find a cut $S$. \label{step:karger}
\end{enumerate}
\item Output the cut $S$ with the minimum value found over all guesses of $\tilde\lambda$.
\end{enumerate}

 Note that every iteration of Step~\ref{item:decomposition} takes time  $\tilde O(m)$.
We will show in \Cref{lem:progress} that the clustering ${\cal C}_j$ guarantees that $\vol^W_{G_{j+1}}(V_{j+1}) \le \vol^W_{G_j}(V_j)/2$ for every $j \ge 0$, so there are $O(\log\vol^W_G(V))$ loop iterations of total time $\tilde O(m)$. 

For correctness 
let us first assume that $\tilde \lambda$ is between $\lambda$ and $1.01\lambda$.
In that case we will show, for each $j$, that after Step~\ref{item:cluster}, every approximately minimum cut can be ``shifted'' (or ``uncrossed'') into another cut that does not intersect any cluster and whose cut value grows at most by factor $(1 + o(1/\log |V|))$. In particular, collapsing such a cluster increases the size of the minimum cut by at most a factor of $(1 + o(1/\log |V|))$, which implies that the minimum cut of $G_{j+1}$ is at most a factor of $(1 + o(1/\log |V|))$ larger than the minimum cut of $G_j$. 
Thus, throughout all $O(\log |V|)$ iterations of Step~\ref{item:decomposition} the minimum cut  in $G_j$ is at most $1.01 \tilde\lambda$, which is a property needed for the decomposition algorithms in  steps\ref{item:first-cluster} and \ref{item:cluster} to correctly compute the clustering ${\cal C}_j$.
Using the guarantees of the clusters ${\cal C}_j$, the algorithm then constructs a skeleton graph $H$ following the pessimistic estimator approach of \cite{Li21}. Finally, the algorithm runs (the remainder of) Karger's near-linear time minimum cut algorithm on the skeleton $H$, which is deterministic and returns a minimum cut of $G$.


The lemmata below summarize the two key measures of progress. First, the total volume halves on each iteration $j$. Second, any near-minimum cut $S$ can be shifted by $O(\tau^2) = \polylog(n)$ vertices into a near-minimum cut $\widetilde S$ that does not cross any cluster in the partition, and, hence, survives the contraction on step~\ref{item:contraction}.

\begin{lemma}[Progress per iteration]\label{lem:progress}
For each iteration $j$, the total volume drops by a factor of at least $2$ in each iteration: $\vol^W_{G_{j+1}}(V_{j+1}) \le \vol^W_{G_j}(V_j)/2$.
\end{lemma}
\begin{proof}
The value of $\del$ used in Step~\ref{item:s0-strong} fulfills the condition of  by \Cref{lem:strong-partition}, and, thus, with $c=6$ the lemma implies that the total weight of edges between clusters is at most $O(\frac{\sqrt{\del}\log |V|}{\sqrt{s_0}})\cdot\textbf{\textup{vol}}^W_{G'_j}(V_j)$
for $s_0 = \Theta(\del \tau^2)$ with $\tau = \Theta(\log^3 |V|)$.
In particular, the sum of the boundaries $\sum_A\partial A$ is at most $O(\frac{\sqrt{\del}\log |V|}{\sqrt{s_0}})\cdot\textbf{\textup{vol}}^W_{G'_j}(V_j)$. 

For each such cluster $A$, by \Cref{lem:final-cuts}, Step~\ref{item:first-cluster} produces a partition into sub-clusters $A_0,A_1,\ldots,A_r$ such that $\sum_i\partial A_i\le O(\epsilon^{-1})\cdot\partial A$. 
Therefore, after Step~\ref{item:first-cluster}, the new sum of boundaries of all clusters is $O(\frac{\sqrt{\del}\log |V|}{\epsilon \sqrt{s_0}})\cdot\textbf{\textup{vol}}^W_{G'_j}(V_j)$. 
Finally on Step~\ref{item:cluster}, by 
\Cref{lem:small-clusters}, each small cluster $A$ is divided one last time into clusters whose sum of boundaries is at most $O(\epsilon^{-3}(\log|A|)^{O(1)}\partial A)$, bringing the new total sum of boundaries to $O(\frac{\sqrt{\del} \log |V| (\log \log |V|)^{O(1)}}{\epsilon^4 \sqrt{s_0}})\cdot\textbf{\textup{vol}}^W_{G'_j}(V_j)$. 
Putting everything together, we conclude that the total weight of inter-cluster edges is at most
\[ O\left(\frac{\sqrt{\del}  \log |V| (\log \log |V|)^{O(1)}}{\epsilon^4 \sqrt{s_0}}\right) \cdot \vol^W_{G_j}(V_j). \]
Recall that $s_0=10^{11}\del \tau^2$, $\tau = O(\log^3 |V|)$. Using  a sufficiently large constant factor for $\tau$, this bound is at most $\vol^W_{G_j}(V_j)/2$, as promised.
\end{proof}

\begin{lemma}[Structure lemma]\label{lem:structure}
Suppose that Step~\ref{item:guess-lambda} is run with $\tilde\lambda \in [\lambda,1.01\lambda]$, and fix an iteration $j$. For any cut $S\subseteq V_j$ with $\partial_{G_j}S\le1.01\lambda$, there exists a set $S'\subseteq V_j$ such that
  \begin{enumerate}
  \item $S'$ does not cross any cluster in $\mathcal C_j$,
  \item $\partial_{G_j}S'\le(1+3\epsilon)\partial_{G_j}S$, and
  \item $\textbf{\textup{vol}}^W_{G_j}(S\triangle S')\le O(s_0^2/(\epsilon\lambda))$.\label{item:structure-1c}
  \end{enumerate}
\end{lemma}

For the remainder of this section, we prove \Cref{lem:structure}.

Since $\partial_{G_j}S\le 1.01\lambda \le 1.01 \tilde\lambda \le \cutsize$ and each cluster $A\in\mathcal C^{s_0}_j$ is $s_0$-strong for $\del$, we have $\min\{\textbf{\textup{vol}}^W_{G_j}(S\cap A),\textbf{\textup{vol}}^W_{G_j}(A\setminus S)\}\le s_0$ for all clusters $A\in\mathcal C^{s_0}_j$. For each such cluster $A$ that crosses $S$, apply property~(\ref{item:final-cuts-3}) of \Cref{lem:final-cuts} on either $U=S\cap A$ or $U=A\setminus S$, whichever has smaller volume. 
By assumption, $\partial_{G_j[A]}U\le1.01\lambda\le1.01\tilde\lambda$ and $\textbf{\textup{vol}}^W_{G_j}(U)\le s_0$, so property~(\ref{item:final-cuts-3}) of \Cref{lem:final-cuts} guarantees a set $U^*$ with $U^*\cap A_0=\emptyset$ and $w_{G_j}(U\setminus U^*,V_j\setminus A)+\partial_{G_j[A]}U^*\le(1+\epsilon)\partial_{G_j[A]}U$. To distinguish the sets among different $A$, we add a subscript $A$ to $U,U^*$.

We now shift $S$ into a new set $S^*$ by, intuitively, replacing each of the above $U_A$ with $U^*_A$. Some care is needed to distinguish the cases $U_A=S\cap A$ and $U_A=A\setminus S$. If $U_A=S\cap A$, then we remove the vertices of $U_A\setminus U^*_A$ from $S$. If $U_A=A\setminus S$, then we add the vertices of $U_A\setminus U^*_A$ to $S$. By construction, we only add or remove vertices in $U_A$ of total volume at most $\textbf{\textup{vol}}_{G_j}^W(U_A)\le s_0$.

\begin{claim}\label{clm:S-star-size}
$\partial_{G_j}S^*\le(1+\epsilon)\partial_{G_j}S$.
\end{claim}
\begin{proof}
Consider a cluster $A\in\mathcal C^{s_0}_j$ with $U_A=S\cap A$. If $U_A\ne U^*_A$, then removing vertices of $U_A\setminus U^*_A$ results in a subset of edges of $E_{G_j}(U\setminus U^*,V_j\setminus A)$ being added to $\partial_{G_j}S^*$, and within cluster $A$, the net difference of edges is exactly $\partial_{G_j[A]}S^*-\partial_{G_j[A]}S$. The total net difference of edges from processing cluster $A$ is at most $w_{G_j}(U\setminus U^*,V_j\setminus A)+\partial_{G_j[A]}S^*-\partial_{G_j[A]}S$, which is at most $\epsilon\partial_{G_j[A]}S$ by property~(\ref{item:final-cuts-3}) of \Cref{lem:small-clusters}. A symmetric argument holds for the case $U_A=A\setminus S$. Finally, over all clusters $A$, the net differences sum to at most $\epsilon\partial_{G_j}S$.
\end{proof}

\begin{claim}\label{clm:in-cluster-at-least}
There are $O(1)$ many clusters for which $U^*_A\ne U_A$.
\end{claim}
\begin{proof}
Consider the value of the cut $\partial_{G_j}(U_A\setminus U_A^*) \ge \lambda_{G_j} \ge \lambda$. We can upper bound $\partial_{G_j}(U_A\setminus U_A^*)$ as
\begin{align*}
\partial_{G_j}(U_A\setminus U_A^*)&=w_{G_j}(U_A\setminus U_A^*,V_j\setminus A)+\partial_{G_j[A]}(U_A\setminus U_A^*)
\\&\le w_{G_j}(U_A\setminus U_A^*,V_j\setminus A)+\partial_{G_j[A]}U_A+\partial_{G_j[A]}U_A^*
\\&\le(2+\epsilon)\partial_{G_j[A]}U_A,
\end{align*}
where the last inequality is by property~(\ref{item:final-cuts-3}) of \Cref{lem:small-clusters}.
It follows that $\partial_{G_j[A]}U_A\ge\frac1{2+\epsilon}\partial_{G_j}(U_A\setminus U_A^*)\ge\frac1{2+\epsilon}\lambda$. Finally, the edges $\partial_{G_j[A]}U_{A}$ are disjoint among different clusters $A$, so there are at most $\partial_{G_j}S/(\frac1{2+\epsilon}\lambda)=O(1)$ many clusters for which $U^*_A\ne U_A$. 
\end{proof}

\begin{claim}
$\textbf{\textup{vol}}_{G_j}^W(S\triangle S^*)\le O(s_0)$.
\end{claim}
\begin{proof}
For each cluster $A$, we add or remove a subset of $U_A$, which has volume at most $\textbf{\textup{vol}}_{G_j}^W(U_A)\le s_0$. By \Cref{clm:in-cluster-at-least}, there are at most $O(1)$ many clusters for which $U^*_A\ne U_A$. In each such cluster, we add/remove vertices with total volume at most $s_0$. For the clusters $A$ with $U^*_A=U_A$, no action is needed. Altogether, we obtain $\textbf{\textup{vol}}_{G_j}^W(S\triangle S^*)\le O(1)\cdot s_0$.
\end{proof}

After step~(\ref{item:first-cluster}), by construction of $S^*$ and \Cref{lem:final-cuts}, the set $S^*$ only crosses small clusters, which have volume $\textbf{\textup{vol}}_{G_j}^W(A)\le O(s_0^2/(\epsilon\lambda))$. Each of these clusters is then processed according to \Cref{lem:small-clusters} in Step \ref{item:cluster} as follows.  
We shift $S^*$ into a new set $S'$ by, intuitively, uncrossing with each cluster according to \Cref{lem:small-clusters}. For each cluster $A$ that crosses $S^*$, let $\mathcal P$ be the partition guaranteed by property~(1) of~\Cref{lem:small-clusters}. Recall from \Cref{def:boundary-sparse} that a set $U$ is $\beta$-boundary-sparse in $A$ if $w(U,A\setminus U) < \beta\min\{w(U,V_j\setminus A),w(A\setminus U,V_j\setminus A)\}$.
For each $P\in\mathcal P$, define $A_P=\bigcup_{A'\in P}A'$ to avoid clutter. If $S^*$ crosses $A_P$, then the set $U_P=S^*\cap A_P$ is not $(1-\epsilon)$-boundary-sparse in $A_P$, so either $w(U_P,A_P\setminus U_P)\ge(1-\epsilon)w(U_P,V_j\setminus A_P)$ or $w(U_P,A_P\setminus U_P)\ge(1-\epsilon)w(A_P\setminus U,V_j\setminus A_P)$. If $w(U_P,A_P\setminus U_P)\ge(1-\epsilon)w(U_P,V_j\setminus A_P)$, then update $S^*\gets S^*\triangle U_P$, i.e., we remove the vertices of $U_P$ from $S^*$. 
If $w(U_P,A_P\setminus U_P)\ge(1-\epsilon)w(A_P\setminus U_P,V_j\setminus A_P)$, then update $S^*\gets S^*\triangle(A_P\setminus U_P)$,  i.e., we add the vertices of $A_P \setminus U_P$ to $S^*$. 

By construction, we only add or remove vertices in $A$ of total volume at most $\textbf{\textup{vol}}_{G_j}^W(A)\le O(\frac{s_0}{\epsilon\tilde\lambda})\cdot s_0$.

\begin{claim}\label{clm:S-prime-size}
$\partial_{G_j}S'\le\frac1{1-\epsilon}\partial_{G_j}S^*$.
\end{claim}
\begin{proof}
Consider a cluster $A$ and its partition $\mathcal P$ guaranteed by \Cref{lem:final-cuts}. 
If $w(U_P,A_P\setminus U_P)\ge(1-\epsilon)w(U_P,V_j\setminus A_P)$, then the update $S^*\gets S^*\triangle U_P$ results in a subset of edges of $E_{G_j}(U_P,V_j\setminus A)$ being added to $\partial_{G_j}S^*$, and within $A_P$, all edges of $E_{G_j}(U_P,A_P\setminus U_P)$ are no longer in $\partial_{G_j}S^*$. So the net increase of $\partial_{G_j}S^*$ is at most $w(U_P,V_j\setminus A_P)-w(U_P,A_P\setminus U_P)\le\frac1{1-\epsilon}w(U_P,A_P\setminus U_P)-w(U_P,A_P\setminus U_P)\le(\frac1{1-\epsilon}-1)w(U_P,A_P\setminus U_P)$. 
The argument for the update
$S^*\gets S^*\triangle (A_p \setminus U_P)$ is symmetric and has the same upper bound of $(\frac1{1-\epsilon}-1)w(U_P,A_P\setminus U_P)$ for the net increase of $\partial_{G_j}S^*$.

Summed over all clusters $A$ and parts $P\in\mathcal P$, the total net increase of $\partial_{G_j}S^*$ is at most $(\frac1{1-\epsilon}-1)\sum_{A,P}w(U_P,A_P\setminus U_P)$, which is at most $(\frac1{1-\epsilon}-1)\partial_{G_j}S^*$ since the edges in $E_{G_j}(U_P,A_P\setminus U_P)$ are disjoint over all $A,P$. It follows that $\partial_{G_j}S'\le\frac1{1-\epsilon}\partial_{G_j}S^*$.
\end{proof}

\begin{claim}\label{clm:S-prime-number}
$\textbf{\textup{vol}}_{G_j}^W(S'\triangle S^*)\le O(s_0^2/(\epsilon \lambda))$.
\end{claim}
\begin{proof}

Suppose that $S^*$ crosses some $A_P$ for a small cluster $A$ and a partition $\mathcal P$ returned by
Step \ref{item:cluster}. Recall that $U_P = S^* \cap A_P$ and that Lemma \ref{lem:small-clusters} guarantees that either $w(U_P,A_P\setminus U_P)\ge(1-\epsilon)w(U_P,V_j\setminus A_P)$ or $w(U_P,A_P\setminus U_P)\ge(1-\epsilon)w(A_P\setminus U,V_j\setminus A_P)$. In the first case, consider the value of the cut $\partial_{G_j}U_P\ge\lambda_{G_j}$. We can upper bound $\partial_{G_j}U_P$ as
\[ \partial_{G_j}U_P=w(U_P,V_j\setminus A_P)+w(U_P,A_P\setminus U_P)\le \frac1{1-\epsilon}w(U_P,A_P\setminus U_P)+w(U_P,A_P\setminus U_P)=\frac{2-\epsilon}{1-\epsilon}w(U_P,A_P\setminus U_P). \] It follows that $w(U_P,A_P\setminus U_P)\ge\frac{1-\epsilon}{2-\epsilon}\partial_{G_j}U_P\ge\frac{1-\epsilon}{2-\epsilon}\lambda_{G_j}\ge\frac{1-\epsilon}{2-\epsilon}\lambda$.

The edges $E(U_P,A_P\setminus U_P)$ are disjoint among pairs $(A,P)$, so there are at most $\partial_{G_j}S^*/(\frac{1-\epsilon}{2-\epsilon}\lambda)\le(1+\epsilon)\partial_{G_j}S/(\frac{1-\epsilon}{2-\epsilon}\lambda)=O(1)$ many pairs $(A,P)$ for which any action is needed. For each, we add or remove vertices of total volume at most $\textbf{\textup{vol}}_{G_j}^W(A)= O(s_0^2/(\epsilon\tilde\lambda))
O(s_0^2/(\epsilon\lambda))$, and the claim follows.
\end{proof}
We now verify the requirements of \Cref{lem:structure}:
 \begin{enumerate}
 \item By construction, $S'$ does not cross any cluster in $\mathcal C_j$.
 \item $\partial_{G_j}S'\le\frac1{1-\epsilon}\partial_{G_j}S^*\le\frac1{1-\epsilon}\cdot(1+\epsilon)\partial_{G_j}S\le(1+3\epsilon)\partial_{G_j}S$.
 \item $\textbf{\textup{vol}}^W_{G_j}(S\triangle S')\le\textbf{\textup{vol}}^W_{G_j}(S\triangle S^*)+\textbf{\textup{vol}}^W_{G_j}(S^*\triangle S')\le O(s_0) + O(s_0^2/(\epsilon\lambda))\le O(s_0^2/(\epsilon\lambda))$. 
 \end{enumerate}
This concludes the proof of \Cref{lem:structure}
\section{Unit-Flow Algorithm}
\label{sec:unit-flow}
We will consider flow problems extensively in our graph decomposition procedure, and we start with some notations. Formally, a {\em flow problem} is defined with a {\em source function}, $\Delta:V\rightarrow \mathbb{R}_{\geq 0}$, a {\em sink function}, $T:V\rightarrow \mathbb{R}_{\geq 0}$, and edge capacities $c(\cdot)$. We say that $v$ is {\em a sink of capacity} $x$ if $T(v)=x$. 
 To avoid confusion with the way flow is used, we use {\em supply} to refer to the substance being routed in flow problems.

For the sake of efficiency, we will not typically obtain a full solution to a flow problem. We will compute a \emph{pre-flow}, which is a function $f:V\times V \rightarrow \mathbb{R}$, where $f(u,v) = -f(v,u)$. A pre-flow $f$ is {\em source-feasible} with respect to source function $\Delta$ if $\forall v:\sum_u f(v,u) \leq \Delta(v)$. A pre-flow $f$ is {\em capacity-feasible} with respect to $c(\cdot)$ if $|f(u,v)| \leq c(e)$ for $e=\{u,v\} \in E$ and $f(u,v) = 0$ otherwise. We say that $f$ is a {\em  feasible pre-flow} for flow problem $\Pi$, or simply a {\em pre-flow} for $\Pi$, if $f$ is both source-feasible and capacity-feasible with respect to $\Pi$. 

For a pre-flow $f$ and a source function $\Delta(\cdot)$, we extend the notation to denote $f(v)\myeq\Delta(v)+\sum_u f(u,v)$ as {\em the amount of supply ending at $v$ after $f$}. Note that $f(v)$ is non-negative for all $v$ if $f$ is source-feasible. When we use a pre-flow as a function on vertices, we refer to the function $f(\cdot)$, and it will be clear from the context what $\Delta(\cdot)$ we are using. If in addition, $\forall v:f(v)\leq T(v)$, the pre-flow $f$ will be a {\em feasible flow (solution)} to the flow problem $\Pi$. 

We denote $\ex(v)\myeq\max(f(v)-T(v),0)$ as the excess supply at $v$, and we call the part of the supply below $T(v)$ as the supply routed to the sink at $v$, or {\em absorbed} by $v$. We call the sum of all the supply absorbed by vertices, $\sum_v \min(f(v),T(v))$, the total supply routed to sinks. Finally, given a source function $\Delta(\cdot)$, we define $|\Delta(\cdot)|\myeq\sum_v\Delta(v)$ as the total amount of supply in the flow problem. Note the total amount of supply is preserved by any pre-flow routing, so $\sum_v f(v) =|\Delta(\cdot)|$ for any source-feasible pre-flow $f$. 

In this section, we consider a particular family of flow problems that is useful in our graph decomposition, and discuss the algorithm we use to compute pre-flows. The flow problem we consider comes with a source function $\Delta(\cdot)$, a weighed undirected graph $G$ with edge weights $w_e$'s, an edge capacity parameter $U$ such that $c_e=w_e\cdot U$ for all edges, and $T(v)=d^W(v)$ for all vertices. As a side note, for the actual flow problems we solve in our graph decomposition procedure, all the quantities (i.e. source supply, edge and sink capacities) are in the same unit whose value can vary across different flow problems, but the flow algorithm should be completely oblivious to the value of the unit. Nonetheless, 
in the discuss of our flow computation we may refer to supply of amount $x$ as $x$ units of supply.  

To solve our flow problems, we use the {\em Unit-Flow} algorithm, which was designed in~\cite{DBLP:journals/siamcomp/HenzingerRW20} for unweighted graphs, and we adapt it to weighted graphs. The algorithm follows the classic Push-belabel framework by Goldberg and Tarjan~\cite{GT88}, which at a high level works as follows. The algorithm maintains an integer label (or height) $l(v)$ for each node $v$ starting at $0$. If there is excess on any node, the algorithm picks one such node and tries to send its excess to any neighbour of smaller label while obeying the edge's residual capacity (a.k.a  {\em Push}), and if there is no possible Push for a node with excess, the algorithm raises the label of the node by $1$ (a.k.a {\em Relabel}). The standard Push-relabel algorithm for max flow keeps carrying out Push or Relabel until it finds a flow (i.e. no excess supply on any node). 

The Unit-Flow algorithm works exactly the same except one difference. In particular, Unit-Flow takes as input a cap $h$ on the height, and when Push-relabel needs to raise a node's label from $h$ to $h+1$, Unit-Flow will instead leave the node at height $h$ without working further on the node's excess supply. Note this is also the only case that Unit-Flow leaves excess on nodes, so when the algorithm stops, if there is any node with excess, the node must have label $h$. Since the algorithm description and the running time analysis follow quite mechanically from the Push-relabel algorithm in~\cite{GT88} using the dynamic-tree data structuer~\cite{ST81}, in this section we only summarize the result of Unit-Flow that we will use in our later graph decomposition. The proof of the main theorem also follows largely from the original analysis of unweighted Unit-Flow in~\cite{DBLP:journals/siamcomp/HenzingerRW20}, so we defer most of the implementation details and analysis of our weighted Unit-Flow to the Appendix~\ref{app:unit-flow}.

Upon termination of the algorithm, it returns
$(1)$ a pre-flow $f$, $(2)$ labels $l(\cdot)$ (i.e. height) on vertices, and (3) a possibly empty cut $A$. More specifically, {\em Unit-Flow} will either successfully route a large amount of supply to sinks and return an empty cut, or it returns a low conductance cut using the labels. Formally, we have the following theorem.
\begin{restatable}{theorem}{unitFlow}
\label{thm:unit-flow}
Given $G,\Delta,h,U,\beta\geq 2$ such that $\Delta(v)\leq \beta d^W(v)$ for all $v$, {\em Unit-Flow} terminates with a pre-flow $f$, where we must have one of the following three cases
\begin{enumerate}[(1)]
\item $f$ is a feasible flow, i.e. $\forall v:\ex(v)=0$. All units of supply are absorbed by sinks.\label{item:unit-flow-1}
\item $f$ is not a feasible flow, but $\forall v:f(v)\geq d^W(v)$, i.e., at least $\vol^W(G)$ units of supply are absorbed by sinks.\label{item:unit-flow-2}
\item If $f$ satisfies neither of the two cases above, denote $T$ as the set of nodes whose labels are $h$ (i.e. largest possible label), and assume $h\geq \ln \vol^W(G)$, then we can find a cut $(A,\bar{A})$ such that \label{item:unit-flow-3}
\begin{enumerate}
\item $T\subseteq A$, and $(\beta + U)d^W(v)\geq f(v)\geq d^W(v)$ for all $v\in T$, and $f(v)= d^W(v)$ for all $v\in A\setminus T$, and $f(v)\leq d^W(v)$ for all $v\notin A$.
\item The cut's conductance $\Phi(A)=\frac{W(A,V\setminus A)}{\min\left(\vol^W(A),\vol^W(G)-\vol^W(A)\right)}\leq \frac{20\ln \vol^W(G)}{h}+\frac{\beta}{U}$.
\item Let $\ex(T)$ be the total amount of excess supply (which can only be on nodes in $T$), i.e. $\ex(T)=\sum_v \max\left(0,f(v)-d^W(v)\right)$, we have $\vol^W(A)\geq \frac{\ex(T)}{(\beta-1)+10U\ln \left(\vol^W(G)\right)/h}$
\end{enumerate}
\end{enumerate}
Let $C=\left\{v:f(v)\geq d^W(v)\right\}$ be the set of saturated nodes when Unit-Flow finishes, the running time for {\em Unit-Flow} is $O(\vol(C)\cdot h\cdot \log \vol(C))$. Moreover, the above guarantees hold in the case when $\Delta$ is incremental, i.e., after Unit-Flow finishes working with the supply from the current $\Delta(\cdot)$, new source supply can be added on nodes to increase $\Delta(\cdot)$, and the theorem holds when Unit-Flow finishes with any interim $\Delta(\cdot)$.
\end{restatable}
Note that the running time only depends on the \emph{unweighted} volume $\vol(C)$ and \emph{not} on the weighted volume $\vol^W(C)$. 





\section{Guaranteeing $s_0$-strength}\label{sec:s0}
Our procedure will make progress on the graph by breaking it into $s_0$-strong components, i.e., components such that if we consider any nearly min cut in the graph, any component will be mostly on one side of the cut up to some small error. 
We prove in this section the following result. 
\StrongPartition
Throughout this section, when we talk about $s$-strong components for any $s\geq s_0$, the definition of strength is with respect to the parameter $\del$, so we only worry about any cut $S$ with $w(S,G\setminus S)\leq \cutsize$. We denote $\tau=\Theta(\sqrt{s_0/\del})=\Theta(\ln^{c/2}|V|)$ as the parameter, such that in the section we always break components along cuts of conductance at most $O(1/\tau)$. Moreover, we need the following  definitions.
\begin{definition}
Given a  component $C$ in a graph $G$ we say that a cut $(A, C \setminus A)$ is \emph{relatively balanced} if $\min\left(\vol^W(A),\vol^W(C\setminus A)\right) \ge \Omega\left(\frac{\vol^W(C)}{\tau^2 \ln m_G}\right)$ where all volumes refer to the graph $G\{C\}$ (which are the same as in $G$ due to the added self-loops).
\end{definition}


We start with some input graph $G$ with $m_G$ edges and minimum weighted node degree at least $\del$. By definition $G$ is $\vol^W(G)/2$-strong. For a component $C$ that is already certified to be $s$-strong for some $s\geq s_0$, our algorithm will try to certify $C$ is $s/2$-strong, but during the process it may find a low conductance cut in $C$ and break up $C$. When we break a $s$-strong component $C$ into $(C',C\setminus C')$, by definition we already certify $C'$ is $\min(s,\vol^W(C')/2)$-strong (and similarly for $C\setminus C'$). The running time to certify that $C$ is
$s/2$-strong or to find a low conductance cut will be $\tilde{O}(m_C)$, where $m_C$ is the number of edges in $C$. We will iteratively work on the graph (by cutting or certifying) and stop when all connected components are certified to be $s_0$-strong where for concreteness we use the explicit value of $s_0 =
10^{11} \del\tau^2$. 

Moreover, for the low conductance cut we use to divide a component $C$, either the cut is relatively balanced (i.e., the weighted volume of both $C'$ and $C \setminus C'$ is at most $(1 - \frac{1}{\tau^2\ln m_G}) \vol^W(C)$), or we can certify that the larger side of the cut is $s/2$-strong. This guarantees that the total running time to decompose the graph $G$ into $s_0$-strong components is near-linear in the number of edges of the graph (i.e. $m_G$).


For now, we will work with the graph $G\{C\}$ of one connected component $C$ that is certified to be $s$-strong for some $s>s_0$. We use $m_C$, $W_C$ to denote the number of edges and weighted volume of $G\{C\}$ respectively. 



\subsection{Flow Problems and Excess Scaling}
We will construct a sequence of flow problems, and have a total of (up to) $L=O(\ln m_C)$ stages for each flow problem. The basic structure is that each flow problem will start at stage $0$ and be associated with a source node at stage $0$. 

Technically for each flow problem, we need its associated source node to reserve a distinct $\del$ adjacent edge weight for that flow problem, and thus we can (and may have to) use the same node $v$ as the source node of multiple flow problems if $d^W(v)\gg \del$ to accommodate a large enough number of flow problems in total. However, for the sake of notation simplicity, it is easier if we can associate each flow problem with a distinct source node (and thus refer to it by its source node), and to make sure the graph can still support enough flow problems, we assume each node $v$ has $\del\leq d^W(v)\leq 2\del$. The lower bound automatically holds, and for any node $v$ with $d^W(v)>2\del$, we can view it as several virtual nodes such that they (virtually) partition the total weight of $v$'s adjacent edges in a way that the weighted degree of each is between $\del$ and $2\del$. This guarantees that we have at least $\frac{W_C}{2\del}$ distinct (virtual) nodes to use as source nodes. Consequently, we WLOG make the following assumption.
\begin{assumption}
\label{assumption:enough-source}
We can pick up to $\frac{W_C}{2\del}$ distinct nodes in the graph as source nodes of our flow problems in stage $0$, and all source nodes have weighted degree at least $\del$.    
\end{assumption}


\paragraph{Stages. }For the flow problem associated with a node $v$, we will inject a certain amount of source supply at $v$ in stage $0$ and try to spread it to sinks iteratively in stages. In each stage we use Unit-Flow, which may terminate with some excess supply. If the total excess supply is at most a $O(1/\ln m_C)$ fraction of the total amount of source supply, we will call the flow problem \emph{successful for that stage} (otherwise \emph{unsuccessful for that stage}) and advance the flow problem into the next stage. We keep working on a flow problem until it is unsuccessful at some stage or we finish successfully through all $L$ stages. We then pick a new source node and work on the flow problem starting with the new node at stage $0$, and keep doing this until enough flow problems are successful at any particular stage. The flow problems happen sequentially, i.e. we work on a flow problem through stages $0,1,\ldots$ until it is unsuccessful at some stage, and then start a new flow problem. Moreover, every flow problem is associated with a distinct (stage-$0$) source node.
\paragraph{Batches. }We will group the sequence of flow problems into batches, where we keep adding flow problems into the current batch until (a) we find enough flow problems successfully finishing at stage L in  the current batch or (b) the total number of flow problems in the current batch reaches a limit. In case (a) we start a new batch, and in Case (b) we will show that we can find a relatively balanced low conductance cut. 

\paragraph{Capacity sharing. }We will run multiple flow problems of the same stage on the same copy of the graph, which means the pre-flows share the edge and sink capacities in the same graph. More specifically, 
\begin{itemize}
\item All the flow problems (over all batches) in stage $0$ will share a single copy of the graph (i.e., the edge and sink capacities). That is, the edge capacities and sink capacities are constraints over the total flow routed by \emph{all} the flow problems in stage $0$.
\item For any later stage $i\in\{1,\ldots,L\}$, all the flow problems \emph{of the same batch in stage $i$} will run on a single copy of the graph, while flow problems of different batches or of different stages have separate copies of the graph 
(i.e.~, do not share their edge and sink capacities). 
\end{itemize}

\paragraph{Flow problems.}
We now fill in more details. For a flow problem associated with source node $v$ in batch $i$, we denote its source supply function at stage $j$ as $\Delta_{i,j}^v(\cdot)$. At stage $0$, we have \tempfix{$\Delta_{i,0}^v(v')=10^7\del\ln m$}{$\Delta_{i,0}^v(v')=10^7\del\tau$}
 if $v'=v$ and $0$ otherwise (i.e. start in stage $0$ with source supply only at node $v$). As in the excess scaling algorithm in~\cite{DBLP:journals/siamcomp/HenzingerRW20}, \emph{all the quantities for a flow problem (i.e. source supply, sink capacity, and edge capacity) in a stage are specified in some unit}. For stage $i$, the unit is
\[
\tempfix{\mu_i = \lceil \frac{2s}{2^i10^7\del\ln m}\rceil}{\mu_i = \lceil \frac{2s}{2^i10^7\del\tau}\rceil},
\]
and in particular the total number of stages $L$ equals the smallest value of $i$ by which the unit becomes at most $1$. 
Since $s$ is at most $W_C/2$, we know $L\leq \ln W_C = O(\ln m_C)$ if edge weights are bounded by polynomials of $m$. 
For stage $j$, the edge capacity (i.e. total net supply routed over the edge by all flow problems sharing the same copy of the graph) of $e$ is $U_j\cdot w_e$ units, where $w_e$ is the edge weight of $e$. We use \tempfix{$U_0=10^8\tau\ln m$}{$U_0=10^8\tau^2$} and $U_j=100\tau$ for $j\geq 1$. The total sink capacity (shared across all flow problems using the same copy of the graph) of a node $v$ is always its weighted degree $d^W(v)$ units.


We consider all the flow problems sharing the same copy of the graph in aggregate as one flow problem, and denote the (aggregate) source supply function as $\Delta_{i,j}(\cdot)$ for the aggregated problem corresponding to batch $i$ and stage $j$; for stage $0$ since there is one copy of the graph over all batches, the function $\Delta_{i,0}(\cdot)$ is only for convenience in the analysis, and we denote the actual aggregate source supply function as $\Delta_{0}(\cdot)$. For $i,j$, the source function $\Delta_{i,j}(\cdot)$ is considered as incremental as we add more flow problems to the aggregate problem, i.e. we add $\Delta_{i,j}^v(\cdot)$ to $\Delta_{i,j}(\cdot)$ when a flow problem in batch $i$ associated with node $v$ reaches stage $j$,  and we solve each aggregate flow problem with one run of the Unit-Flow algorithm with incremental source supply. 
Because of how Unit-Flow works, if some supply introduced by $\Delta_{i,j}^v(\cdot)$ is absorbed in a sink during the execution of Unit-Flow, it will remain absorbed throughout the execution even with additional source supply introduced later. Similarly, if some supply is left as excess supply after Unit-Flow terminates with some interim $\Delta_{i,j}(\cdot)$, the excess must be at a node already at height $h$, and it will remain as excess even with additional source supply added later. In particular, this means whenever Unit-Flow finishes with the interim aggregate flow problem after we add some $\Delta_{i,j}^v(\cdot)$ to it, we can tell at that point whether the flow problem of $v$ is successful for stage $j$ or not, and we know which of the supply introduced in $\Delta_{i,j}^v(\cdot)$ will be absorbed or be excess (without even knowing the complete aggregate flow problem in future). 

The source supply function $\Delta_{i,j+1}^v(\cdot)$ is defined to be  double the amount of absorbed supply on each node, more specifically it is defined as follows:
Let $f_{i,j}^v(v')$ be the total amount of supply introduced by $\Delta_{i,j}^v(\cdot)$ that gets routed to any node $v'$ by Unit-Flow, and let $\ex_{i,j}^v(v')$ be the amount of excess (from that flow) routed to $v'$, then formally we set
\[
\Delta_{i,j+1}^v(v') = 2\left(f_{i,j}^v(v')-\ex_{i,j}^v(v')\right) \qquad \forall v'.
\]
Note that the excess also depends on the earlier flow problems using this copy of the graph since the sink capacities are shared. Recall that each flow problem starts with supply \tempfix{$10^7\del\ln m$}{$10^7\del\tau$} (in unit $\mu_0$) at stage $0$, and if all supply has been carried over into the next stage (i.e. no excess) since stage $0$, the amount of supply at stage $j$ 
 would increase to $2^j10^7\del\tempfix{\ln m}{\tau}$ (in unit $\mu_j$ since $\mu_0=2^j\mu_j$). We call a flow problem \emph{successful at stage $j$} if at the end of stage $j$, the amount of supply for the flow problem being absorbed by sinks, i.e. $\sum_{v'}f_{i,j}^v(v')$, is at least $(1-\frac{j+1}{10(L+1)})2^j10^7\del\tempfix{\ln m}{\tau}$. Equivalently, a flow problem reaching stage $j\geq 1$ loses at most a $\frac{j}{10(L+1)}=O(j/\ln m_C)$ fraction of the initial source supply as excess over all previous stages. The absorbed supply on nodes after removing excess after stage $j$ would serve as the source supply if we proceed the flow problem to the next stage $j+1$, and the amount would double (i.e. the factor of $2$ in the formula above for $\Delta_{i,j+1}^v(\cdot)$) as we switch to the unit $\mu_{j+1}=\mu_j/2$.
 
In our Unit Flow computation, we will use \tempfix{$h=\Theta(\tau\ln^2 m_G)$}{$h=\Theta(\tau^2\ln m_G)$} for stage $0$, and we put $10^7\del\tau$ units of source supply on the source node for every flow problem. As the source nodes are distinct for different flow problems and all nodes have weighted degree at least $\del$, we know the ratio $\Delta_0(v)/d^W(v)$  (i.e. $\beta$ in Theorem~\ref{thm:unit-flow}) is at most $O(\tau)$. For all stages $j\geq 1$, we will use $h=\Theta(\tau\ln m_G)$ and since we only carry absorbed supply into the next stage, we have $\beta\leq 2$ in later stages. Note the $2$ comes from the scaling of unit value from stage $j$ to $j+1$.
\paragraph{Termination of the Certification Procedure: Cut or Certification.} 
Whenever we finish the (incremental) Unit-Flow computation of a flow problem in a stage, we perform two checks as described below. If none of the checks passes, we proceed the flow problem to the next stage $j+1$ if it's successful for the current stage $j<L$, or start a new flow problem with a new sourve vertex at stage $0$ if the current flow problem is not successful or it finishes stage $j=L$.

\emph{Check $1$:} Recall that if a flow problem is successful after $L$ stages, it has lost at most a $\frac{L}{10(L+1)} \le \frac{1}{10}$ fraction of its initial source supply. Thus, it routes at least $0.9\cdot 10^7 \del \tau \cdot \mu_0 = 1.8 s$ source supply to sinks.
When we stop with a particular flow problem successfully at stage $L$,  we check if there are already $\frac{W_C}{10s}$
successful flow problems at stage $L$ in the current batch. 

If so, we call the current batch \emph{successful} and start a new batch of flow problems unless we already have \tempfix{$5000\ln m$}{$5000\tau$} successful batches. If we succeed in constructing \tempfix{$5000\ln m$}{$5000\tau$} successful batches, we can certify that $C$ (or a constant fraction of it) is $s'$-strong with $s'=s/2$ (see the proof in Section~\ref{sec:certification}). 


\emph{Check $2$:} 
If Check 1 fails, 
 we  check if the result of Unit-Flow for the (aggregate) flow problem of any stage $j\geq 0$ in the current batch  has more than $\frac{W_C}{10^6\tau(L+1)}$ excess supply in total. 
As soon as this happens, we call this batch \emph{unsuccessful} and we will show in Section~\ref{sec:cut} that we can find a \emph{relatively balanced} low conductance cut in component $C$, i.e., both sides have weighted volume $\Omega(\frac{W_c}{\tau^2\ln m_G})$. Thus, we stop the certification procedure,
divide $C$ along the cut, leave the cut edges as inter-cluster edges, and continue the certification procedure on each of the two components, which are both $s$-strong by definition since $C$ is $s$-strong.

We show that the second stopping condition guarantees that the flow problems are well-defined. We start with the following lemma bounding the maximum number of flow problems we may run in a batch and stage (whether the flow problem is successful or not in the stage).
\begin{lemma}
\label{lem:num_flows}
For any batch $i$ we will work on at most $\frac{W_C}{10^{11}\del\tau^2}$ flow problems in stage $0$, and at most $\frac{W_C}{2\cdot 2^j10^7\del\tau}$ flow problems in any stage $j\geq 1$. 
\end{lemma}
\begin{proof}
We start with proving the claim for stage $0$. Assume by contradiction that we have worked on $\frac{W_C}{10^{11}\del\tau^2}$ flow problems in any batch $i$, but neither of the two stopping conditions are satisfied. In particular, this means that for every stage $j\geq 0$ the total excess in the Unit-Flow result of stage $j$ is less than $\frac{W_C}{10^6\tau(L+1)}$. We know that in order for a flow problem to reach stage $j$ but not stage $j+1$, the flow problem must lose at least \tempfix{$\frac{2^{j}10^6\del\ln m}{L+1}$}{$\frac{2^{j}10^6\del\tau}{L+1}$} units (in unit $\mu_{j}$) of supply in stage $j$ alone (i.e. as excess supply). Thus, dividing the upper-bound of the total excess in stage $j$ shows that, in any stage $j$, the number of unsuccessful flow problems  is at most $\frac{W_C}{2^j10^{12}\del\tau^2}$. Hence, the total number of flow problems that are unsuccessful over all stages $j\geq 0$ is at most $\frac{2W_C}{10^{12}\del\tau^2}$. As we have started $\frac{W_C}{10^{11}\del\tau^2}$ flow problems in stage $0$, we must have at least $\frac{W_C}{10^{12}\del\tau^2}=\frac{W_C}{10s_0}\geq \frac{W_C}{10s}$ successful flow problems in stage $L$, which contradicts the assumption that we haven't finished this batch. 

The claim for any $j\geq 1$ follows the similar argument. In particular, our analysis above says that the number of flow problems that can be unsuccessful over all stages $j'\geq j$ is at most $\frac{2W_C}{2^j10^{12}\del\tau^2}$, and if we already have $\frac{W_C}{2\cdot 2^j10^7\del\tau}$ flow problems reaching stage $j$, we must have at least $\frac{W_C}{4\cdot 2^j10^7\del\tau}$  of them being successful in stage $L$. Since $2^j\leq 2^L=\frac{2s}{10^7\del\tau}$, we have 
$
\frac{W_C}{4\cdot 2^j10^7\del\tau}\geq \frac{W_C}{4\cdot 2^L10^7\del\tau} =  \frac{W_C}{8s}\geq \frac{W_C}{10s},
$
which again is a contradiction.

\end{proof}

 With the above result, we can show below that there will be enough distinct nodes in $C$ that we can use as the source node for flow problems in stage $0$. Furthermore, we show the aggregate flow problem for any batch and stage has at most $W_C/2$ units of supply. This ensures we must get case $(1)$ or $(3)$ in~\Cref{thm:unit-flow}, and if we get case $(3)$ then $A$ (i.e. the side with excess) must be the smaller side of the cut.
\begin{lemma}\label{lemma:supplybound}
    There are enough distinct nodes to serve as source nodes in stage $0$, and in every stage the total source supply is at most $W_C/2$ in any aggregated flow problem. 
\end{lemma}

\begin{proof}
    We show first that there are enough source nodes in stage 0 to support all the (successful or not) flow problems. By~\Cref{lem:num_flows} we have at most $\frac{W_C}{10^{11}\del\tau^2}$ flow problems in a batch in stage $0$, so we need at most $\frac{5000\tau W_C}{10^{11}\del\tau^2}=\frac{5 W_C}{10^8\del\tau}$ nodes to serve as the source node in stage $0$ over all the $5000\tau$ batches. Recall we have WLOG~\Cref{assumption:enough-source}, so we have up to $\frac{W_C}{2\del}>\frac{5 W_C}{10^8\del\tau}$ nodes to use, which is enough in this case. 

    Recall that quantities of supply are in the unit $\mu_j$ when we discuss flow in stage $j$.
    We know the total sink capacity in the graph of stage $0$ (which is shared across all batches) is more than twice the total source supply as follows. We calculated above that we have at most $\frac{5 W_C}{10^8\del\tau}$ flow problems in stage $0$, and we add $10^7\del\tau$ source supply in each flow problem, so the total supply is at most $W_C/2$. Note the total sink capacity is $W_C$ in the shared graph.
    
    Now we show the total sink capacity in the graph of stage $j\geq 1$ for any batch $i$ is more than twice of the total supply. Since we work on at most $\frac{W_C}{2\cdot 2^j10^7\del\tau}$ flow problems in stage $j$ by~\Cref{lem:num_flows}, and each flow problem can have at most $2^j 10^7\del\tau$ units (in unit $\mu_j$) of source supply in stage $j$ since we start with $10^7\del\tau$ units (in unit $\mu_0$) in stage $0$, and only the absorbed supply goes into the next stage while the value of one unit being halved. Thus the total supply $|\Delta_{i,j}(\cdot)|\leq W_C/2$ through the procedure for any $i,j$.  
\end{proof}
\subsection{The Cut Case}
\label{sec:cut}
In this section we show that if in any batch $i$ and stage $j$, after Unit-Flow finishes with a flow problem, if the second stopping condition (i.e. Check $2$) is satisfied, then we can find a relatively balanced low conductance cut. Recall in this case we have at least $\frac{W_C}{10^6\tau(L+1)}=\Omega(\frac{W_C}{\tau\ln m_G})$ excess in the Unit Flow result for the aggregated flow problem of batch $i$ stage $j$.
%
\begin{lemma}
    \label{lem:cut}
    After (incrementally) running Unit-Flow on a flow problem, if we have at least $\Omega(\frac{W_C}{\tau\ln m_G})$ excess in the flow for some stage $j\ge 0$, then we can find a cut $(A,C\setminus A)$ of conductance $O(\frac{1}{\tau})$ such that
\[
\min\left(\vol^W(A),\vol^W(C\setminus A)\right) \ge \Omega\left(\frac{W_C}{\tau^2 \ln m_G}\right).
\]
\end{lemma}
\begin{proof}

We will apply Theorem~\ref{thm:unit-flow} to it to show the result. Note that the flow problem cannot fulfill Case $(1)$ of the theorem since there is excess remaining in the current pre-flow solution, and we also cannot get case $(2)$ as there is not enough supply to saturate all sinks since $|\Delta_{i,j}(\cdot)|\leq W_C/2$ for all $i$ and $j$ by \Cref{lemma:supplybound}.
Thus, we must get a cut as in case $(3)$. We look at the parameters of our flow problems show that the conductance and volume properties in the lemma follow from Theorem~\ref{thm:unit-flow}$(3)(b,c)$.   

In stage $0$, the source supply we put on any node $v$ is \tempfix{$10^7\del\ln m$}{$10^7\del\tau$}, and since $d^W(v)\geq \del$ we know the $\beta$ in Theorem~\ref{thm:unit-flow} is $O(\tau)$, and $U$ is \tempfix{$10^8\tau\ln m$}{$10^8\tau^2$}. Recall that we use \tempfix{$h=\Theta(\tau\ln m_G)$}{$h=\Theta(\tau^2\ln m_G)$} in stage $0$, so Theorem~\ref{thm:unit-flow}$(3)(b)$ guarantees the conductance of the cut is $O(\frac{1}{\tau})$. By \Cref{lemma:supplybound}, we have at most $W_C/2$ supply in the aggregate flow problem in each Unit-Flow computation, so the side $A$ in the returned cut must be the smaller side in terms of weighted volume. Since $\ex(T)\geq \Omega(\frac{W_C}{\tau\ln m_G})$, by Theorem~\ref{thm:unit-flow}$(3)(c)$ we know \tempfix{$\vol^W(A)\geq \Omega\left(\frac{W_C}{\ln^3m}\right)$}{$\vol^W(A)\geq \Omega\left(\frac{W_C}{\tau^2 \ln m_G}\right)$} from our choice of $U,h,\beta$.

It is straightforward to do the same calculation for stages $j\geq 1$ to show the conductance and volume guarantees, where we have $\beta=2,U=100\tau,h=\Theta(\tau\ln m_G)$ in these stages. 
\end{proof}
\paragraph{Edges on the cut.} When we find a relatively balanced low conductance cut as above, we will break the $C$ along the cut into $A$ and $C\setminus A$. The edges crossing the cut are duplicated to exist in both $G\{A\}$ and $G\{C\setminus A\}$ (as self-loops) by our definition of these sub-graphs. Note that if $C$ is $s$-strong, $G\{A\}$ and $G\{C\setminus A\}$ must also be $s$-strong, even though we add a copy of each cut edge to both sides. Also since the conductance of the cut is $O(1/\tau)$, the total weight of the inter-cluster edges is at most a $O(1/\tau)$ fraction of the weighted volume of the smaller side, so even with two copies of each cut edge on both side of a cut, we can charge $O(1/\tau)$ to each unit of weight on the smaller side of the cut. Since a unit of weight can be on the smaller side of a cut at most $O(\ln m_G)$ times before we get to a trivial cluster, the total additional weight from adding the cut boundary edges over all recursive calls starting from the whole graph is at most $O(\frac{\vol^W(G)\ln m_G}{\tau})\leq 0.001 \vol^W(G)$ as we have a large enough $\tau=\Omega(\ln m_G)$, so for simplicity we can safely ignore the impact of the additional edges to the total asymptotic running time.

\subsection{The Certification Case}
\label{sec:certification}
In the previous section we showed that if the algorithm terminated in the second stopping condition (Check $2$), then the algorithm finds a relatively balanced low conductance cut. Next, we look at the case when the algorithm terminated in Check $1$. In this section on we will explicitly multiply values in stage $j$ by the unit $\mu_j$ so all quantities are in unit $1$. 

We now consider only the flow problems that successfully finishes all $L$ stages. For such a flow problem associated with source node $v$, we call $v$ together with the edges on which its supply was transported  an \emph{edge-bundle} with \emph{center} $v$. As the algorithm terminated in Check 1, we find $5000\tempfix{\ln m}{\tau}$ batches, each with $\frac{W}{10s}$ edge-bundles. That is, each edge-bundle routes at least $1.8s$ (i.e. $0.9\cdot 10^7\del \tau \cdot \mu_0$) supply to sinks in stage $L$. We start with the following definition.

 \begin{definition}
For an edge-bundle with source node $v$ in a $s$-strong component $C$, we call the edge-bundle \emph{$s$-captured} if in stage $0$, the flow problem from $v$ is successful and at most $0.5s$ of its supply is routed to nodes in $C\setminus S$ for any cut $S$ in $G$ with $\vol^W(S\cap C)\leq \vol^W(C\setminus S)$ and cut-size $w(S,G\setminus S)\leq \cutsize$.  
We call an edge-bundle $s$-free if it's not $s$-captured. 
 \end{definition}
Note in the above definition, since $C$ is $s$-strong w.r.t the same $\del$ (\Cref{def:strength}), we know 
all the cut $S$ being considered must have $\vol^W(S\cap C)\leq s$. We start with the following result, which implies that in the case where we successfully get $5000\tempfix{\ln m}{\tau}$ bathes each with $\frac{W}{10s}$ edge-bundles, all these edge-bundles are $s$-free since at least $1.8s>1.6s$ supply is routed to sinks in stage $L$ for each edge-bundle.
\begin{lemma}
\label{lem:no-s-captured}
The flow problem corresponding to any $s$-captured edge-bundle cannot have a total of more than $1.6$s supply routed to sinks in stage $L$. 
\end{lemma}
\begin{proof}
Consider any $s$-captured edge-bundle and a corresponding cut $S$. By definition at most $0.5s$ supply is routed to $C\setminus S$ in stage $0$. The remaining at least $1.5s$ supply is inside $C\cap S$ after stage $0$, and we can look at how much can be routed into $C\setminus S$ in later stages.

By our choice of edge capacity for edge $e$ being $U_j\mu_j w_e=100\tau\frac{2s}{2^j 10^7\del\tau}w_e$ in stage $j\geq 1$, we know that over all later stages, the total edge capacity is at most $\frac{2sw_e}{10^5\del}$ for edge $e$. If we sum over all edges across the cut $C\cap S$ and $C\setminus S$, the total capacity is at most $\frac{2s\cdot \cutsize}{10^5\del}$, so we can route at most $0.1s$ supply into $C\setminus S$. 

In total this means at most $0.6s$ supply can be routed to nodes (and thus absorbed by sinks) in $C\setminus S$ over all $L$ stages. The nodes in $C\cap S$ can absorb at most $s$ supply since $\vol^W(C\cap S)\leq s$, so in total at most $1.6s$ supply can be absorbed by sinks. 
\end{proof}
\begin{lemma} 
\label{lem:num-s-free}
There are at most $50\tau$ edge-bundles that are $s$-free but with centers inside a cut $S$ such that $\vol^W(S\cap C)\leq \vol^W(C\setminus S)$ and $w(S\cap C,C\setminus S)\leq \cutsize$. 
\end{lemma}
\begin{proof}
For an edge-bundle with center inside $C\cap S$ to be $s$-free, the flow problem must send at least $0.5s$ supply from $v\in C\cap S$ to nodes in $C\setminus S$ in stage $0$. The total edge capacity across $C\cap S$ to $C\setminus S$ in stage $0$ is at most $U_0\mu_0\cdot\cutsize=10^8\tau^2\frac{2s}{10^7\del\tau}\cutsize\leq 22s\tau$, so we can have at most $\frac{22s\tau}{0.5s}\leq 50\tau$ such edge-bundles.
\end{proof}
We construct a new flow problem $\Pi_F$ based on all the successful stage $L$ flows we get as follows. There are in total $5000\tau \cdot \frac{W_C}{10s}$ such flow problems, and in the stage-$L$ pre-flow result for each of them, we have at least $1.8s$ supply absorbed in sinks. We aggregate the absorbed supply over all the $5000\tau \cdot \frac{W_C}{10s}$ pre-flows, and scale the amount down by $500\tau$, and we will use the result amount of supply on each node as the source supply function $\tilde{\Delta}$ for our new flow problem $\Pi_F$. The sink capacity of any node $v$ in our new flow problem will again be $d^W(v)$.
\begin{observation}
\label{obs:total_supply}
The total amount of source supply in $\Pi_F$ is more than the total sink capacity (i.e. $W_C$), since $|\tilde{\Delta}(\cdot)|\geq \frac{1.8s}{500\tau}\frac{5000\tau W_C}{10s}=1.8W_C>W_C$. Moreover, $\tilde{\Delta}(v)\leq \frac{5000\tau\cdot d^W(v)}{500\tau}=10 d^W(v)$, since pre-flows in the same batch share the sink capacity $d^W(v)$.
\end{observation}

Note that any source supply in $\tilde{\Delta}$ on some node $v$ is a result of our previous flow problems, that is, the supply on $v$ must be from either the initial source supply we put on $v$ in stage $0$, or being routed to $v$ by the pre-flows over the stages $0,\ldots L$. Consider any cut $S$ being considered in the definition of $s$-captured, we first bound the amount of source flow in $\tilde{\Delta}$ on nodes in $S\cap C$.
\begin{lemma}\label{lem:s-bottleneck-inflow} For any cut $S$ such that $w(S,G\setminus S)\leq \cutsize$ and $\vol^W(S\cap C)\leq \vol^W(C\setminus S)$, we have
$\tilde{\Delta}(C\cap S)\leq s/4$
\end{lemma}
\begin{proof}
We first bound the amount of supply inside $S$ from the initial source supply we put on nodes in $S$ in stage $0$. By~\Cref{lem:no-s-captured} and~\Cref{lem:num-s-free}, at most $50\tau$ nodes in $S$ are centers of the edge-bundles we include, so the total amount of supply is at most $50\tau\cdot 2s=100\tau s$.

Now we bound the amount of supply that can be routed in to $S$ by the pre-flows. For stage $0$, since all flow problems share a single copy of the graph, the total supply can be routed into $S$ over all of them is at most the total capacity into $S$, which is
\[
U_0\cdot w(S\cap C,C\setminus S)\cdot \mu_0 \leq 10^8\tau^2\cdot\cutsize\cdot \frac{2s}{10^7\del\tau}\leq 22\tau s.
\]
For stages $1,\ldots, L$, the different batches don't share edge capacities, so the total capacity into $S$ is
\begin{align*}
5000\tau \cdot w(S\cap C,C\setminus S)\cdot \sum_{j=1}^L U_j\mu_j \leq &5000\tau\cdot \cutsize\cdot 2U_1\mu_1\\
= & 5000\tau\cdot \cutsize\cdot 100\tau\cdot \frac{2s}{10^7\del\tau} \leq \tau s.
\end{align*}
Since we scale the supply down by $500\tau$ to get $\tilde{\Delta}$, we know
\[
\tilde{\Delta}(S\cap C)\leq \frac{100\tau s + 22\tau s + \tau s}{500\tau} \leq \frac{s}{4}.
\]
\end{proof}
For the flow problem $\Pi_F$, we will again use $U=100\tau$ and call Unit-Flow with $h=\Theta(\tau\ln m_G)$, and we have the following result.
\begin{lemma}
\label{lem:certify}
We can either certify $C$ is $s/2$-strong, or find a cut $A\subsetneq C$ such that $\Phi_C(A)=\frac{w(A,C\setminus A)}{\min(\vol^W(A),\vol^W(C\setminus A))}\leq \tau$, $\vol^W(A)= \Theta(W_C)$ and $A$ is $s/2$-strong.
\end{lemma}
\begin{proof}
Consider any cut $S$ as in~\Cref{lem:s-bottleneck-inflow}, we first show that the set $\tilde{C}$ of nodes whose sinks are saturated in the Unit-Flow result of $\Pi_F$ must have $\vol^W(\tilde{C}\cap S)\leq s/2$. This is because at most $s/4$ supply starts in $S$ by~\Cref{lem:s-bottleneck-inflow}, and in the pre-flow for $\Pi_F$, the additional supply that can be routed into $S$ is bounded by the total edge capacity into $S$, which by our choice of $U$ is at most $100\tau\cdot\cutsize\leq s_0/\tau < s/\tau$, and thus the total supply in $S$ when Unit-Flow finishes is at most $s/2$, and thus $\vol^W(\tilde{C}\cap S)\leq s/2$. Note this certified that $\tilde{C}$ (and thus any subset of it) is $s/2$-strong.

By~\Cref{obs:total_supply}, we know there must be excess in the pre-flow from Unit-Flow, so we cannot get case $(1)$ in~\Cref{thm:unit-flow}. If we get case $(2)$, then we know $C=\tilde{C}$ and by the above argument $\vol^W(C\cap S)\leq s/2$ for all such $S$, which certifies $C$ is $s/2$-strong. If we get case $(3)$, let $(A,C\setminus A)$ be the returned cut. We know all nodes in $A$ are saturated, i.e. $A\subset \tilde{C}$, so $A$ is $s/2$-strong. The claim on conductance $\Phi(A)$ and volume $\vol^W(A)$ follow $(3)(b)$ and $(3)(c)$ in~\Cref{thm:unit-flow} respectively by our choice of $U=100\tau,h=\Theta(\tau\ln W_C)$, and note by~\Cref{obs:total_supply} $\beta\leq 10, \ex(T)\geq 0.8W_C$.
\end{proof}
We conclude this section with the proof of the main lemma of $s_0$-strong partition.
\StrongPartition*
\begin{proof}
Our procedure starts with the graph $G$, and recursively running the certification procedure until all components are certified to be $s_0$-strong. Note a component with weighted volume at most $s_0$ is by definition $s_0$-strong. Thus, if for any component $C$ we always divide it along a relatively balanced low conductance cut whose smaller side has weighted volume at least $\Omega(\vol^W(C)/\polylog \vol^W(G))$, then the depth of our recursion is at most $O(\polylog \vol^W(G))$ until all components have weighted volume at most $s_0$.

In our certification procedure, the only case when we don't divide $C$ along a relatively balanced cut is in the case of~\Cref{sec:certification} where we get a cut $(A,C\setminus A)$ as in~\Cref{lem:certify}. Since $\vol^W(A)=\Theta(\vol^W(C))$, if it is not relatively balanced, we know $A$ must be the larger side, and we certify it's $s/2$-strong. Thus, when we cannot find a relatively balanced cut, we make progress by certifying a smaller strength on the larger side of the cut (or possibly the entire component $C$). This can happen at most $\ln\frac{\vol^W(G)}{s_0}$ times before we are done with the component. Consequently, when we look at both cases (i.e. cut or certify), the depth of the recursion is at most $O(\polylog \vol^W(G))$ since in each level we either decrease the weighted volume by a factor of $1-\Omega(1/\polylog \vol^W(G))$ or certify a strength of $1/2$ the current certified strength for the larger side of the unbalanced cut.

The running time of our certification procedure for a component $C$ is $\tilde{O}(m_C)$ if $C$ has $m_C$ edges. This is because we run one Unit-Flow computation for the aggregate flow problem in each stage and each batch, so $L\cdot 5000\tau=O(\tau\ln m_C)$ of them in total (plus the last flow problem $\Pi_F$ if needed). Each of the Unit-Flow computation takes time $\tilde{O}(m_C\cdot h)$ by~\Cref{thm:unit-flow}, which is at most $\tilde{O}(m_C\tau^2)$ by our choice of $h$, i.e. $h=\tau^2\ln m_G$ for stage $0$, and $h=\tau\ln m_G$ for all later stages and $\Pi_F$. In total, the running time of all the Unit-Flow calls take $\tilde{O}(m_C\tau^3)$ in the certification procedure of $C$. By our earlier discussion, the recursion depth of the partition procedure is at most $O(\polylog \vol^W(G))$, the total running time for all the Unit-Flow calls over all components is $\tilde{O}(m_G\tau^3)$. 

As to the total weight of edges between distinct components, since we only divide a component along cut of conductance at most $O(1/\tau)$, so we can charge the weight of cut edges to the weighted volume of the smaller side of the cut. Since one unit of volume can be on the smaller side of a cut at most $O(\ln \vol^W(G))$ times before we stop, each unit of weight is charged at most $O(\frac{\ln m_G}{\tau})$, so the total weight of edges between components is $O(\frac{\vol^W(V)\ln |V|}{\tau})$ (assuming polynomially bounded edge weights). Given our choice of $\tau=\Theta(\sqrt{s_0/\del})=\Theta(\ln^{c/2}|V|)$, the running time is $\tilde{O}(m_G\tau^3)=\tilde{O}(m_G)$ as in the lemma statement. The weight of edges between components is $O(\frac{\ln|V|}{\tau})\cdot \vol^W(G)$, which is $O(\frac{\sqrt{\del}\ln|V|}{\sqrt{s_0}})\cdot \vol^W(G)$ as claimed in the lemma.
\end{proof}

\section{From $s_0$-strength to efficient uncrossing.}\label{sec:uncrossing}

Recall that an $s_0$-strong decomposition, ${\cal D}$, of a graph $G$ is a decomposition of $V$ in which any min-cut, $S$, has
at most volume $s_0$ in common with any cluster $A \in {\cal D}$ in the
decomposition, i.e., the volume of $A \cap S$ is at most $s_0$.

In this section, we convert such a decomposition into a decomposition,
${\cal D'}$ such that for any min-cut $S$ and cluster $A \in {\cal D'}$  an uncrossing of the cut $S$ with $A$ does not locally increase
the cut-size ``too much''. Recall the definition of boundary-sparsity given in \Cref{def:boundary-sparse}.



Intuitively, for a cut $S$ that crosses $A$, if $S\cap A$ is not $\beta$-boundary-sparse in $A$, then uncrossing $S$ with respect to $A$
only increases the cut size by a small amount proportional to $1/\beta$. In other words, we want to find $\beta$-boundary-sparse cuts in $A$ since these are the cuts that do not allow for cheap uncrossing. Ideally, in our final decomposition, no cluster would contain any $\beta$-boundary-sparse sets. In reality, we will only be able to achieve a slightly relaxed version of this
statement; see \Cref{lem:small-clusters} below.

We will proceed by considering each cluster in the original
decomposition and further decomposing it as necessary.  In
Section~\ref{sec:smallpieces}, we give a procedure that works for
clusters that have volume $\tilde{O}(s_0) $ and, thus, consist of $\tilde O(1)$ vertices. We call such clusters \emph{small} clusters. In Section
\ref{sec:bigpieces} we give a procedure for 
\emph{large} clusters, i.e. clusters that are not small.

\subsection{Efficient uncrossing for small clusters.}
\label{sec:smallpieces}
Recall that a small cluster  contains $\tilde O(1)$ vertices. Thus a running time that is polynomial in the number of vertices of $A$ is sufficient for our purpose.
Throughout this section we assume that 
$\tilde\lambda \le 1.01\lambda$ is an approximate lower bound of the minimum cut value that is available to our algorithm.
\SmallCuts

To prove \Cref{lem:small-clusters}, 
we will need the following helper lemma, whose proof is standard.
\begin{theorem}\label{thm:enumerate-cuts}
Given a weighted graph $H$ and a parameter $\alpha\ge1$, there is a deterministic algorithm that enumerates all $\alpha$-approximate minimum cuts of $H$, of which there are $O(n^{\lfloor2\alpha\rfloor})$. The algorithm runs in time $\tilde{O}(m^2n^{\lfloor2\alpha\rfloor}) \le \tilde{O}(n^{\lfloor 2\alpha\rfloor+4})$.
\end{theorem}
\begin{proof}
First, compute a maximum (spanning) tree packing of $H$, which can be done in $\tilde{O}(mn)$ time~\cite{barahona1995packing} to produce a packing of $\tilde{O}(m)$ trees. It was shown in~\cite{karger2000minimum}) that for any $\alpha$-approximate minimum cut $S$, there is a tree $T$ in a maximum tree packing for which $\partial_TS\le2\alpha$, i.e., at most $2\alpha$ many edges of $T$ cross the cut $S$. The algorithm iterates over all $\tilde{O}(m)$ trees in the packing, and for each tree $T$, considers all ways to remove at most $2\alpha$ edges, group the components together into two sides of a cut, and determine the size of the cut. For constant $\alpha$, this takes $O(mn^{\lfloor2\alpha\rfloor})$ time per tree, for a total of $\tilde{O}(m^2n^{\lfloor2\alpha\rfloor})$ time. The algorithm then computes the minimum cut $\lambda$ over this collection, which is the minimum cut of $G$, and outputs the subcollection of cuts of value at most $\alpha \lambda$. The bound of $O(n^{\lfloor2\alpha\rfloor})$ for the number of such cuts was shown in \cite{karger2000minimum}.
\end{proof}
We now prove \Cref{lem:small-clusters}.
The algorithm for \Cref{lem:small-clusters} has two stages. The first stage, described in the next lemma, decomposes $A$ into clusters such that (a) each has boundary size $O(\tilde\lambda)$ and minimum cut in the induced graph of at least $0.49 \tilde \lambda$ and (b) the size of the sum of their boundaries is small, namely $O(\epsilon^{-1}\partial A)$. These clusters are further broken down in the next stage given in \Cref{lem:small-clusters-2}. To shorten the notation given a cluster $A$ and a set $U \subseteq A$, we use $\partial U$ to denote $w(U, \bar U)$ and
$\partial_{G[A]}U$ 
to denote $w(U, A \setminus U)$.
\begin{lemma}\label{lem:small-clusters-1}
Let $A\subseteq V$ be a cluster. Suppose an algorithm is given a mincut estimate $\tilde\lambda\in[\lambda,1.01\lambda]$. Then, the algorithm can decompose $A$ into a disjoint union of clusters $A_1\cup A_2\cup\cdots\cup A_k=A$ such that
 \begin{enumerate}
 \item For each cluster $A_i$, either $\partial A_i\le3\tilde\lambda$ or there is no $(1-\epsilon)$-boundary-sparse set $U$ in $A_i$ with $\partial_{G[A_i]}U\le1.01\tilde\lambda$.\label{item:small-clusters-1-1}
 \item The sum of boundaries $\sum_i\partial A_i$ is $O(\epsilon^{-1}\partial A)$.\label{item:small-clusters-1-2}
 \item For each cluster $A_i$, the induced graph $G[A_i]$ has minimum cut at least $0.49\tilde\lambda$.\label{item:small-clusters-1-3}
 \end{enumerate}
The algorithm takes time $\tilde{O}(|A|^9)$.
\end{lemma}
\begin{proof}
We first describe an algorithm that fulfills the first two properties. At the end, we augment it to satisfy property~(\ref{item:small-clusters-1-3}).

\emph{ Property \ref{item:small-clusters-1-1}.}
Given a cluster $A$ the algorithm maintains a current decomposition $\mathcal D$, initialized to the singleton set $\{A\}$. The algorithm iteratively considers sets in $\mathcal D$ that have not been processed yet until all have been processed, starting with $A$. For a given set $A'$, the algorithm does the following.

First, if $\partial A'\le3\tilde\lambda$, then the algorithm keeps $A'$ as is, declares $A'$ as processed, and proceeds to the next set in $\mathcal D$. Otherwise, it applies the algorithm of \Cref{thm:enumerate-cuts} to the induced graph $G[A']$ for $\alpha=2.1$, which takes $\tilde{O}(|A'|^8)$ time, and let $\mathcal C$ be the set of returned cuts. Compute the cut $U^* \in\mathcal C$ such that  $\partial_{G[A']}U^*$ is the smallest cut over all cuts in $\mathcal C$. 
By the choice of $\mathcal C$, $U^*$  is the minimum cut in $G[A']$. We now have a few cases.
\begin{enumerate}[(a)]
\item If $\partial_{G[A']}U^*\le0.49\tilde\lambda$, then the algorithm decomposes $A'$ into $U^*$ and $A'\setminus U^*$, i.e., it removes $A'$ from $\mathcal D$ and adds the unprocessed sets $U^*$ and $A'\setminus U^*$ to $\mathcal D$. (The sets $U^*$ and $A'\setminus U^*$ are processed at a later iteration.)\label{item:small-clusters-case-1}
\item If $\partial_{G[A']}U^*>0.49\tilde\lambda$, then it follows from our choice of $\alpha=2.1$ that the collection $\mathcal C$ contains all sets $U\subseteq A'$ with $\partial_{G[A']}U\le1.01\tilde\lambda$. The algorithm iterates over all sets in $\mathcal C$ and looks for one that is $(1-\epsilon)$-boundary-sparse and has boundary size $\le 1.01 \tilde \lambda$. If a $(1-\epsilon)$-boundary-sparse set with boundary size $\le 1.01 \tilde \lambda$ does not exist, then the algorithm keeps $A'$ as is, declares $A'$ as processed, and proceeds to the next set in $\mathcal D$. Otherwise, let $U$ be an arbitrary $(1-\epsilon)$-boundary-sparse set with boundary size $\le 1.01 \tilde \lambda$. The algorithm decomposes $A'$ into $U$ and $A'\setminus U$, i.e., removes $A'$ from $\mathcal D$ and adds the unprocessed sets $U$ and $A'\setminus U$ into $\mathcal D$. (The sets $U$ and $A'\setminus U$ are processed at a later iteration.)\label{item:small-clusters-case-2}
\end{enumerate}

All sets in the final decomposition satisfy property~(\ref{item:small-clusters-1-1}) by construction. 
As each iteration of the algorithm processes one set, each set is processed only once, and the collection of processed sets forms a laminar family, there are at most $|A|$ many iterations.
As each iteration takes time $\tilde{O}(|A|^8)$, the total running time is $\tilde{O}(|A|^9)$. 

\emph{ Property \ref{item:small-clusters-1-2}.}
Let $\mathcal D_1$ be the decomposition of $A$ at the end of the algorithm.
If $\partial A\le3\tilde\lambda$ then the algorithm does nothing and property~(\ref{item:small-clusters-1-2}) is trivially true. Otherwise, assume that the algorithm decomposes $A$. Throughout the algorithm, we maintain a potential function
\[ \Phi(\mathcal D) = \sum_{A'\in\mathcal D}\max\{0,\partial A'-2.1\tilde \lambda\} .\]
The initial value of $\Phi(\mathcal D) = 
\Phi(\{A\}) = O(\partial A).$
We will show (1) that $\Phi(\mathcal D)$ does not increase throughout the decomposition, which implies that
$\Phi(\mathcal D_1) \le \Phi(\{A\})$. Moreover, we show that (2) each time a set $A'$ with $\partial A'\le2.2\tilde\lambda$ is added to $\mathcal D$, the potential $\Phi(\mathcal D)$ drops by at least $\Omega(\epsilon\tilde\lambda)$. These two claims together imply property~(\ref{item:small-clusters-1-2}) because
\begin{align*}
\sum_{A' \in \mathcal D_1}\partial A'&= \sum_{A' \in \mathcal D_1:\partial A'\le2.2\tilde\lambda}\partial A' + \sum_{A' \in \mathcal D_1:\partial A'>2.2\tilde\lambda}\partial A'
\\&\le \frac{\Phi(\{A\})}{\Omega(\epsilon \tilde \lambda)} \cdot 0.1 \tilde \lambda + O(\Phi(\mathcal D_1)) \le
O(\Phi(\{A\})/\epsilon)\le O(\partial A/\epsilon) .
\end{align*}

To show the two claims, suppose a set $A'$ is decomposed into $U$ and $A'\setminus U$ for some $U\subseteq A'$.
By construction, we always have $\partial_{G[A']}U \le 1.01 \tilde \lambda$. 
If $\partial U$ and $\partial(A'\setminus U)$ are both at least  $2.1\tilde\lambda$, then the net increase in $\Phi(\mathcal D)$ is
\begin{align*}
& \max\{0,\partial U-2.1\tilde\lambda\} + \max\{0,\partial(A'\setminus U)-2.1\tilde\lambda\}-\max\{0,\partial A'-2.1\tilde\lambda\}
\\&= (\partial U-2.1\tilde\lambda) + (\partial(A'\setminus U)-2.1\tilde\lambda) - (\partial A'-2.1\tilde\lambda)
\\&= 2w(U,A'\setminus U) - 2.1\tilde\lambda \le 2\cdot 1.01\tilde\lambda-2.1\tilde\lambda < -0.08\tilde\lambda.
\end{align*}
If $\partial U$ and $\partial(A'\setminus U)$ are both at most $2.1\tilde\lambda$, then the net increase in $\Phi(\mathcal D)$ is
\begin{align*}
& \max\{0,\partial U-2.1\tilde\lambda\} + \max\{0,\partial(A'\setminus U)-2.1\tilde\lambda\}-\max\{0,\partial A'-2.1\tilde\lambda\}
\\&\le 0 + 0- (\partial A'-2.1\tilde\lambda) \le  0 + 0- (3\tilde\lambda-2.1\tilde\lambda) \le -0.9 \tilde\lambda.
\end{align*}
Else, suppose without loss of generality that $\partial U\le2.1\tilde\lambda\le\partial(A'\setminus U)$. The net increase in $\Phi(\mathcal D)$ is
\begin{align*}
& \max\{0,\partial U-2.1\tilde\lambda\} + \max\{0,\partial(A'\setminus U)-2.1\tilde\lambda\}-\max\{0,\partial A'-2.1\tilde\lambda\}
\\&= 0 + (\partial(A'\setminus U)-2.1\tilde\lambda) - (\partial A'-2.1\tilde\lambda)
\\&= \partial(A'\setminus U)-\partial A'
\\&=w(U,A'\setminus U)-w(U,\overline{A'}) .
\end{align*}
If $\partial_{G[A']}U\le0.49\tilde\lambda$, i.e., the algorithm went through case~(\ref{item:small-clusters-case-1}), then
\[ w(U,\overline{A'}) = \partial U-w(U,A'\setminus U) \ge \lambda - 0.49\tilde\lambda \ge 0.99\tilde\lambda-0.49\tilde\lambda=0.5\tilde\lambda ,\]
and the net increase in $\Phi(\mathcal D)$ is
\[ w(U,A'\setminus U)-w(U,\overline{A'}) \le 0.49\tilde\lambda-0.5\tilde\lambda < -0.01 \tilde\lambda. \]
If  $\partial_{G[A']}U>0.49\tilde\lambda$, i.e., case~(\ref{item:small-clusters-case-2}), then $U$ is $(1-\epsilon)$-boundary-sparse in $A'$ by construction, so the net increase in $\Phi(\mathcal D)$ is
\[ w(U,A'\setminus U)-w(U,\overline{A'}) \le w(U,A'\setminus U) - \frac{w(U,A'\setminus U)}{1-\epsilon} \le -\Omega(\epsilon w(U,A'\setminus U)) \le -\Omega(\epsilon \tilde\lambda). \]
This concludes the proof that the algorithm that fulfills the properties~(\ref{item:small-clusters-1-1})~and~(\ref{item:small-clusters-1-2}).

\emph{Property \ref{item:small-clusters-1-3}.}
To satisfy property~(\ref{item:small-clusters-1-3}), we take the decomposition $\mathcal D_1$ after termination of the algorithm above and, while there exists a cluster $A'$ in the current decomposition whose induced graph $G[A']$ has minimum cut $0.49\tilde\lambda$ or less, we find the minimum cut $U^*\subseteq A'$ and decompose $A'$ into $U^*$ and $A'\setminus U^*$. Computing the minimum cut can be done by applying \Cref{thm:enumerate-cuts} for $\alpha=1$, which takes $\tilde{O}(|A'|^6)$ time.

We claim that this additional step increases $\sum_{A'\in\mathcal D_1}\partial A'$ by only a constant factor. Let $\mathcal D$ be the decomposition during this process and let $\mathcal D_{2}$ be the decomposition at the end of this process. Consider the potential function
\[ \Phi'(\mathcal D)=\sum_{A'\in\mathcal D}(\partial A'-0.98\tilde\lambda) .\]
Note that $\Phi'(\mathcal D_1) \le 
\sum_{A'\in\mathcal D_1}\partial A'
$.
Since $\partial A'\ge\lambda\ge\frac1{1.01}\tilde\lambda\ge0.99\tilde\lambda$, we have $\partial A'\le 99(\partial A'-0.98\tilde\lambda)$ and $\sum_{A'\in\mathcal D}\partial A'\le O(\Phi'(\mathcal D))$. 
In particular, $\sum_{A'\in\mathcal D_2}\partial A' \le O(\Phi'(\mathcal D_2))$.
We will show that $\Phi'(\mathcal D)$ does not increase over time. 
This shows the claim as it implies that $\sum_{A'\in\mathcal D_2}\partial A'\le O(\Phi'(\mathcal D_2))
\le
O(\Phi'(\mathcal D_1))
=
O(\sum_{A'\in\mathcal D_1}\partial A')
$. To show that $\Phi'$ does not increase
upon decomposing some $A'\in\mathcal D$ into $U$ and $A'\setminus U$, the net increase of $\Phi'(\mathcal D)$ is
\[ (\partial U-0.98\tilde\lambda)+(\partial(A'\setminus U)-0.98\tilde\lambda)-(\partial A'-0.98\tilde\lambda) = 2w(U,A'\setminus U)-0.98\tilde\lambda \le 2\cdot0.49\tilde\lambda-0.98\tilde\lambda\le0 ,\]
proving the claim. This concludes the proof.
\end{proof}

Using \Cref{lem:small-clusters-1} as a first  step, we now apply the following lemma to each cluster with boundary at most $3\tilde\lambda$ output by \Cref{lem:small-clusters-1}.

\begin{lemma}\label{lem:small-clusters-2}
Let $A\subseteq V$ be a cluster satisfying the following two guarantees:
 \begin{itemize}
 \item $\partial A\le3\tilde\lambda$, and
 \item the minimum cut of $G[A]$ is at least $0.49\tilde\lambda$.
 \end{itemize}
Let $0<\epsilon\le0.01$ be a parameter. Suppose an algorithm is given a mincut estimate $\tilde\lambda\in[\lambda,1.01\lambda]$. Then, the algorithm can decompose $A$ into a disjoint union of clusters $A_1\cup A_2\cup\cdots\cup A_k=A$ such that
 \begin{enumerate}
 \item For any set $S\subseteq A$ with $\partial_{G[A]}S\le1.01\tilde\lambda$, there is a partition $\mathcal P$ of $\{A_1,A_2,\ldots,A_k\}$ such that for each part $P\in\mathcal P$, the set $S\cap\bigcup_{A'\in P}A'$ is non-$(1-\epsilon)$-boundary-sparse in $\bigcup_{A'\in P}A'$.\label{item:small-clusters-2-1}
 \item There are at most $O(\epsilon^{-2}(\log|A|)^{O(1)})$ many clusters, and each cluster $A_i$ satisfies $\partial A_i\le\partial A$.\label{item:small-clusters-2-2}
 \end{enumerate}
The algorithm runs in time $\tilde{O}(|A|^8)$.
\end{lemma}
\begin{proof}
Apply \Cref{thm:enumerate-cuts} for $\alpha=2.1$ to obtain, in time $\tilde{O}(|A|^8)$, a collection $\mathcal C_0$ of $\tilde{O}(|A|^4)$ sets that contains all $\emptyset\subsetneq U\subsetneq A$ with $\partial_{G[A]}U\le1.01\tilde\lambda$. Take the subcollection $\mathcal C\subseteq\mathcal C_0$ of all $U\in\mathcal C$ with $\partial_{G[A]}U\le1.01\tilde\lambda$.   The rest of the algorithm recursively decomposes an input cluster $A'$ (starting with $A'=A$) while keeping a set $\mathcal C'\subseteq\mathcal C$ that starts as $\mathcal C'=\mathcal C$ and may lose elements over time.

Initially, the algorithm removes from $\mathcal C'$ all cuts $C\in\mathcal C'$ such that $C\cap A'$ is not $(1-\epsilon)$-boundary sparse in $A'$. We call this the initial \emph{pruning} step. The recursive algorithm then considers the following cases in order:
 \begin{enumerate}[(a)]
 \item If there is $C\in\mathcal C'$ that crosses $A'$ such that $\partial_{G[A']}(C\cap A')\le0.4\tilde\lambda$, decompose $A'$ into $C\cap A'$ and $A'\setminus C$ and recursively call $C\cap A'$ and $A'\setminus C$ with the updated $\mathcal C'$.\label{item:small-clusters-2-case-1}
 \item Otherwise, if there exists $C\in\mathcal C'$ that crosses $A'$
 , take a set $C \in \mathcal C'$ with minimum $|C\cap A'|$ crossing $A'$. Decompose $A'$ into $C\cap A'$ and $A'\setminus C$ and recursively call $C\cap A'$ and $A'\setminus C$ with the updated $\mathcal C'$.\label{item:small-clusters-2-case-2}
 \item Otherwise, return the trivial decomposition $A'$.\label{item:small-clusters-2-case-3}
 \end{enumerate}
 \emph{Running time analysis.}
In the initial pruning step, the algorithm can check each of the $\tilde{O}(|A|^4)$ sets in $\mathcal C'$ in $O(|A|^2)$ time. As the set of processed sets forms a laminar family, there are $O(|A|)$ many recursive instances, so the recursive algorithm takes time $\tilde{O}(|A|^7)$. Together with the initial $\tilde{O}(|A|^8)$ time to compute $\mathcal C'$, this establishes the running time.

\emph{Correctness: Property (\ref{item:small-clusters-2-1}).}
We first show that the resulting decomposition satisfies property~(\ref{item:small-clusters-2-1}) of the lemma. Consider a set $S\subseteq A$ with $\partial_{G[A]}S\le1.01\tilde\lambda$. If $S=\emptyset$ or $S=A$, then property~(\ref{item:small-clusters-2-1}) is trivial because $\emptyset$ and $A$ are always non-$(1-\epsilon)$-boundary sparse in $A$. 

Thus we are left with proving property~(\ref{item:small-clusters-2-1}) if $S$ is a cut in $A$.
Since $\partial_{G[A]}S\le1.01\tilde\lambda$, we have $S\cap A\in\mathcal C$, so initially $S\cap A\in\mathcal C'$ on the root instance $A$. We now prove property~(\ref{item:small-clusters-2-1}) by induction on the recursion tree of the algorithm. As base cases, we consider instances $A'$ where $S$ is removed from $\mathcal C'$ in the initial pruning step. Since $S$ is removed from $\mathcal C'$, it is non-$(1-\epsilon)$-boundary sparse in $A'$, i.e.,~property~(\ref{item:small-clusters-2-1}) holds for the trivial partition with the entire decomposition of $A'$ as a single part. For instances $A'$ where $S\in\mathcal C'$ survives the initial pruning step, we apply induction on the two recursive instances $C\cap A'$ and $A'\setminus C$ and take the union of the two partitions, which is a partition of the decomposition returned by instance $A'$. 
It follows that property~(\ref{item:small-clusters-2-1}) holds.

\emph{Correctness: Property (\ref{item:small-clusters-2-2}).}
To show property~(\ref{item:small-clusters-2-2}), 
we first make two observations.
\begin{observation}\label{o:1}
If the algorithm decomposes by case~(\ref{item:small-clusters-2-case-1})
into $U=C\cap A'$ and $A'\setminus C=A'\setminus U$, then 
$\partial U \le  \partial A'-0.19\tilde\lambda$ and
$\partial(A' \setminus U) \le  \partial A'-0.19\tilde\lambda$.
\end{observation}
\begin{proof}
Suppose the algorithm decomposes into $U=C\cap A'$ and $A'\setminus C=A'\setminus U$. First observe that if the algorithm decomposes by case~(\ref{item:small-clusters-2-case-1}), then
\begin{align*}
&w(U,\overline{A'}) = \partial U-w(U,A'\setminus U) \ge \lambda - 0.4 \tilde\lambda \ge 0.99\tilde\lambda-0.4\tilde\lambda=0.59\tilde\lambda
\\\implies&\partial(A'\setminus U)=\partial A'-w(U,\overline{A'})+w(U,A'\setminus U) \le \partial A'-0.59\tilde\lambda+0.4\tilde\lambda = \partial A'-0.19\tilde\lambda
\end{align*}
and by symmetry, we also obtain $\partial U\le\partial A'-0.19\tilde\lambda$. In other words, boundary size decreases by at least $0.19\lambda$ on both recursive instances. 
\end{proof}
\begin{observation}\label{o:2}
If the algorithm decomposes by case~(\ref{item:small-clusters-2-case-2})
into $U=C\cap A'$ and $A'\setminus C=A'\setminus U$, then $\partial U \le  \partial A'-0.4 \epsilon \tilde\lambda$ and
$\partial(A' \setminus U) \le  \partial A'-0.4 \epsilon \tilde\lambda$.
\end{observation}
\begin{proof}
If the algorithm decomposes by case~(\ref{item:small-clusters-2-case-2}), then since $U$ is $(1-\epsilon)$-boundary-sparse in $A'$,
\begin{align*}
\partial(A'\setminus U)&=\partial A'-w(U,\overline{A'})+w(U,A'\setminus U)
\\&\le \partial A' - \frac1{1-\epsilon}w(U,A'\setminus U) + w(U,A'\setminus U) \le \partial A'-\epsilon w(U,A'\setminus U)\le\partial A'-0.4\epsilon\tilde\lambda
\end{align*}
and by symmetry, $\partial U\le\partial A'-0.4\epsilon\tilde\lambda$.
\end{proof}
The two observations together show that boundary size decreases by at least $\min\{0.19\tilde\lambda,0.4\epsilon\tilde\lambda\}$ on both recursive instances.  

We now upper bound the number of clusters by tracking the following three parameters of a cluster $A'$:
 \begin{enumerate}
 \item The boundary size $\partial A'$, and
 \item The maximum cut size within $A'$ of a cut in $\mathcal C'$ after the pruning step, defined by
 \[ c(A')=\max_{\substack{C\in\mathcal C'\\\text{after}\\\text{pruning}}}\partial_{G[A']}(C\cap A') ,\]
 which we define as $0$ if $\mathcal C'=\emptyset$, and
 \item The cluster size $|A'|$.
 \end{enumerate}
We claim that the first two parameters are at most $3\tilde\lambda$ and $1.01\tilde\lambda$, respectively. The first parameter is always at most $3\tilde\lambda$ since originally $\partial A\le3\tilde\lambda$ by assumption, and we have shown that boundary size never increases upon recursion. The second parameter is at most $1.01\tilde\lambda$ originally, and it never increases upon recursion since the value $\partial_{G[A']}(C\cap A')$ decreases as $A'$ gets smaller upon recursion and $\mathcal C'$ can only lose elements over time.

Let $f(b,c,d)$ be the maximum number of clusters in the final decomposition starting from a cluster $A'$ with $\partial A'\le b$, $c(A')\le c$, and $|A'|\le d$.
Our goal is to show that
$$f(b,c,d) \le O\left( \big(\frac{b}{\epsilon \tilde \lambda}\big)^2 (\log d)^{O(b/\tilde\lambda)}\right).$$
We recursively bound $f(b,c,d)$ by stepping through each case of the algorithm. To do so we first make two observations.

\begin{observation}\label{o:3}\label{o:casea}
If the algorithm decomposes by case~(\ref{item:small-clusters-2-case-1})
into $U=C\cap A'$ and $A'\setminus C=A'\setminus U$, then 
the maximum number of clusters in the final decomposition of each
recursive instance $U,A'\setminus U$ is upper bounded by $f(b-0.19 \tilde \lambda, c, d)$, i.e., we have 
$f(b,c,d) = 2f(b-0.19\tilde\lambda,c,d).$
\end{observation}
\begin{proof}
This observation follows directly from Observation~\ref{o:1} as
it shows that 
$\partial U,\partial(A'\setminus U)\le \partial A'-0.19\tilde\lambda \le b-0.19\tilde\lambda$.
\end{proof}

\begin{observation}\label{obs:smaller-than-half}
If the algorithm decomposes by case~(\ref{item:small-clusters-2-case-2}), we must have $|C\cap A'|\le|A'|/2$. Moreover, all sets $C'\in\mathcal C'$ that cross $C\cap A'$ must also cross $A'\setminus C$.
\end{observation}
\begin{proof}
For the first statement, if $C\cap A'$ is $(1-\epsilon)$-boundary-sparse, then so is $A'\setminus C$, and since the algorithm minimizes for $|C\cap A'|$, we must have $|C\cap A'|\le|A'\setminus C|$, which means that $|C\cap A'|\le|A'|/2$.

For the second statement, suppose for contradiction that there exists $C'\in\mathcal C'$ that crosses $C\cap A'$ but not $A'\setminus C$. Without loss of generality, by replacing $C'$ with $\overline{C'}$ if needed, we can assume that $C'$ is disjoint from $A'\setminus C$.
But then $|C' \cap A'| < |C\cap A'|$, and by minimality, $C'$ should have been selected over $C$, a contradiction.
\end{proof}

\begin{observation}\label{o:caseb}
If the algorithm decomposes by case~(\ref{item:small-clusters-2-case-2}) we obtain the recursive bound
\begin{align*}
f(b,c,d) \le \max \{ 1,\
 2 f(b-0.19\tilde\lambda,c,d) + f(b,c,d/2),
\ f(b,0.6c,d) + f(b-0.4\epsilon\tilde\lambda,c,d) \} 
\end{align*}
\end{observation}
\begin{proof}
Suppose that the algorithm decomposes by case~(\ref{item:small-clusters-2-case-2}) into $U=C\cap A'$ and $A'\setminus C=A'\setminus U$. By \Cref{obs:smaller-than-half}, we have $|U|\le|A'|/2$ and moreover, all sets $C'\in\mathcal C'$ that cross $U$ must also cross $A'\setminus U$. We case on whether $c(U)\ge0.6 c(A')$, i.e., whether there exists $C'\in\mathcal C'$ that crosses $U$ with $\partial_{G[U]}C'\ge0.6 c(A')$.
 \begin{enumerate}
 \item If 
 such a set $C'$ exists, then from $\partial_{G[U]}C'+\partial_{G[A'\setminus U]}C'\le\partial_{G[A']}C'$ we obtain $\partial_{G[A'\setminus U]}C'\le\partial_{G[A']}C'-\partial_{G[U]}C'\le c(A')-0.6c(A')=0.4c(A')$. Since $C'$ also crosses $A'\setminus U$, the recursive instance $A'\setminus U$ must go through case~(\ref{item:small-clusters-2-case-1}); let $A_1,A_2$ be the decomposition of $A'\setminus U$. Then by \Cref{o:1}, $\partial A_1,\partial A_2 \le \partial(A'\setminus U)-0.19\tilde\lambda\le \partial A'-0.19\tilde\lambda$. By induction on clusters $U,A_1,A_2$, the total number of clusters at the end is at most
\[  2 f(b-0.19\tilde\lambda,c,d) + f(b,c,d/2) .\]
 \item Otherwise, we have $c(U)<0.6c(A')$. By \Cref{o:2}, we have $\partial(A'\setminus U)\le\partial A'-0.4\epsilon\tilde\lambda$. By induction, the total number of clusters at the end is at most
\[ f(b,0.6c,d) + f(b-0.4\epsilon\tilde\lambda,c,d) .\]
 \end{enumerate}
Overall, if the algorithm decomposes by case~(\ref{item:small-clusters-2-case-2}) we obtain the recursive bound
\begin{align*}
f(b,c,d) \le \max \{ & 1,
\\& 2 f(b-0.19\tilde\lambda,c,d) + f(b,c,d/2),
\\& f(b,0.6c,d) + f(b-0.4\epsilon\tilde\lambda,c,d) \}.
\end{align*}
\end{proof}
Next we state the base cases for our induction that bounds $f(b,c,d)$.
\begin{enumerate}
    \item Since we always have $\partial A'\ge \lambda>0.99\tilde\lambda$, we have the vacuous bound
$f(b,c,d)=0$ for $b\le0.99\tilde\lambda$.
    \item If $\mathcal C'=\emptyset$, no further recursion is necessary. Thus, $f(b,0,d) = 1$.
    \item If $|A'| = 1$, the cluster cannot be cut any further. Thus, $f(b,c,1) = 1$.
\end{enumerate}

We now bound $f(b,c,d)$ using a case analysis depending on the value of the second parameter. As shown by Observation~\ref{o:casea} and ~\ref{o:caseb}, the value of the second parameter never increases. Furthermore, if it drops to at most $0.4 \tilde \lambda$ on a set $A'$ then the algorithm always executes case~(\ref{item:small-clusters-2-case-1}) for all subsequent recursive calls on subsets of $A'$. Thus, we bound this setting first.

\emph{Case 1:} If $c(A')\le0.4\tilde\lambda$, then the algorithm either takes case~(\ref{item:small-clusters-2-case-1}) or does nothing (in the case $c(A')=0$), so
\[ f(b,c,d) \le \max \{ 1, 2f(b-0.19\tilde\lambda,c,d) \} \text{ for }c\le0.4\tilde\lambda .\]
Every recursive case still has $c(A')\le0.4\tilde\lambda$, so the algorithm only makes recursive calls through case~(\ref{item:small-clusters-2-case-1}). The number of recursive branches doubles for at most $\lceil\frac{b-0.99\tilde\lambda}{0.19\tilde\lambda}\rceil \le \frac{b-0.99\tilde\lambda}{0.19\tilde\lambda} + 1 \le \frac{b}{0.19\tilde\lambda} \le \frac{6b}{\tilde \lambda}$ levels of recursion and $2^x \ge 1$ for all positive $x$, so it holds that 
$f(b,c,d) \le  2^{6b/\tilde\lambda}$.


\emph{Case 2:} For $c\in(0.4\tilde\lambda,0.61\tilde\lambda]$, 
the algorithm either takes case~(\ref{item:small-clusters-2-case-2}) or
 case~(\ref{item:small-clusters-2-case-3}). 
We claim the bound
\[ f(b,c,d) \le \frac b{0.4\epsilon\tilde\lambda} (2+2\log_2d)^{6b/\tilde\lambda} .\]
It is trivially true for $b\le0.99\tilde\lambda$ and $d=1$.
Note that $0.6c\le0.6\cdot0.61\tilde\lambda\le0.4\tilde\lambda$ so that a call
on a set with second parameter $0.6c$ guarantees that case (\ref{item:small-clusters-2-case-1}) 
is executed on this call and we can apply the bound from Case 1.
Now by induction on $b\ge0.99\tilde\lambda$ and $d>1$ it holds that
\begin{alignat*}{2}
f(b,c,d) &\le \max \{ && 1,
\\& && 2 f(b-0.19\tilde\lambda,c,d) + f(b,c,d/2),
\\& && f(b,0.6c,d) + f(b-0.4\epsilon\tilde\lambda,c,d) \}
\\&\le \max\{ && 1,
\\& && 2 \cdot \frac b{0.4\epsilon\tilde\lambda} (2+2\log_2d)^{6(b-0.19\tilde\lambda)/\tilde\lambda} + \frac b{0.4\epsilon\tilde\lambda} (2+2\log_2(d/2))^{6b/\tilde\lambda},
\\& && 2^{6b/\tilde\lambda} + \frac{b-0.4\epsilon\tilde\lambda}{0.4\epsilon\tilde\lambda} (2+2\log_2d)^{6b/\tilde\lambda} \}
\\&\le \max\{ && 1,
\\& && 2 \cdot \frac b{0.4\epsilon\tilde\lambda} (2+2\log_2d)^{6b/\tilde\lambda-1} + \frac b{0.4\epsilon\tilde\lambda} (2+2\log_2d)^{6b/\tilde\lambda-1} \cdot (2+2\log_2 d - 2),
\\& && (2+2\log_2d)^{6b/\tilde\lambda} + \frac{b-0.4\epsilon\tilde\lambda}{0.4\epsilon\tilde\lambda} (2+2\log_2d)^{6b/\tilde\lambda} \}
\\& = && \frac b{0.4\epsilon\tilde\lambda}  (2+2\log_2d)^{6b/\tilde\lambda} .
\end{alignat*}

\emph{Case 3:} 
For $c\in[0.61\tilde\lambda,1.01\tilde\lambda]$, we claim the bound
\[ f(b,c,d) \le \left(\frac b{0.4\epsilon\tilde\lambda}\right)^2 (2+2\log_2d)^{6b/\tilde\lambda} .\]
It is trivially true for $b\le0.99\tilde\lambda$ and $d=1$.
Note that $0.6c\le0.6\cdot1.01\tilde\lambda\le0.61\tilde\lambda$, which implies that a call on a set with second parameter $0.6c$ guarantees that the bound from Case 2 can be applied to this call.
Now by induction on $b\ge0.99\tilde\lambda$ and $d>1$, it holds that
\begin{alignat*}{2}
f(b,c,d) &\le \max \{ && 1,
\\& && 2 f(b-0.19\tilde\lambda,c,d) + f(b,c,d/2),
\\& && f(b,0.6c,d) + f(b-0.4\epsilon\tilde\lambda,c,d) \}
\\&\le \max\{ && 1,
\\& && 2 \cdot \left(\frac b{0.4\epsilon\tilde\lambda}\right)^2 (2+2\log_2d)^{6(b-0.19\tilde\lambda)/\tilde\lambda} + \left(\frac b{0.4\epsilon\tilde\lambda}\right)^2 (2+2\log_2(d/2))^{6b/\tilde\lambda},
\\& && \frac b{0.4\epsilon\tilde\lambda} (2+2\log_2d)^{6b/\tilde\lambda} + \bigg(\frac{b-0.4\epsilon\tilde\lambda}{0.4\epsilon\tilde\lambda}\bigg)^2 (2+2\log_2d)^{6b/\tilde\lambda} \}
\\&\le \max\{ && 1,
\\& && 2 \cdot \left(\frac b{0.4\epsilon\tilde\lambda}\right)^2 (2+2\log_2d)^{6b/\tilde\lambda-1} + \left(\frac b{0.4\epsilon\tilde\lambda}\right)^2 (2+2\log_2d)^{6b/\tilde\lambda-1} \cdot (2+2\log_2d - 2),
\\& && \frac b{0.4\epsilon\tilde\lambda} (2+2\log_2d)^{6b/\tilde\lambda} + \bigg(\frac{b-0.4\epsilon\tilde\lambda}{0.4\epsilon\tilde\lambda}\bigg)\left(\frac b{0.4\epsilon\tilde\lambda}\right) (2+2\log_2d)^{6b/\tilde\lambda} \}
\\& = && \left(\frac b{0.4\epsilon\tilde\lambda}\right)^2  (2+2\log_2d)^{6b/\tilde\lambda} .
\end{alignat*}
Since $b\le3\tilde\lambda$, $c\le1.01\tilde\lambda$, and $d\le|A|$, we obtain the desired bound $O(\epsilon^{-2}(\log|A|)^{O(1)})$.
\end{proof}

Finally, we claim that \Cref{lem:small-clusters-1,lem:small-clusters-2} together imply \Cref{lem:small-clusters}. Applying \Cref{lem:small-clusters-1} to the initial cluster $A$ produces clusters $A'$ whose total boundary size is $O(\epsilon^{-1}\partial A)$. Each cluster $A'$ with $\partial A'>3\tilde\lambda$ has no $(1-\epsilon)$-boundary-sparse set $U\subseteq A$ with $\partial_{G[A]}A\le1.01\tilde\lambda$, so in particular, for every $1.01$-approximate minimum cut $S\subseteq V$ of $G$, the set $S\cap A'$ is non-$(1-\epsilon)$-boundary-sparse in $A'$.

For each cluster $A'$ with $\partial A'\le3\tilde\lambda$, \Cref{lem:small-clusters-2} produces $O(\epsilon^{-2}(\log|A'|)^{O(1)})$ many clusters, each of boundary size at most $\partial A'$, such that for every $1.01$-approximate minimum cut $S\subseteq V$ of $G$, there is a partition $\mathcal P$ of the decomposition such that for each part $P\in\mathcal P$, the set $S\cap\bigcup_{A'\in P}A'$ is non-$(1-\epsilon)$-boundary-sparse in $\bigcup_{A'\in P}A'$. 
Recall that the total boundary size of the output of \Cref{lem:small-clusters-1} is $O(\epsilon^{-1}\partial A)$.
It follows that the total boundary size of the outputs of \Cref{lem:small-clusters-2} over all $A'$ is $O(\epsilon^{-1}\partial A) \cdot O(\epsilon^{-2}(\log|A|)^{O(1)}) = O(\epsilon^{-3}(\log|A|)^{O(1)}\partial A)$.

Given a $1.01$-approximate minimum cut $S\subseteq V$ of $G$, we define the partition $\mathcal P$ so that every cluster $A'$ with $\partial A'>3\tilde\lambda$ produced by \Cref{lem:small-clusters-1} is a singleton part in $\mathcal P$, and for each cluster $A'$ that is further decomposed by \Cref{lem:small-clusters-2}, we add to $\mathcal P$ the partition guaranteed property~(\ref{item:small-clusters-2-1}) of \Cref{lem:small-clusters-2}.

By construction, for each part $P$ in the final partition $\mathcal P$, the set $S\cap\bigcup_{A'\in P}A'$ is non-$(1-\epsilon)$-boundary sparse in $\bigcup_{A'\in P}A'$. For the running time, the algorithm calls \Cref{lem:small-clusters-1} once and \Cref{lem:small-clusters-2} up to $|A|$ times, for a total running time of $\tilde{O}(|A|^9)$. This concludes the proof of \Cref{lem:small-clusters}.

\subsection{Large clusters}
\label{sec:bigpieces}

In this section, we present a fast algorithm to break down a large cluster $A$, i.e., a cluster with $\Omega(s_0)$ total volume, into a 
(potentially large) cluster $A_0$ and a (potentially empty) collection of small clusters $A_1, \dots, A_r$. The key property of cluster $A_0$ is that any near-minimum cut $U$ with small volume can be transformed, or ``uncrossed'', into a small enough cut $U^*$ with $U^*\cap A_0=\emptyset$. More formally, the main result of this section is the following.
\FinalCuts

We begin with a high-level description of the algorithm. In the rest of the section let $A$ be a large cluster.

\paragraph{Setting up the sources.}
To find cuts that are not uncrossable with $A$, the algorithm sets up a collection of flow problems. Let $U \subset A$ is a set with $\vol_G(U)\le s_0$ that is not uncrossable with $A$. More precisely, suppose that $w(U,V\setminus A) \ge (1+\epsilon)w(U,A\setminus U)$; we will use this definition of ``uncrossable'' for intuition in this section. Consider a flow problem in $G[A]$ where each vertex $v\in U$ has source $w(v,V\setminus A)$, each edge has flow capacity equal to its weight, and each vertex has sink capacity equal to $\eta d^W_G(v)$ for small enough $\eta>0$. Within $U$, there is $w(U,V\setminus A)$ source originating from $U$, at most $w(U,A\setminus U)$ flow leaves through the edges $\partial_{G[A]}U$, and at most $\eta\vol^W _G(U)\le\eta s_0$ flow that is absorbed in $U$. Since $w(U,V\setminus A) \ge (1+\epsilon)w(U,A\setminus U)$, the flow cannot be feasible as long as $\eta>0$ is small enough. It turns out that to make this argument work, we must set $\eta>0$ small enough that there is not enough sink to absorb all the source. Therefore, we split the source up into many sub-sources, each of them small enough to be fully absorbed. More formally, we construct a collection of subsets $S\subseteq A$ that are small---$O(s_0)$ volume each---such that for each $U$ with (i) $\partial_{G[A]}U\lesssim\tilde\lambda$ and (ii) $\vol^W (U)\le s_0$ that is (iii) not uncrossable, there exists a set $S$ containing $U$. For each set $S$, we construct a flow where only vertices $v\in S$ receive their $w(v,V\setminus A)$ of source (with edge capacities and sinks unchanged). This is not quite achievable for technical purposes, but we show that it is possible \emph{if the induced graph $G[U]$ has minimum cut $\Omega(\lambda)$}. In this case, if we run a variant of tree packing on $G[A]$, then there is a lot of ``room'' in $G[U]$ to pack trees, and the vertices in $U$ will be contained in the union of $O(1)$ many trees, even after a post-processing step where each tree is broken down into pieces of volume $O(s_0)$. We then stitch the trees together by connecting pieces together in a bounded fashion so that one of the resulting sets $S$ contains $U$.

\paragraph{Executing the flows.}
Consider a set $U$ with $\partial_{G[A]}U\lesssim\tilde\lambda$ and $\vol^W (U)\le s_0$ that is not uncrossable. In the previous step we have constructed a ``source set'' $S\subseteq A$ with $\vol^W (S)\le O(s_0)$ and $U \subseteq S$. We cannot afford to run a separate flow where each vertex of $S$ is a source, for each ``source set'' $S$, since each flow computation may take linear time. (Note that if the graph $G[A]$ were unweighted, then we can resort to local flow algorithms, but these do not have local guarantees in weighted graphs.)

Instead, we set up instances of \emph{approximate isolating cuts} which allow us to execute the flows of many sources at once. We partition the collection of sources into $\textup{polylog}(n)$ subcollections that are vertex-disjoint (which is a requirement for approximate isolating cuts) and solve an instance of the approximate isolating cuts problem for each subcollection, obtaining a cut $C_i\subseteq A$ for each source $S_i$ such that if $U\subseteq S_i$ then either $U\subseteq C_i$ or $U\setminus C_i$ is now uncrossable in $G[A\setminus C_i]$. We then remove all cuts $C_i$ from $A$, and continue to the instance of the next subcollection where we restrict to the remaining sources $S\cap A$ on the remaining induced graph $G[A]$. 

\paragraph{Removing the connected assumption.}
In general, $G[U]$ may not be well-connected. (It may even be disconnected, or very loosely connected.) The algorithm simply handles this issue by cycling through the approximate isolating cut problem instances some $\textup{polylog}(n)$ times. The analysis is a lot more complicated and technical, and we have to break down a general set $U$ into well-connected sets and analyze the progress the algorithm makes in each cycle.

\subsubsection{Setting up the sources}

Let $W=\tilde\lambda^3/(50s_0^2)$. Let $A \subseteq V$. The algorithm first constructs the following unweighted multi-graph $H$ supported on the same vertices and edges as $G[A]$: for each edge $e$ in $G[A]$, let $w_H(e)=\lfloor w_G(e)/W\rfloor$, i.e., there are $\lfloor w_G(e)/W\rfloor$ parallel edges between the endpoints of $e$ in $H$.

\paragraph{Well-connected case.}

We first treat the ``well-connected case'' established by the following lemma, whose proof is the main focus of this subsection. By well-connected, we refer to the condition in property~(\ref{item:well-connected-3}) below that $G[U']$ has minimum cut at least $0.1\tilde\lambda$.
\begin{lemma}\label{lem:well-connected}
Let $A\subseteq V$ be a cluster. There is an algorithm that runs in $\tilde{O}((s_0/\tilde\lambda)^{O(1)}|E(G[A])|)$ time and computes a collection $\mathcal X$ of subsets of $A$ such that
 \begin{enumerate}
 \item $\vol^W (X)\le3s_0$ for each $X\in\mathcal X$.\label{item:well-connected-1}
 \item Each vertex is in $O((s_0/\tilde\lambda)^5\ln m)$ many sets in $\mathcal X$. In particular, $|\mathcal X|\le O((s_0/\tilde\lambda)^5\ln m)|A|$.\label{item:well-connected-2}
 \item Consider a set $U\subseteq A$ such that $\partial_{G[A]}U\le1.01\tilde\lambda$ and $\vol^W (U)\le s_0$. Let $U'\subseteq U$ and suppose that the minimum cut in $G[U']$ is at least $0.1\tilde\lambda$. There exist a family of at most $2525$ subsets in $\mathcal X$ whose union contains $U'$.\label{item:well-connected-3}
 \end{enumerate}
\end{lemma}

We prove \Cref{lem:well-connected} with the help of the following lemmata.

\begin{lemma}\label{lem:induced-subgraph-mincut}
Consider a set $U\subseteq A$ such that $\vol^W (U)\le s_0$ and the minimum cut in $G[U]$ is at least $0.1\tilde\lambda$. Then, the minimum cut in $H[U]$ is at least $\tilde\lambda/(25W)=s_0^2/(5\tilde\lambda^2)$.
\end{lemma}
\begin{proof}
Let $S\subseteq U$ be the minimum cut in $H[U]$. Since $|U|\le\vol^W (U)/\delta\le\vol^W (U)/\lambda\le1.01\vol^W (U)/\tilde\lambda$, there are at most $(|U|/2)^2\le(1.01\vol^W (U)/(2\tilde\lambda))^2$ many edges between $S$ and $U\setminus S$ in $H[U]$. For each such edge $e$, we have $w_G(e)\le w_H(e)\cdot W+W$. Summing over all edges $e$ between $S$ and $U\setminus S$ gives
\begin{align*}
w_G(S,U\setminus S)&\le w_H(S,U\setminus S)\cdot W+\left(\frac{1.01\vol^W (U)}{2\tilde\lambda}\right)^2\cdot W
\\&\le w_H(S,U\setminus S)\cdot W+\left(\frac{1.01s_0}{2\tilde\lambda}\right)^2\cdot\frac{\tilde\lambda^3}{5s_0^2}
\\&\le w_H(S,U\setminus S)\cdot W+0.06\tilde\lambda .
\end{align*}
Rearranging and using that $w_G(S,U\setminus S)\ge0.1\tilde\lambda$ gives $w_H(S,U\setminus S)\ge\tilde\lambda/(25W)=s_0^2/(5\tilde\lambda^2)$, as promised.
\end{proof}

We will need the following two lemmas, whose technical proofs are deferred to \Cref{sec:deferred-proofs}. First, we run a variant of tree packing on $G[A]$ where we pack forests instead, since $G[A]$ may not be well-connected enough to admit a large tree packing.
\begin{lemma}\label{lem:forest-packing}
Let $\kappa\ge1$ be a parameter. There is an algorithm that, given $\kappa$, runs in $O(\kappa|E(H)|\log^2m)$ time and computes $\lceil\kappa\ln m\rceil$ forests in $H$ such that
 \begin{enumerate}
 \item Each edge $e$ in $H$ participates in at most $100\ln m$ forests.\label{item:forest-packing-1}
 \item For any induced subgraph $H[U]$ with minimum cut at least $\kappa$, all vertices belong to a single connected component for each of the forests.\label{item:forest-packing-2}
 \end{enumerate}
\end{lemma}

Then, the lemma below decomposes each forest in our packing into a collection of subtrees whose vertex sets have small volume.
\begin{lemma}\label{lem:partition-tree}
Consider a tree $T$ supported on a subset of vertices in $A$. We can compute subsets $V_1,\ldots,V_\ell\subseteq V(T)$ satisfying the following.
 \begin{enumerate}
 \item $\vol^W (V_i)\le3s_0$ for all $i$.\label{item:partition-tree-1}
 \item Each vertex $v\in A$ is in at most $d_T(v)$ many subsets $V_i$.\label{item:partition-tree-2}
 \item Consider any (not necessarily connected) subset $U\subseteq V(T)$ with $\vol^W (U)\le s_0$ and let $k=\partial_TU$ be the number of edges in $T$ with exactly one endpoint in $U$. Then, there are at most $k$ subsets $V_i$ intersecting $U$, and their union contains $U$.\label{item:partition-tree-3}
 \end{enumerate}
The algorithm runs in $\tilde O(|V(T)|)$ time.
\end{lemma}

The algorithm applies \Cref{lem:forest-packing} to $\kappa=\tilde\lambda/(25W)=s_0^2/(5\tilde\lambda^2)$, producing a collection $\mathcal F$ of $\lceil\tilde\lambda/(25W)\cdot\ln m\rceil$ forests. For each forest $F\in\mathcal F$, apply \Cref{lem:partition-tree} to each connected component (tree) of $F$. Take the union of the output $\{V_1,\ldots,V_\ell\}$ over all connected components, and let this collection of vertex subsets be $\mathcal V(F)$.

Our proof of \Cref{lem:well-connected} proceeds as follows: 
We set $\mathcal X$ as the union of $\mathcal V(F)$ over all $F\in\mathcal F$.
\begin{itemize}
    \item 
Property~(\ref{item:well-connected-1}) follows from property~(\ref{item:partition-tree-1}) of \Cref{lem:partition-tree}. 
\item For property~(\ref{item:well-connected-2}), observe that there are $\lceil\tilde\lambda/(25W)\cdot\ln m\rceil = \lceil s_0^2/(5\tilde\lambda^2)\cdot\ln m\rceil$ forests $F\in\mathcal F$. For each $F\in\mathcal F$, by property~(\ref{item:partition-tree-2}) of \Cref{lem:partition-tree}, each vertex $v\in A$ is in at most $d_F(v)$ many subsets in $\mathcal V(F)$, and by property~(\ref{item:partition-tree-1}), each vertex $v\in A$ that joins some subset in $\mathcal V(F)$ must have $d^W_G(v)\le3s_0$. 
Since $F$ is a spanning forest of $H$, we have $d_F(v)\leq_H(v) \le d^W_G(v)/W \le 3s_0/W \le O((s_0/\tilde\lambda)^3)$. So in total, each vertex $v\in A$ is in at most $O((s_0/\tilde\lambda)^3)$ subsets in $\mathcal V(F)$ for each of $O((s_0/\tilde\lambda)^2\ln m)$ many forests $F\in\mathcal F$, concluding property~(\ref{item:well-connected-2}).
\end{itemize}
So all that we are left to show is property~(\ref{item:well-connected-3}). For that we use the following lemma, which concludes the proof of \Cref{lem:well-connected}.

\begin{lemma}\label{lem:vertex-subsets-intersecting-U}
Consider a set $U\subseteq A$ such that $\partial_{G[A]}U\le1.01\tilde\lambda$ and $\vol^W (U)\le s_0$. Let $U'\subseteq U$ and suppose that the minimum cut in $G[U']$ is at least $0.1\tilde\lambda$. There exists a forest $F$ such that there are at most $2525$ many subsets $V'\in\mathcal V(F)$ that intersect $U'$, and their union contains $U'$.
\end{lemma}
\begin{proof}
By construction of $H$, we have $\partial_HU\le\partial_{G[A]}U/W\le1.01\tilde\lambda/W$. By property~(\ref{item:forest-packing-1}) of \Cref{lem:forest-packing}, each of the at most $1.01\tilde\lambda/W$ edges in $\partial_HU$ participates in at most $100\ln m$ forests. Since there are $\lceil\tilde\lambda/(25W)\cdot\ln m\rceil$ many forests, there is a forest $F\in\mathcal F$ with
\[ \partial_FU \le \frac{1.01\tilde\lambda/W\cdot100\ln m}{\tilde\lambda/(25W)\cdot\ln m} = 2525. \]
By \Cref{lem:induced-subgraph-mincut}, the minimum cut in $H[U']$ is at least $\kappa=\tilde\lambda/(25W)$, so property~(\ref{item:forest-packing-2}) of \Cref{lem:forest-packing} guarantees that all vertices of $U'$ belong to a single tree $T$ of $F$.
Then, property~(\ref{item:partition-tree-3}) of \Cref{lem:partition-tree} applied to subset $U\cap V(T)$ guarantees at most $\partial_TU\le\partial_FU\le2525$ subsets $V_i\in\mathcal V(F)$ intersecting $U\cap V(T)$, and their union contains $U\cap V(T)$ which contains $U'$.
\end{proof}
To finish the proof of \Cref{lem:well-connected}, it remains to bound the running time. Since $\kappa=O(s_0/\tilde\lambda)^{O(1)}$ and $|E(H)|\le|E(G[A])|$, \Cref{lem:forest-packing} runs in $\tilde O((s_0/\tilde\lambda)^{O(1)}|E(G[A])|)$ time. There are $O(s_0/\tilde\lambda)^{O(1)}\log m$ many forests, and \Cref{lem:partition-tree} runs in $\tilde O(|V(F)|)$ time total over all the trees of each forest $F\in\mathcal F$. It follows that the overall running time is $\tilde O((s_0/\tilde\lambda)^{O(1)}|E(G[A])|)$.

\paragraph{General case.}

We show the following lemmata which relate the general case to the well-connected case. We remark that the partition specified below is not computed by the algorithm; we only need its existence to prove correctness. 

\begin{lemma}\label{lem:general-case}
Consider a set $U\subseteq A$ such that $\partial_{G[A]}U\le1.01\tilde\lambda$ and $\vol^W (U)\le s_0$. There exists an integer $k$ and a partition $U_1,\ldots,U_k$ of $U$ such that
 \begin{enumerate}
 \item For each $i\in[k]$, the minimum cut of $G[U_i]$ is at least $0.1\tilde\lambda$.\label{item:general-case-1}
 \item For at least $k/2-5$ many indices $i\in[k]$, the boundary $\partial_{G[A]}U_i$ is at most $0.4\tilde\lambda$. \label{item:general-case-2}
 \item The total weight $w(U_1,\ldots,U_k)$ of inter-cluster edges is at most $0.1(k-1)\tilde\lambda$.\label{item:general-case-3}
 \end{enumerate}
More generally, for any $A'\subseteq A$, if we consider the subset $I\subseteq[k]$ of indices $i\in[k]$ with $U_i\cap A'\ne\emptyset$, then
 \begin{enumerate}
 \item[2'.] For at least $|I|/2-5$ many indices $i\in I$, the boundary $\partial_{G[A']}(U_i\cap A')$ is at most $0.4\tilde\lambda$. \label{item:general-case-2'}
 \end{enumerate}
\end{lemma}

\begin{proof}
Iteratively partition the vertex set $U$ as follows: start with the collection $\mathcal P=\{U\}$, and while there exists a vertex set $U'\in\mathcal P$ such that $G[U']$ has a cut of weight less than $0.1\tilde\lambda$, replace $U'$ in $\mathcal P$ with the two sides $U'_1,U'_2\subseteq U'$ of the cut. Let $\mathcal P=\{U_1,\ldots,U_k\}$ be the final partition. Property~(\ref{item:general-case-1}) holds by construction. The algorithm makes $k-1$ many cuts of weight at most $0.1\tilde\lambda$ each, so the total weight of inter-cluster edges is at most $(k-1)\cdot0.1\tilde\lambda$, proving property~(\ref{item:general-case-3}).

We now show property~(\ref{item:general-case-2}). Suppose for contradiction that fewer 
 than $k/2-5$ many indices $i\in[k]$ have $\partial_{G[A]}U_i\le0.4\tilde\lambda$. Equivalently, more than $k/2+5$ many indices $i\in[k]$ have $\partial_{G[A]}U_i>0.4\tilde\lambda$. Consider the partition $U_1,\ldots,U_k,A\setminus U$ of $A$; the total weight $w(U_1,\ldots,U_k,A\setminus U)=\frac12(\sum_i\partial_{G[A]}U_i+\partial_{G[A]}U)$ of inter-cluster edges is at least $\frac12\cdot(k/2+5)\cdot0.4\tilde\lambda $. On the other hand, $w(U_1,\ldots,U_k,A\setminus U) = w(U_1,\ldots,U_k)+\partial_{G[A]}U \le 0.1(k-1)\tilde\lambda+1.01\tilde\lambda$ by property~(\ref{item:general-case-3}) and the assumption $\partial_{G[A]}U\le1.01\tilde\lambda$. Altogether, 
\[ w(U_1,\ldots,U_k,A\setminus U) \le 0.1(k-1)\tilde\lambda+1.01\tilde\lambda < \frac12\cdot(k/2+5)\cdot0.4\tilde\lambda  \le w(U_1,\ldots,U_k,A\setminus U) ,\]
a contradiction.

Finally, to show property~(2'), consider any $A'\subseteq A$. Consider the sequence of cuts of weight less than $0.1\tilde\lambda$ that were made to construct $\mathcal P=\{U_1,\ldots,U_k\}$. Repeat this sequence of cuts on $G[A']$ as follows: for the next cut $(C,\bar C)$ of $G[U']$ for some $U'\subseteq A$, if $C\cap A'=\emptyset$ or $\bar C\cap A'=\emptyset$ then do nothing; otherwise make the cut $(C\cap A',\bar C\cap A')$ and split $U'$ into $C\cap A'$ and $\bar C\cap A'$. It is clear that the new output $\mathcal P'$ is $\{U_i\cap A':i\in I\}$. In total, there are $|I|-1$ cuts that split a set into two, so the total weight in $G[A']$ of inter-cluster edges of $\mathcal P'$ is at most $(|I|-1)\cdot0.1\tilde\lambda$. The rest of property~(2') follows similarly from the proof of property~(\ref{item:general-case-2}): suppose for contradiction that more than $|I|/2+5$ many indices in $I$ have $\partial_{G[A']}U_i>0.4\tilde\lambda$; then writing $\mathcal P'=U_{i_1}\cap A',\ldots,U_{i_{|I|}}\cap A'$, we obtain
\begin{align*}
w_{G[A']}(U_{i_1}\cap A',\ldots,U_{i_{|I|}}\cap A',A'\setminus U)&= w_{G[A']}(U_{i_1}\cap A',\ldots,U_{i_{|I|}}\cap A')+\partial_{G[A']}(U\cap A')
\\&\le w_{G[A']}(U_{i_1}\cap A',\ldots,U_{i_{|I|}}\cap A')+\partial_{G[A]}U
\\&\le0.1|I|\tilde\lambda+1.01\tilde\lambda
\\&<\frac12\cdot(|I|/2+5)\cdot0.4\tilde\lambda+\frac12\cdot1.01\tilde\lambda
\\&\le w_{G[A']}(U_{i_1}\cap A',\ldots,U_{i_{|I|}}\cap A',A'\setminus U),
\end{align*}
a contradiction.
\end{proof}

\paragraph{Stitching pieces together.} 
Let $U'$ be defined as in \Cref{lem:well-connected}. Property~(\ref{item:well-connected-3}) of \Cref{lem:well-connected} guarantees that there are at most $2525$ sets in $\mathcal X$ whose union contains a set $U'$ satisfying the conditions. In the lemma below, we stitch together sets in $\mathcal X$ so that there is now a single set containing $U'$, as long as the induced subgraph $G[U']$ is moderately connected: instead of a minimum cut of $\Omega(\tilde\lambda)$, we can allow $\Omega(\gamma\tilde\lambda)$ for any $\gamma=1/\textup{polylog}(n)$.
\begin{lemma}\label{lem:construct-S}
Let $A\subseteq V$ be a cluster and let $\gamma,K$ be parameters. There is an algorithm that runs in $\tilde{O}((\frac{s_0\ln m}{\tilde\lambda\gamma})^{O(K)}|E(G[A])|)$ time and computes a collection $\mathcal S$ of subsets of $A$ such that
 \begin{enumerate}
 \item $\vol^W (S)\le O(Ks_0)$ for each $S\in\mathcal S$.\label{item:construct-S-1}
 \item Each vertex is in $O(\frac{s_0\ln m}{\tilde\lambda\gamma})^{O(K)}$ many sets in $\mathcal S$. In particular, $|\mathcal S|\le O(\frac{s_0\ln m}{\tilde\lambda\gamma})^{O(K)}|A|$. \label{item:construct-S-2}
 \item Consider a set $U\subseteq A$ such that $\partial_{G[A]}U\le1.01\tilde\lambda$ and $\vol^W (U)\le s_0$, and consider a subset $U'\subseteq U$. Suppose that the minimum cut in $G[U']$ is at least $\gamma\tilde\lambda$, and suppose that there exist (not necessarily disjoint) $U_1,\ldots,U_k\subseteq U$ with $k\le K$ such that $U_1\cup\cdots\cup U_k\supseteq U'$ and each $G[U_i]$ has minimum cut at least $0.1\tilde\lambda$. Then, there exists a set $S\in\mathcal S$ containing $U'$.\label{item:construct-S-3}
 \end{enumerate}
\end{lemma}
\begin{proof}
The algorithm begins with the set $\mathcal X$ from \Cref{lem:well-connected}. Construct an auxiliary graph $H$ whose vertices are the sets $X\in\mathcal X$. Add an edge between $X,X'\in\mathcal X$ if $X\cap X'\ne\emptyset$ or there exists an edge between $X$ and $X'$ of weight at least $3\gamma\tilde\lambda^3/s_0^2$. This completes the construction of $H$.

We show that the degree of $H$ is $O((s_0/\tilde\lambda)^8(\ln m)/\gamma)$. For a given $X\in\mathcal X$, since $\vol^W (X)\le3s_0$ by property~(\ref{item:well-connected-1}) of \Cref{lem:well-connected}, and since each vertex $v\in X$ has $d^W(v)\ge\lambda$, we have $|X|\le3s_0/\lambda\le O(s_0/\tilde\lambda)$. By property~(\ref{item:well-connected-2}), each vertex $v\in X$ is in $O((s_0/\tilde\lambda)^5\ln m)$ many sets $X\in\mathcal X$, so the total number of edges $(X,X')$ added incident to $X$ because $X\cap X'\ne\emptyset$ is $O((s_0/\tilde\lambda)^6\ln m)$. Next, we bound the number of edges $(X,X')$ added because of an edge between $X$ and $X'$ of weight at least $3\gamma\tilde\lambda^3/s_0^2$. Note that from $\vol^W (X)\le3s_0$ we obtain that there are at most $3s_0/(3\gamma\tilde\lambda^3/s_0^2) = (s_0/\tilde\lambda)^3/\gamma$ many edges of weight at least $3\gamma\tilde\lambda^3/s_0^2$ incident to $X$. For each of these edges, there are $O((s_0/\tilde\lambda)^5\ln m)$ many sets $X'\in\mathcal X$ containing the other endpoint, for a total of $O((s_0/\tilde\lambda)^8(\ln m)/\gamma)$ many edges $(X,X')$.

The algorithm considers all connected subgraphs of at most $2525K$ vertices in $H$. For each connected subgraph with vertices $X_1,\ldots,X_\ell$, add $X_1\cup\cdots\cup X_\ell$ to $\mathcal S$. Since $\vol^W (X_i)\le3s_0$ for each $i\in[\ell]$, property~(\ref{item:construct-S-1}) holds by construction. By the degree bound on $H$, each vertex $X$ in $H$ is in $O((s_0/\tilde\lambda)^8(\ln m)/\gamma)^{O(K)}$ many connected subgraphs (since we can identify each one by a spanning tree of size $O(K)$). Each vertex is in at most $O((s_0/\tilde\lambda)^5\ln m)$ many sets $X\in\mathcal X$, each of which shows up $O((s_0/\tilde\lambda)^8(\ln m)/\gamma)^{O(K)}$ times in the construction of $\mathcal S$, fulfilling property~(\ref{item:construct-S-2}).

For the rest of the proof, we prove property~(\ref{item:construct-S-3}). For each $i\in[k]$, since $U_i\subseteq U$ and $G[U_i]$ has minimum cut at least $0.1\tilde\lambda$, property~(\ref{item:well-connected-3}) of \Cref{lem:well-connected} guarantees at most $2525$ sets $X\in\mathcal X$ whose union contains $U_i$. In total, over all $i\in[k]$, there are at most $2525K$ many subsets $X\in\mathcal X$ whose union contains $U_1\cup\cdots\cup U_k\supseteq U'$. It remains to show that these sets, when viewed as vertices in $H$, form a connected component in $H$.

Let $Y$ be this set of vertices in $H$. Suppose for contradiction that there is a bipartition $\{X_1,\ldots,X_\ell\}$, $\{X'_1,\ldots,X'_{\ell'}\}$ of $Y$ such that there is no edge in $H$ between any pair of vertices $X_i,X'_j$. Then, $Z=X_1\cup\cdots\cup X_\ell$ and $Z'=X'_1\cup\cdots\cup X'_{\ell'}$ are disjoint vertex sets in $A$ since a common vertex would imply an edge between the bipartition in $H$. In particular, since $Z\cup Z'$ contains $U'$, we have that $Z\cap U'$ and $Z'\cap U'$ partition $U'$. Since $\vol^W (U')\le\vol^W (U)\le s_0$, we have $|U'|\le s_0/\lambda$, so there are at most $(|U'|/2)^2\le s_0^2/(4\lambda^2)$ many edges between $Z\cap U'$ and $Z'\cap U'$. Since the minimum cut in $G[U']$ is at least $\gamma\tilde\lambda$, the total weight of edges between $Z\cap U'$ and $Z'\cap U'$ in $G$ is at least $\gamma\tilde\lambda$. It follows that at least one edge between $Z\cap U'$ and $Z'\cap U'$ has weight at least $\gamma\tilde\lambda/(s_0^2/(4\lambda^2))=4\gamma\tilde\lambda\lambda^2/s_0^2\ge3\gamma\tilde\lambda^3/s_0^2$. By construction of the edges of $H$, there must be an edge between some $X_i,X'_j$, a contradiction.

Finally, we bound the running time. By property~(\ref{item:well-connected-2}) of \Cref{lem:well-connected}, each vertex appears at most $O((s_0/\tilde\lambda)^5\ln m)$ times in $\mathcal X$, so in the construction of $H$, it is visited $O((s_0/\tilde\lambda)^5\ln m)$ times in total. Also, $|\mathcal X|\le O((s_0/\tilde\lambda)^5\ln m) \cdot |A|$. Therefore, constructing $H$ takes time $O((s_0/\tilde\lambda)^5\ln m)|E(G[A])|$. For each vertex in $H$, there are $O((s_0/\tilde\lambda)^8(\ln m)/\gamma)^{O(K)}$ many connected components with at most $2525K$ vertices, and these can be enumerated in the same asymptotic running time. For each connected subgraph with vertices $X_1,\ldots,X_\ell$, we have $|X_1\cup\cdots\cup X_\ell|\le\sum_i|X_i|\le\sum_i\vol(X_i)/\lambda\le\sum_i3s_0/\lambda\le\sum_iO(s_0/\tilde\lambda)$, so it takes $O(Ks_0/\tilde\lambda)$ time to add the set $X_1\cup\cdots\cup X_\ell$ to $\mathcal S$. Overall, the running time is bounded by $\tilde{O}((\frac{s_0\ln m}{\tilde\lambda\gamma})^{O(K)}|E(G[A])|)$.
\end{proof}

\subsubsection{Executing the flows}

With the sources established by \Cref{lem:construct-S}, we now set up and run the flow instances.
We need the following algorithm which is based on the \emph{approximate isolating cuts} algorithm of \cite{li2023near}. We defer the details to \Cref{sec:deferred-proofs}. 
Given a large cluster $A$, the algorithm repeatedly removes subsets that form small clusters from $A$. In this subsection we use $A'$ to denote the current version of $A$, which potentially has some small clusters already removed.

Given a cluster $A'$ and a family of vertex-disjoints sets of $\mathcal S$ (where $\mathcal S$ was constructed for $A'$), the algorithm in the next lemma constructs a family of subsets $C_i$ of $A'$ such that three properties hold, which are crucial to prove \Cref{lem:final-cuts}. Intuitively, Property~(\ref{item:isolating-cuts-3}) below says that the cut $U^\dag$ can be uncrossed to $U^*$ that is contained within $C_i$, but we are also penalized for boundary edges in the difference $U^\dag\setminus U^*$; see Figure~\ref{fig:isolating-cuts}. Note that we do not guarantee that $U^*$ equals $U^\dag\cap C_i$ as standard uncrossing may suggest; we only guarantee that is it a subset.

\begin{figure}
    \begin{center}
        \includegraphics{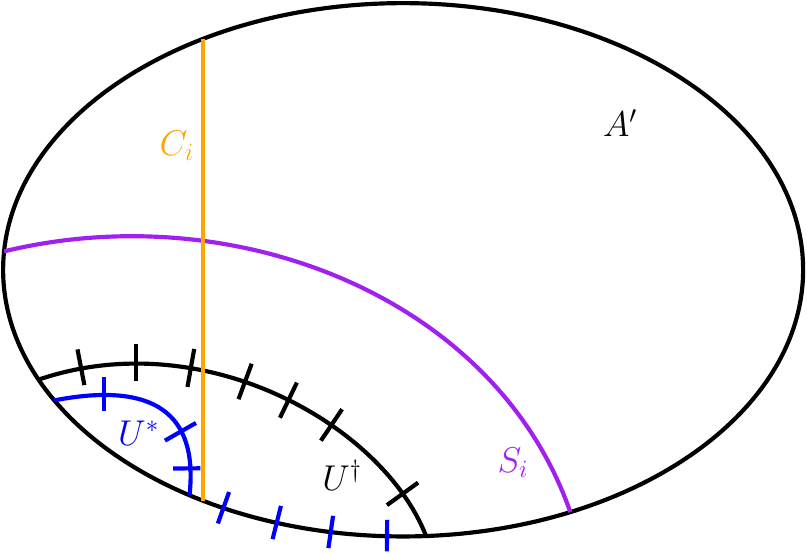}
        \caption{The set $U^\dag$ is uncrossed to $U^*\subseteq U^\dag\cap C_i$ according to property~(\ref{item:isolating-cuts-3}) of \Cref{lem:isolating-cuts}. The blue edges have weight at most $(1+\epsilon)$ times the black edges.}\label{fig:isolating-cuts}
    \end{center}
\end{figure}

\begin{lemma}\label{lem:isolating-cuts}
Let $A'\subseteq V$ be a cluster, let $\epsilon\le0.1$ be a parameter, and let $S_1,\ldots,S_k$ be disjoint vertex subsets of $A'$. There is an algorithm that computes, for each $i\in[k]$, a set $C_i\subseteq A'$ such that
 \begin{enumerate}
 \item $\vol^W (C_i) \le (\frac{10 s_0}{\epsilon \tilde \lambda})\, \vol^W (S_i)$.\label{item:isolating-cuts-1}
 \item $\partial_{G[A']}C_i \le (1-\epsilon)w(C_i,V\setminus A')$.\label{item:isolating-cuts-2}
 \item For any set $U^\dag\subseteq S_i$ with $\vol^W (U^\dag)\le s_0$, there exists $U^*\subseteq U^\dag\cap C_i$ such that $w(V\setminus A',U^\dag\setminus U^*) +\partial_{G[A']}U^*\le(1+\epsilon)\partial_{G[A']}U^\dag$. Moreover, if $U^*\ne U^\dag$ then $\partial_{G[A']}U^\dag\ge0.2\lambda$.
\label{item:isolating-cuts-3}
 \end{enumerate}
The sets $C_i$ are vertex-disjoint, and the algorithm runs in time $\tilde{O}((\frac{s_0}{\epsilon\tilde\lambda})^{O(1)}|E(G[A'])|)$ (over all $i\in[k]$).
\end{lemma}
We now prove the main lemma of this section, restated below. In the lemma, we decompose a cluster $A$ into sets $A_i$ (that result from the $C_i$ of the lemma above) and the ``leftover'' piece, called $A_0$.
The first property of the lemma below guarantees that each set $A_i$ with $i\ge 1$ is  small and, thus, the algorithm for small clusters can be used to uncross all approximate minimum cuts in it. 
The second property guarantees that the boundary $\partial_{G[A']}C_i$ is a factor of $(1- \epsilon)$ smaller than the piece of the boundary $\partial A'$ that is cut out when $C_i$ is removed from $A'$. This property can be used to bound the number of new inter-cluster edges that are created when $C_i$ is removed from $A'$.
Think of $U$ in the third property
as being an approximately minimum cut, which is then ``uncrossed'' with $A_0$: the new $U^*$ no longer crosses the leftover piece $A_0$ and its cut size $\partial_{G[A]}U^*$ is only a factor $(1+\epsilon)$ larger. For technical reasons, we also have to account for the total weight $w(U\setminus U^*,V\setminus A)$ of ``external edges'' lost moving from $U$ to $U^*$. 
\FinalCuts*
For the rest of the section, we describe and analyze the algorithm for \Cref{lem:final-cuts}, called $\mathcal A$. 

The algorithm $\mathcal{A}$ consists of the following steps:
\begin{enumerate}
\item Apply \Cref{lem:construct-S} with $\gamma=\frac{\epsilon\tilde\lambda}{100s_0}$ and $K=10$ to obtain the set $\mathcal S$. 
\item Partition $\mathcal S$ into subsets $\mathcal S_1,\ldots,\mathcal S_\ell$, where $\ell\le(\frac{s_0\ln m}{\epsilon\tilde\lambda})^{O(1)}$ and the sets in each $\mathcal S_i$ are vertex-disjoint. This can be done by partitioning the sets in $\mathcal S$ into independent sets of the graph $H$ in the proof of \Cref{lem:construct-S}, which has degree $O((s_0/\tilde\lambda)^8(\ln m)/\gamma)\le (\frac{s_0\ln m}{\epsilon\tilde\lambda})^{O(1)}$.

\item Set $A'\gets A$ and $\mathcal C\gets\emptyset$. 
\item Repeat $O(\log|A|)$ times: 
\begin{enumerate}
    \item For every set $\mathcal S_j$ in $\mathcal S_1,\ldots,\mathcal S_\ell$ execute the following steps:
\begin{enumerate}
    \item For each $S_i \in \mathcal S_j$ run  the algorithm in \Cref{lem:isolating-cuts} with parameter $\epsilon/2$, the current cluster $A'$ and the sets $S_i\cap A'$ (discarding any sets that are empty), which computes for every set $S_i \cap A'$ with $S \in\mathcal S_j$  a set $C_i$
    \item Add all sets $C_i$  as a new elements to $\mathcal C$
    \item Set  $A'\gets A'\setminus\cup_iC_i$.
\end{enumerate}
\item Set $A_0$ as the final $A'$, and rename the sets in  $\mathcal C$ to be $\{A_1,\ldots,A_r\}$.
\end{enumerate}
\end{enumerate}
This concludes the description of the algorithm, which runs in time $\tilde{O}((\frac{s_0}{\epsilon\tilde\lambda})^{O(1)}|E(G[A])|)$ because \Cref{lem:construct-S} and \Cref{lem:isolating-cuts} run in the same asymptotic time and the latter is called $\ell\log|A|\le(\frac{s_0\ln m}{\epsilon\tilde\lambda})^{O(1)}$ times.

We first prove Properties~(\ref{item:final-cuts-1}) and (\ref{item:final-cuts-2}). Property~(\ref{item:final-cuts-1}) follows from Property~(\ref{item:isolating-cuts-1}) of \Cref{lem:isolating-cuts} and property~(\ref{item:construct-S-1}) of \Cref{lem:construct-S}. For Property~(\ref{item:final-cuts-2}), we track the value of $\partial A'$ throughout the algorithm. Every time a collection of sets $C_i$ are removed from $A'$, the value $\partial A'$ loses $\sum_iw(C_i,V\setminus A')$ and gains at most $\sum_i\partial_{G[A']}C_i$. For each $i$, we have $\partial_{G[A']}C_i\le(1-\epsilon/2)w(C_i,V\setminus A')$ by property~(\ref{item:isolating-cuts-2}) of \Cref{lem:isolating-cuts}, so the net loss of $\partial A'$ is at least $(\epsilon/2)\sum_iw(C_i,V\setminus A') \ge (\epsilon/2)\sum_i\partial_{G[A']}C_i$. We charge the new edges cut, which have total weight at most $\sum_i\partial_{G[A']}C_i$, to the decrease in $\partial A'$. It follows that the total weight of inter-cluster edges cannot exceed $O(\epsilon^{-1})$ times the original value of $\partial A'$, which is $\partial A$.

For the remainder of the proof, we focus on property~(\ref{item:final-cuts-3}). Fix a set $U\subseteq A$ with $\partial_{G[A]}U\le1.001\tilde\lambda$ and $\vol^W (U)\le s_0$. Let $U_1,\ldots, U_k$ be a partition the fulfills the conditions of
\Cref{lem:general-case}. Note that each $U_i$ is a subset $U'$ of $A$ that satisfies the conditions in property~(\ref{item:construct-S-3}) of \Cref{lem:construct-S} with  $\gamma=0.1$ and $k=1$.
Thus, \Cref{lem:construct-S} implies that each $
U_i$ is contained in some set that belongs to the collection $\mathcal S$ computed by the above algorithm in the first step.

Let $I$ be the subset of indices $i\in[k]$ with $U_i\cap A'\ne\emptyset$ (for the current set $A'$), which may lose elements over time as $A'$ shrinks.

\begin{lemma}
On each cycle through $\mathcal S_1,\ldots,\mathcal S_\ell$, the set $I$ loses at least $|I|/2-5$ elements.
\end{lemma}
\begin{proof}
By property~(2') of \Cref{lem:general-case}, there are at least $|I|/2-5$ indices $i\in I$ such that $\partial_{G[A']}(U_i\cap A')\le0.4\tilde\lambda$. We show that after another cycle through $\mathcal S_1,\ldots,\mathcal S_\ell$, each of these indices is removed from $I$. Let $i\in I$ be an index with $\partial_{G[A']}(U_i\cap A')\le0.4\tilde\lambda$. Recall that $U_i$ is contained in some set $S_i$, which belongs to some $\mathcal S_j$. As the algorithm goes through $\mathcal S_1,\ldots,\mathcal S_{j-1}$, the set $A'$ can only shrink, so we still have $U_i\cap A'\ne\emptyset$ and $\partial_{G[A']}(U_i\cap A')\le0.4\tilde\lambda$ right before processing $\mathcal S_j$. 

To avoid clutter, define $U'_i=U_i\cap A'$ and $S'_i=S_i\cap A'$. Recall that on iteration $\mathcal S_j$, the algorithm runs \Cref{lem:isolating-cuts} on the sets $S\cap A'$ for $S\in\mathcal S_j$. In particular, we have $U'_i=U_i\cap A'\subseteq S_i\cap A'=S'_i$ and $\vol^W (U'_i) \le s_0$, so $U'_i$ satisfies the conditions of property~(\ref{item:isolating-cuts-3}) of \Cref{lem:isolating-cuts} (for $U^\dag\gets U'_i$ and $S_i\gets S'_i$). Let $C_i\subseteq A'$ be the corresponding set for $S_i\gets S'_i$, and let $U_i^*$ be the set $U^*$ guaranteed by property~(\ref{item:isolating-cuts-3}) for $U^\dag\gets U'_i$, so that $U^*_i\subseteq U'_i\cap C_i$. If $U^*_i=U'_i$, then $U_i\cap A'=U^*_i\subseteq C_i$ and $C_i$ is removed from $A'$ on iteration $\mathcal S_j$, which means index $i$ is removed from $I$, a contradiction. Therefore, $U^*_i\subsetneq U'_i$.

Consider the quantity $\partial_G(U'_i\setminus U^*_i)$. All edges involved are either in $\partial_{G[A']}(U'_i\setminus U^*_i)$ or $w(U'_i\setminus U^*_i,V\setminus A')$. The former expression can be further divided: the edges in $\partial_{G[A']}(U'_i\setminus U^*_i)$ are either in $\partial_{G[A']}U'_i$ or $\partial_{G[A']}U^*_i$ (or both). Applying the conditions of property~(\ref{item:isolating-cuts-3}), we obtain
\[ \lambda \le \partial_G(U'_i\setminus U^*_i) \le \partial_{G[A']}U'_i+\partial_{G[A']}U^*_i+w(U'_i\setminus U^*_i,V\setminus A') \le (2+\epsilon)\partial_{G[A']}U'_i .\]
It follows that $\partial_{G[A']}U'_i>0.4\tilde\lambda$, a contradiction. It follows that index $i$ leaves set $I$ after processing $\mathcal S_j$.
\end{proof}

As long as $|I|\ge11$, we have $|I|/2-5>0$, so the set $I$ loses a constant fraction of its vertices each cycle. It follows that after $O(\log|A|)$ cycles, we have $|I|\le10$.

At this point, let $A^\dag $ be the current value of $A'$, and consider $U\cap A^\dag $ in $G[A^\dag ]$, which also satisfies $\partial_{G[A^\dag]}(U\cap A^\dag)\le\partial_{G[A]}U\le1.001\tilde\lambda$ and $\vol(U\cap A^\dag)\le s_0$. We define a refinement of this partition as follows (these are not steps executed by our algorithm, as $U$ is unknown, but are just used for refining $\mathcal P$):
Iteratively partition the vertex set $U\cap A^\dag $ as follows: start with the collection $\mathcal P=\{U\cap A^\dag \}$, and while there exists a vertex set $U'\in\mathcal P$ such that $G[U']$ has a cut of weight less than $\gamma\tilde\lambda=\frac{\epsilon\tilde\lambda}{100s_0}\tilde\lambda$, replace $U'$ in $\mathcal P$ with the two sides $U'_1,U'_2\subseteq U'$ of the cut. Let $\mathcal P$ be the final partition of $U\cap A^\dag $.

\begin{lemma}\label{lem:last-cycle}
After one more cycle through $\mathcal S_1,\ldots,\mathcal S_\ell$ (and setting $A_0\gets A'$ at the end), for each $U'\in\mathcal P$, there exist $A^*\supseteq A_0$ and $U^*\subseteq U'\cap A^*\setminus A_0$ such that $w(V\setminus A^*,(U'\cap A^*)\setminus U^*)+\partial_{G[A^*]}U^*\le(1+\epsilon/2)\partial_{G[A^*]}(U'\cap A^*)$. Moreover, if $U^*\ne U'\cap A^*$ then $\partial_{G[A^*]}(U'\cap A^*)\ge0.2\lambda$. 
\end{lemma}
\begin{proof}
In the construction of $\mathcal P$, we make at most $|U\cap A^\dag |\le|U|\le s_0/\lambda \le 2s_0/\tilde\lambda$ many cuts of weight at most $\gamma\tilde\lambda=\frac{\epsilon\tilde\lambda}{100s_0}\tilde\lambda$, so each $U'\in\mathcal P$ at the end has boundary
\[ \partial_{G[A^\dag ]}U'\le\partial_{G[A^\dag ]}(U\cap A^\dag )+2s_0/\tilde\lambda\cdot\gamma\tilde\lambda \le 1.001\tilde\lambda + 0.02\epsilon\tilde\lambda \le 1.01\tilde\lambda .\]
By construction, $G[U']$ has minimum cut at least $\gamma\tilde\lambda$, so applying property~(\ref{item:isolating-cuts-3}) of \Cref{lem:construct-S} on $U'\subseteq U$ and the (at most $10$) sets $U_i$ for $i\in I$ guarantees a set $S'\in\mathcal S$ containing $U'$. Let $\mathcal S_j$ be the set containing $S'$.

As the algorithm $\mathcal A$ goes through $\mathcal S_1,\ldots,\mathcal S_{j-1}$, the set $A'$ can only shrink, so we still have $\partial_{G[A']}(U'\cap A')\le1.01\tilde\lambda$ right before processing $\mathcal S_j$. Recall that on iteration $\mathcal S_j$, the algorithm runs \Cref{lem:isolating-cuts} on the sets $S\cap A'$ for $S\in\mathcal S_j$. In particular, we have $U'\cap A'\subseteq S'\cap A'$, so the conditions of property~(\ref{item:isolating-cuts-3}) of \Cref{lem:isolating-cuts} are satisfied (for $U^\dag\gets U'\cap A'$ and $S_i\gets S'\cap A'$). We obtain a set $U^*\subseteq U'\cap A'\cap C_i$ with $w(V\setminus A',(U'\cap A')\setminus U^*)+\partial_{G[A']}U^*\le(1+\epsilon/2)\partial_{G[A']}(U'\cap A')$.
We set $A^*=A'$ for the current set $A' $ (which might be a superset of the final set $A'$) which implies  that the inequality in \Cref{lem:last-cycle} holds. 

If $U^*\ne U'\cap A^*$ then $\partial_{G[A^*]}(U'\cap A^*)\ge0.2\lambda$ by the additional guarantee of property~(\ref{item:isolating-cuts-3}) of \Cref{lem:isolating-cuts}. At the end of the cycle, we have $A^*\supseteq A'=A_0$ because $A'$ can only lose elements over time, and $U^*\subseteq U'\cap A^*\setminus A_0$ because $U^*\subseteq U'\cap A^*\cap C_i$ and $A_0\cap C_i=\emptyset$ (since $C_i$ is removed from $A'$ on iteration $\mathcal S_j$).
\end{proof}

\begin{figure}
\begin{center}
\includegraphics{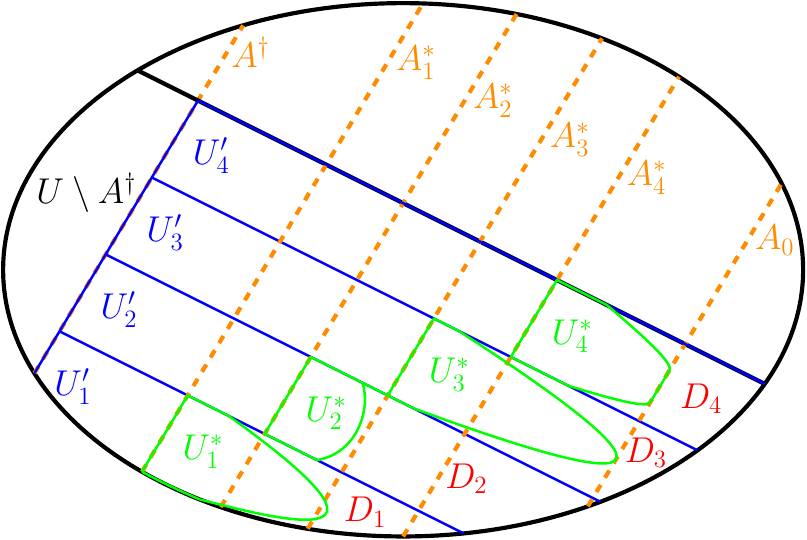}
\caption{The variables defined for the remainder of the proof. The black ellipse represents the cluster $A$ and the bottom half (below the black diagonal) is the set $U$.}\label{fig:executing-the-flows}
\end{center}
\end{figure}

For the remainder of the proof, we encourage the reader to use Figure~\ref{fig:executing-the-flows} as a reference.
Suppose the algorithm $\mathcal A$ terminates after the cycle analyzed by \Cref{lem:last-cycle}. Let us denote the sets in $\mathcal P$ by $U'_1,\ldots,U'_k$. \Cref{lem:last-cycle} applied to each $U'_i$ shows the existence of a set $A^*_i\supseteq A_0$ and a set $U^*_i\subseteq U'_i\cap A^*_i\setminus A_0$. To avoid clutter, define $D_i=(U'_i\cap A^*_i)\setminus U^*_i$. We finally set $U^*=U\setminus\cup_iD_i$.

Recall that $A$ is a cluster and we need to show a claim about a set $U \subseteq A$.
Observe that $U$ can be decomposed as $U=(U\setminus A^\dag )\cup(\cup_iU'_i)$, and $U^*$ can be decomposed as $U^*=(U\setminus A^\dag )\cup(\cup_i(U'_i\setminus D_i))$. We write
\begin{align}
w(D_i,U^*) &= w(D_i, (U\setminus A^\dag )\cup(\cup_j(U'_j\setminus D_j)) ) \nonumber
\\&= w(D_i,U\setminus A^\dag ) + \sum_{j\in[k]}w(D_i,U'_j\setminus D_j) \nonumber
\\&= w(D_i,U\setminus A^\dag ) + w(D_i,U'_i\setminus D_i) + \sum_{j\ne i}w(D_i,U'_j\setminus D_j) \nonumber
\\&= w(D_i,(U\setminus A^\dag )\cup(U'_i\setminus D_i)) + \sum_{j\ne i}w(D_i,U'_j\setminus D_j) \nonumber
\\&\le w(D_i,(A\setminus A^*_i)\cup(U'_i\cap A^*_i\setminus D_i)) + \sum_{j\ne i}w(U'_i,U'_j) \nonumber
\\&=w(D_i,A\setminus A^*_i) + w(D_i,U'_i\cap A^*_i\setminus D_i) + \sum_{j\ne i}w(U'_i,U'_j) \nonumber
\\&=w(D_i,A\setminus A^*_i) + w(D_i,U_i^*) + \sum_{j\ne i}w(U'_i,U'_j), \label{eq:medium-cuts-0}
\end{align}
where the last equation holds as $U'_i\cap A^*_i\setminus D_i = U_i^*$.
and
\begin{align}
w(D_i,A\setminus U) &\ge w(D_i,A^*_i\setminus U) \label{eq:medium-cuts-1}
\\&= w(D_i,A^*_i\setminus \cup_jU'_j) & \text{since }A^*_i\subseteq A^\dag  \nonumber
\\&= w(D_i,(A^*_i\setminus U'_i)\setminus \cup_{j\ne i}U'_j) \nonumber
\\&\ge w(D_i,A^*_i\setminus U'_i) - \sum_{j\ne i}w(D_i,U'_j) \nonumber
\\&\ge w(D_i,A^*_i\setminus U'_i) - \sum_{j\ne i}w(U'_i,U'_j). \label{eq:medium-cuts-2}
\end{align}
To avoid clutter, define $D=\cup_iD_i$. Note that $U^* = U \setminus D$. By examining the difference in edges, observe that $\partial_{G[A]}(U^*)-\partial_{G[A]}U = w(D,U^*) - w(D,A\setminus U)$. Using Equations~(\ref{eq:medium-cuts-0}) and (\ref{eq:medium-cuts-2}), we obtain
\begin{align}
\partial_{G[A]}U^*-\partial_{G[A]}U &= \partial_{G[A]}(U^*)-\partial_{G[A]}U \nonumber
\\&= w(D,U^*) - w(D,A\setminus U) \nonumber
\\&= \sum_{i\in[k]}\big( w((D_i,U^*) - w((D_i,A\setminus U) \big) \nonumber
\\&\stackrel{\mathclap{(\ref{eq:medium-cuts-0}),(\ref{eq:medium-cuts-2})}}\le \sum_{i\in[k]}\bigg(\big( w(D_i,A\setminus A^*_i) + w(D_i,U_i^*) + \sum_{j\ne i}w(U'_i,U'_j) \big) \nonumber
\\& \qquad\qquad \; -\big( w(D_i,A^*_i\setminus U'_i) - \sum_{j\ne i}w(U'_i,U'_j) \big)\bigg) \nonumber
\\&= \sum_{i\in[k]}\big( w(D_i,A\setminus A^*_i) + w(D_i,U_i^*) - w(D_i,A^*_i\setminus U'_i) \big) \label{eq:medium-cuts-4}
\\& \quad \ + 2\sum_{\substack{i,j\in[k]\\i\ne j}}w(U'_i,U'_j).\nonumber
\end{align}
In the following note that
$w(A^*_i\setminus U^*_i,U^*_i) =  \partial_{G[A^*_i]}U^*_i
$ and
$w(U'_i\cap A^*_i,A^*_i\setminus U'_i) = \partial_{G[A^*_i]}(U'_i\cap A^*_i)$.
For the expression $w(D_i,A\setminus A^*_i) + w(D_i,U_i^*) - w(D_i,A^*_i\setminus U'_i)$, we add $w(U_i^*,A^*_i\setminus U'_i)$ to the second and third terms so that they become
\begin{align*}
&w(D_i,A\setminus A^*_i) + w(D_i\cup (A^*_i\setminus U'_i),U_i^*) - w(D_i\cup U_i^*,A^*_i\setminus U'_i)
\\={}&w(D_i,A\setminus A^*_i) + w(A^*_i\setminus U^*_i,U_i^*) - w(U'_i\cap A^*_i,A^*_i\setminus U'_i)
\\={}&w(D_i,A\setminus A^*_i) + w(A^*_i\setminus U^*_i,U^*_i) - w(U'_i\cap A^*_i,A^*_i\setminus U'_i)
\\={}&w(D_i,A\setminus A^*_i) + \partial_{G[A^*_i]}U^*_i - \partial_{G[A^*_i]}(U'_i\cap A^*_i)
\\={}&w(D_i,V\setminus A^*_i)-w(D_i,V\setminus A) + \partial_{G[A^*_i]}U^*_i - \partial_{G[A^*_i]}(U'_i\cap A^*_i)
\\ \le{}&(1+ \epsilon/2) \partial_{G[A_i^*]}(U_i' \cap A_i^*) - \partial_{G[A_i^*]} U_i^* -w(D_i,V\setminus A) + \partial_{G[A^*_i]}U^*_i - \partial_{G[A^*_i]}(U'_i\cap A^*_i)
\\={}&\epsilon/2\partial_{G[A^*_i]}(U'_i\cap A^*_i)-w(D_i,V\setminus A)
,
\end{align*}
where the last inequality follow from \Cref{lem:last-cycle}.

Continuing from Equation~(\ref{eq:medium-cuts-4}), we obtain
the following bound. Notat that the second inequality below
comes from the fact that $U'_j$ can be larger than
$A^*_i\cap U'_j$ in general.
\begin{align}
\partial_{G[A]}U^*-\partial_{G[A]}U &\le \sum_{i\in[k]} \big(\epsilon/2 \cdot \partial_{G[A^*_i]}(U'_i\cap A^*_i)-w(D_i,V\setminus A)\big)+ 2\sum_{i,j:i\ne j}w(U'_i,U'_j) \nonumber
\\&= \big(\sum_{i\in[k]}\epsilon/2 \cdot \partial_{G[A^*_i]}(U'_i\cap A^*_i) \big)-w(D,V\setminus A) +2 \sum_{i,j:i\ne j}w(U'_i,U'_j) \nonumber
\\&= \sum_{i\in[k]}\epsilon/2 \cdot w(U'_i\cap A^*_i,A^*_i\setminus U'_i)-w(D,V\setminus A) +2 \sum_{i,j:i\ne j}w(U'_i,U'_j) \nonumber
\\&\le \sum_{i\in[k]}\epsilon/2\bigg( w(U'_i\cap A^*_i,A^*_i\setminus U) + \sum_{j\ne i}w(U'_i\cap A^*_i,U'_j)  \bigg)-w(D,V\setminus A)\nonumber
\\& \quad \ + 2 \sum_{i,j:i\ne j}w(U'_i,U'_j) \nonumber
\\&\le\sum_{i\in[k]}\epsilon/2 \cdot w(U'_i\cap A^*_i,A^*_i\setminus U)-w(D,V\setminus A) + (2+\epsilon/2) \sum_{i,j:i\ne j}w(U'_i,U'_j). \label{eq:medium-cuts-5}
\end{align}
We bound $\sum_{i\in[k]}\epsilon/2 \cdot  w(U'_i\cap A^*_i,A^*_i\setminus U)$ as
\begin{align}
\sum_{i\in[k]}\epsilon/2 \cdot  w(U'_i\cap A^*_i,A^*_i\setminus U) \le \sum_{i\in[k]}\epsilon/2 \cdot  w(U'_i\cap A^\dag,A^\dag\setminus U) \nonumber
&= \epsilon/2 \cdot  w(\cup_i(U'_i\cap A^\dag),A^\dag\setminus U) \nonumber
\\&= \epsilon/2 \cdot  w(U\cap A^\dag,A^\dag\setminus U) \nonumber
\\&= \epsilon/2 \cdot \partial_{G[A^\dag]}(U\cap A^\dag) \nonumber
\\&\le \epsilon/2 \cdot  \partial_{G[A]}U . \label{eq:medium-cuts-6}
\end{align}
To bound $\sum_{i,j:i\ne j}w(U'_i,U'_j)$, recall that the construction of $\mathcal P=\{U'_1,\ldots,U'_k\}$ was formed by at most $|U\cap A^\dag|\le|U|\le s_0/\lambda \le 2s_0/\tilde\lambda$ many cuts, each of weight less than $\gamma\tilde\lambda=\frac{\epsilon\tilde\lambda}{100s_0}\tilde\lambda$, so $\sum_{i,j:i\ne j}w(U'_i,U'_j) = w(U'_1,\ldots,U'_k)\le2s_0/\tilde\lambda\cdot\gamma\tilde\lambda\le0.02\epsilon\tilde\lambda$. Continuing from Equation~(\ref{eq:medium-cuts-5}), we obtain
\[ \partial_{G[A]}U^*-\partial_{G[A]}U \le \epsilon/2 \cdot \partial_{G[A]}U-w(D,V\setminus A)+0.05\epsilon\tilde\lambda .\]
Note that we need to show that
\[ \partial_{G[A]}U^*-\partial_{G[A]}U \le \epsilon\partial_{G[A]}U-w(D,V\setminus A).\]

If $\partial_{G[A]}U\ge0.1\lambda$ then property~(\ref{item:final-cuts-3}) follows, since $D=U\setminus U^*$ and $\epsilon/2 \cdot \partial_{G[A]}U+0.05\epsilon\tilde\lambda \le \epsilon\partial_{G[A]}U$. 

If $\partial_{G[A]}U<0.1\lambda$ then we actually claim that $U^*=U$. If not, then there exists $i\in[k]$ such that $U^*_i\ne U'_i\cap A^*_i$. By \Cref{lem:last-cycle}, we must have $\partial_{G[A^*_i]}(U'_i\cap A^*_i)\ge0.2\lambda$. Repeating the chain of inequalities from Equation~(\ref{eq:medium-cuts-6}) for the last inequality below,
\begin{align*}
0.2\lambda \le \partial_{G[A^*_i]}(U'_i\cap A^*_i) = 
w(U'_i\cap A^*_i,A^*_i\setminus U'_i) &\le w(U'_i\cap A^*_i,A^*_i\setminus U) + \sum_{j\ne i}w(U'_i\cap A^*_i,U'_j) \\&\le  w(U'_i\cap A^*_i,A^*_i\setminus U)+0.02\epsilon\tilde\lambda \stackrel{(\ref{eq:medium-cuts-6})}\le \partial_{G[A]}U+0.02\epsilon\tilde\lambda ,
\end{align*}
contradicting the assumption that $\partial_{G[A]}U<0.1\lambda$. It follows that $U^*=U$. Note that $U^* \cap A_0 = \emptyset$ follows by construction and $\partial_{G[A]}U^*-\partial_{G[A]}U \le \epsilon\partial_{G[A]}U-w(D,V\setminus A)$ hold trivially. Thus
 property~(\ref{item:final-cuts-3}) is fulfilled.

Therefore, we have ensured property~(\ref{item:final-cuts-3}) follows.
This concludes the proof of \Cref{lem:final-cuts}.

\subsubsection{Deferred Proofs}\label{sec:deferred-proofs}

\begin{proof}[Proof of \Cref{lem:forest-packing}]
The algorithm is a greedy minimum spanning tree packing algorithm. Initialize each edge length 
$\ell_e$ in $H[U]$ to be $1$. For iteration $i$ from $1$ to $\lceil\kappa\ln m\rceil$, compute a minimum spanning tree $T_i$ with respect to lengths $\ell_e$, and then remove all edges in $T_i$ of length more than $e^5m^6$. Add the resulting forest $F_i$ to the collection and then increase the length $\ell_e$ of each edge $e$ in $F_i$ by factor $1.1$.

We first prove property~(\ref{item:forest-packing-1}). For each edge $e$ in $H$, its length $\ell_e$ increases by factor $1.1$ every time $e$ joins a forest $F_i$, so the edge participates in at most $\log_{1.1}e^5m^6 \le 100\ln m$ forests for $m$ large enough, as promised.

To prove property~(\ref{item:forest-packing-2}), fix $U$ such that $H[U]$ has minimum cut at least $\kappa$. It suffices to show that over $\lceil\kappa\ln m\rceil$ iterations of the algorithm, no edge in $H[U]$ has length $\ell_e$ more than $e^5m^6$. In that case, each minimum spanning tree $T_i$ must connect the vertices in $U$ with edges of length $\ell_e\le e^5m^6$, and these edges are not deleted in $F_i$, i.e., all nodes of $u$ are connected in each forest $F_i$.

To show this claim, we interpret the algorithm in the fractional packing setting of Young~\cite{young2002randomized} and follow the proof of Lemma~6.2 in~\cite{young2002randomized}. Consider an induced subgraph $H[U]$ with minimum cut at least $\kappa$. For each $i$ from $1$ to $\lceil\kappa\ln m\rceil$, consider the forest $T_i\cap H[U]$. It is not hard to see, by properties of the minimum spanning tree, that $T_i\cap H[U]$ is a subset of some minimum spanning tree $T'_i$ of $H[U]$ with respect to the current lengths $\ell$. In other words, on each iteration, within the scope of $H[U]$, the algorithm increases by factor $1.1$ the lengths of a subset of the edges of a minimum spanning tree with the current lengths $\ell$. We show that over $\lceil\kappa\ln m\rceil$ iterations, no edge in $H[U]$ can have length more than $e^5m^6$.

Define $P\subseteq\mathbb R^{E(H[U])}$ as the spanning tree polytope on $H[U]$, i.e., the convex hull of the indicator vectors of the spanning trees of $H[U]$. Define $\Delta\subseteq\mathbb R^{E(H[U])}$ as the simplex on $E(H[U])$, i.e., the convex hull of the unit vectors. By duality, we have
\[ \min_{x\in P}\max_{y\in\Delta}\,\langle x,y\rangle = \max_{y\in\Delta}\min_{x\in P}\,\langle x,y\rangle = \max_{\substack{y\ge\mathbf0\\y\ne\mathbf0}}\frac{\min_{x\in P}\sum_ex_ey_e}{\sum_ey_e}  .\]
By a theorem of Nash-Williams~\cite{nash1961edge}, there is an (integral) spanning tree packing of any graph of value at least half the minimum cut of the graph. By scaling down the packing and the fact that $H[U]$ has minimum cut at least $\kappa$, we obtain
\[  \min_{x\in P}\max_{y\in\Delta}\,\langle x,y\rangle \le \frac1{\kappa/2} = \frac2\kappa .\]
For each iteration $i$, by construction, the quantity $\sum_{e\in H[U]}\ell_e$ increases by factor $\sum_{e\in T_i\cap H[U]}0.1\ell_e$. Since $T_i\cap H[U] \subseteq T'_i$, this is at most $\sum_{e\in T'_i}0.1\ell_e$. Since the spanning tree polytope $P$ is integral, we have $\sum_{e\in T'_i}\ell_e=\min_{x\in P}\sum_ex_e\ell_e$. Putting everything together, the increase in $\sum_{e\in H[U]}\ell_e$ is
\[ \sum_{e\in T_i\cap H[U]}0.1\ell_e \le \sum_{e\in T'_i}0.1\ell_e = 0.1 \min_{x\in P}\sum_{e\in H[U]}x_e\ell_e \le 0.1\cdot (\min_{x\in P}\max_{y\in\Delta}\,\langle x,y\rangle)\sum_{e\in H[U]}\ell_e \le 0.1\cdot\frac2\kappa\sum_{e\in H[U]}\ell_e .\]
After $\lceil\kappa\ln m\rceil$ iterations, the quantity $\sum_{e\in H[U]}\ell_e$ is at most
\[ m \cdot \left(1+\frac5\kappa\right)^{\lceil\kappa\ln m\rceil} \le m \cdot e^{(5/\kappa)(\kappa\ln m+1)} \le m \cdot e^{5\ln m+5}= e^5m^6 .\]
It follows that $\ell_e\le\sum_{e\in H[U]}\ell_e\le e^5m^6$ for each edge $e\in H[U]$, as claimed.

\end{proof}

\begin{proof}[Proof of \Cref{lem:partition-tree}]
First, remove from $T$ all vertices $v\in V(H)=A$ with $d^W_G(v)>s_0$. Then, consider each connected component (tree) $T'$ in the remaining forest. If $\vol^W _G(V(T'))\le 3s_0$, we add $V(T')$ to the output and proceed to the next connected component. If $\vol^W _G(V(T'))>3s_0$, we will partition the edges of $T'$ into edge-disjoint trees $T_1,\ldots,T_\ell$ such that $\vol^W (V(T_i))\in( s_0,3s_0]$, and then add the sets $V(T_i)$ to the output. 

We now describe the partitioning process. Root the tree $T'$ at an arbitrary vertex and while the remaining tree has (vertex) volume more than $3s_0$, do the following. Find the lowest (i.e., farthest from the root) node $v$ whose subtree rooted at $v$ has volume more than $s_0$. By construction, the subtrees rooted at the children of $v$ have volume at most $s_0$. We now have two cases.
 \begin{enumerate}
 \item If the subtrees rooted at the children of $v$ have total volume  at most $2s_0$, then since $d^W_G(v)\le s_0$, the total volume of the subtree rooted at $v$ is at most $3s_0$, and we declare the subtree rooted at $v$ as one of the trees $T_i$.
 \item Otherwise, take a subset of these subtrees whose total volume is in $(s_0,2s_0]$. Connect $v$ to these subtrees, declare the result as one of the output trees $T_i$.
 \end{enumerate}
In both cases, we have $\vol^W (V(T_i))\in(s_0,3s_0]$ by construction. We then remove $T_i$ from $T'$ (as well as the newly isolated vertices) and proceed to the next iteration. Each iteration decreases $\vol^W (V(T'))$ by at most $2s_0$ until $\vol^W (V(T_i))\le3s_0$, which means $\vol^W (V(T_i))>s_0$ and we finish by taking $T'$ as the final $T_i$.

Property~(\ref{item:partition-tree-1}) holds by construction. For property~(\ref{item:partition-tree-2}), since the trees $T_i$ are edge-disjoint, each vertex $v\in A$ can belong to at most $d_{T'}(v)\leq_T(v)$ many sets $V(T_i)$. For the rest of the proof, we focus on property~(\ref{item:partition-tree-3}). Since $\vol^W (U)\le s_0$, we cannot remove any vertices from $U$ on the very first step, so $U$ is contained in a connected component $T'$. If $\vol^W _G(V(T'))\le 3 s_0$, then the entire $V(T')$ is added and property~(\ref{item:partition-tree-3}) holds. Otherwise, the edges of $T'$ are partitioned into edge-disjoint trees $T_1,\ldots,T_\ell$. Since $\vol^W (V(T_i))>s_0$ for all $i$, we cannot have $V(T_i)\subseteq U$ for any $i$. So any $V(T_i)$ that intersects $U$ must contain an edge with exactly one endpoint in $U$. Since there are $k$ such edges, there are at most $k$ many $V(T_i)$ that intersect $U$, concluding property~(\ref{item:partition-tree-3}).

This algorithm can be easily implemented in $\tilde O(|V(T)|)$ time using dynamic tree data structures that update the volume of subtrees after each vertex deletion.
\end{proof}

\begin{proof}[Proof of \Cref{lem:isolating-cuts}.]
We use the following \emph{approximate isolating cuts} algorithm of \cite{li2023near}. The theorem below resembles Theorem~6.2 of~\cite{li2023near} with an additional property that follows directly from the algorithm itself. We also relabel a few of the parameters to avoid overloading variables.

 \begin{theorem}\label{thm:isolating-cuts-source}
 Fix any $\epsilon'<1$. Given an undirected graph $H=(V,E)$ and a set of terminal vertices $T\subseteq V$, there is an algorithm that outputs, for each $t\in T$, a $(1+\epsilon')$-approximate minimum cut $V_t$ separating $t$ from $T\setminus t$, i.e., $V_t\cap T=\{t\}$. Moreover, the sets $V_t$ are disjoint, and for each $t\in T$ and a set $U$ separating $t$ from $T\setminus t$, there exists $U'\subseteq U\cap V_t$ such that $\partial U'\le(1+\epsilon')\partial U$. Finally, for each $t\in T$, there exists a feasible flow such that
 \begin{enumerate}
 \item The flow originates at $t$ and is not absorbed by any nodes in $V_t$. More precisely, vertex $t$ has positive source and zero sink, and each vertex $v\in V_t\setminus t$ has zero source and zero sink.\label{item:isolating-cuts-source-1}
 \item Each edge $e\in E(V_t,V\setminus V_t)$ sends at least $(1-\epsilon')w(e)$ flow in the direction from $V_t$ to $V\setminus V_t$.\label{item:isolating-cuts-source-2}
 \end{enumerate}

Note that in $V\setminus V_t$, the flow can do anything as long as it remains feasible.

Let $\phi$ be the \emph{conductance} of $H$, defined as $\min_{\emptyset\subsetneq S\subsetneq V}\partial S/\min\{\vol^W (S),\vol^W (\bar S)\}$. Then, the algorithm takes deterministic time $O(m/(\epsilon\phi)^{O(1)})$.
 \end{theorem}
\begin{proof}[Proof (Sketch)]
We refer to Algorithm~1 in Section~6 of~\cite{li2023near}. The algorithm begins with a collection of $O(\log|T|)$ bipartitions of $T$ such that every pair of terminals in $T$ is separated by at least one bipartition. For each bipartition $X_i,Y_i$ of $T$, the algorithm computes a $(1+\Theta(\epsilon/\log|T|))$-\emph{fair cut} $S_i$ between $X_i$ and $Y_i$. The definition of fair cuts from~\cite{li2023near} is rather technical so we skip it for simplicity. Instead, we focus on its key \emph{approximate uncrossing} property formalized by Lemma~1.2 of~\cite{li2023near}. Let $V_t^*$ be the minimum cut separating $t$ from $T\setminus t$, fix an index $i\le O(\log|T|)$, and suppose without loss of generality that $S_i$ is the side of the $(X_i,Y_i)$-cut containing $t$. We can ``uncross'' $V_t^*$ with the cut $S_i$ so that the new cut $V_t^*\cap S_i$ only increases by factor $(1+\Theta(\epsilon/\log|T|))$ compared to $V_t^*$. We uncross $V_t^*$ iteratively with each $S_i$ so that the final cut $V_t^*\cap(\cap_iS_i)$ increases by a factor $(1+\Theta(\epsilon/\log|T|))^{O(\log|T|)}\le1+O(\epsilon)$ compared to $V_t^*$.

The rest of the algorithm looks for $V_t^*\cap(\cap_iS_i)$ instead of $V_t^*$. After computing the $O(\log|T|)$ cuts, the algorithm removes all cut edges from the graph, isolating the terminals from each other. For each terminal $t\in T$, the algorithm considers the connected component containing $t$, contracts the graph outside this component, and makes one final $(1+\Theta(\epsilon/\log|T|))$-fair cut computation to obtain $V_t$, which approximates $V_t^*\cap(\cap_iS_i)$. The existence of a feasible flow satisfying properties~(\ref{item:isolating-cuts-source-1}) and~(\ref{item:isolating-cuts-source-2}) follows from the definition of fair cuts; we omit the details.


For the running time, we note that the fair cuts algorithm in \cite{li2023near} runs in $\tilde{O}(m/\epsilon^{O(1)})$ time regardless of $\phi$, but their algorithm is randomized. Fortunately, as remarked right before Section~1.2 of~\cite{li2023near}, the only randomized component is building a \emph{congestion approximator}~\cite{sherman2013nearly} of the graph. If $G$ has conductance $\phi$, then there is a trivial congestion approximator (the singleton vertices) with ``quality'' $1/\phi$, so there is a fast, deterministic construction. The final running time picks up a $1/\phi^{O(1)}$ factor due to the quality. We omit the details, which are standard in the literature.
\end{proof}

We construct $H=(V_H,E_H)$ as follows. Start with $G[A']$ and add a vertex $t^*$ connected to each vertex $v\in A'$ by an edge of weight $\frac{\epsilon\tilde\lambda}{10s_0}d^W_G(v)$. Next, for each $i\in[k]$, add a vertex $t_i$ connected to each vertex $v\in S_i$ by an edge of weight $(1-2\epsilon)w_G(V\setminus A',v)$. Call the resulting graph $H$, and let $T=\{t_1,\ldots,t_k,t^*\}$.

We first bound the conductance of the graph. Consider a subset $S\subseteq V_H$, and assume without loss of generality that $t^*\notin S$. Then, $\partial_HS\ge w_H(S,t^*)=\frac{\epsilon\tilde\lambda}{10s_0}\vol^W _G(S\setminus T)$. Since $t^*\notin S$, we have $\vol^W _H(S)=\Theta(\vol^W _G(S\setminus T))$ by construction. Therefore, $\partial_HS\ge\Omega(\frac{\epsilon\tilde\lambda}{s_0}\vol^W _H(S))\ge\Omega(\frac{\epsilon\tilde\lambda}{s_0})\cdot\min\{\vol^W _H(S),\vol^W _H(\bar S)\}$, and $H$ has conductance $\Omega(\frac{\epsilon\tilde\lambda}{s_0})$. 

We run \Cref{thm:isolating-cuts-source} on $H$ and $T$ with parameter $\epsilon'\gets\min\{\epsilon,\,c\cdot\tilde\lambda/s_0\}$ for sufficiently small constant $c>0$, and let $V_{t_1},\ldots,V_{t_k},V_{t^*}$ be the returned set. By the conductance bound, the algorithm takes time $\tilde{O}((\frac{\epsilon\tilde\lambda}{s_0})^{O(1)}|E(G[A'])|)$. We define $C_i=V_{t_i}\setminus \{t_i\}$ for each $i\in[k]$. (We ignore the set $V_{t^*}$.)

We now show the properties of \Cref{lem:isolating-cuts} in order.
 \begin{enumerate}
 \item The singleton cut $\{t_i\}$ separates $t_i$ from $T\setminus\{ t_i\}$ , and since $C_i$ is a $(1+\epsilon)$-approximate minimum cut separating $t_i$ from $T\setminus \{t_i\}$, we have $\partial_HC_i\le(1+\epsilon)\partial_H\{t_i\}=(1+\epsilon)(1-2\epsilon)w_G(V\setminus A',S_i)\le(1-\epsilon)w_G(V\setminus A',S_i)$. Also, since $t^*\in T\setminus \{t_i\}$, we have $t^*\notin C_i$, so $\partial_HC_i\ge w_H(C_i,t^*)=\frac{\epsilon\tilde\lambda}{10s_0}\vol^W _G(C_i)$. Combining these two inequalities, we obtain $\frac{\epsilon\tilde\lambda}{10s_0}\vol^W _G(C_i)\le\partial_HC_i\le(1-\epsilon)w_G(V\setminus A',S_i)\le(1-\epsilon)\vol^W _G(S_i)$, which proves property~(\ref{item:isolating-cuts-1}).
 
 \item By property~(\ref{item:isolating-cuts-source-2}) of \Cref{thm:isolating-cuts-source}, there is at least $(1-\epsilon')w_H(V_{t_i},V_H\setminus V_{t_i})$ flow going out of $V_{t_i}$. By property~(\ref{item:isolating-cuts-source-1}), this flow can only originate from $t_i$. Since there is at most $w_H(t_i,V_{t_i}\setminus \{t_i\})$ total capacity from $t_i$ into $V_{t_i}\setminus \{t_i\}$, we must have $w_H(t_i,V_{t_i}\setminus \{t_i\})\ge(1-\epsilon')w_H(V_{t_i},V_H\setminus V_{t_i})$. We now relate these quantities to those involving graph $G$. By construction, $w_H(t_i,V_{t_i}\setminus \{t_i\})=w_H(t_i,C_i)=(1-2\epsilon)w_G(V\setminus A',C_i\cap S_i)$. Also, $\partial_{G[A']}C_i=w_H(V_{t_i}\setminus \{t_i\},V_H\setminus(V_{t_i}\cup T))$ since no additional edges in $H$ are added between $V_{t_i}\setminus \{t_i\}$ and $V_H\setminus(V_{t_i}\cup T)$.  Therefore,
\begin{align*}
\partial_{G[A']}C_i&=w_H(V_{t_i}\setminus \{t_i\},V_H\setminus(V_{t_i}\cup T))
\\&\le w_H(V_{t_i},V_H\setminus V_{t_i})
\\&\le\frac1{1-\epsilon'}w_H(t_i,V_{t_i}\setminus \{t_i\})
\\&=\frac1{1-\epsilon'}\cdot(1-2\epsilon)w_G(V\setminus A',C_i\cap S_i)
\\&\le\frac1{1-\epsilon'}\cdot(1-2\epsilon)w_G(V\setminus A',C_i)
\\&\le(1-\epsilon)w_G(V\setminus A',C_i) ,
\end{align*}
where the last inequality uses $\epsilon'\le\epsilon$. This concludes property~(\ref{item:isolating-cuts-source-2}).
 \item Consider a set $U^\dag\subseteq S_i$ with $\partial_{G[A']}U^\dag\le1.01\tilde\lambda$ and $\vol^W _G(U^\dag)\le s_0$. Let $U=U^\dag\cup\{t_i\}$. 
 Since $U\subseteq S_i$ and $S_i$ is disjoint from $S_j$ for any $j\ne i$, there are no edges from $U$ to $t_j$, so $$\partial_HU=E_H(U,A'\setminus U)+E_H(t_i,S_i\setminus U)+E_H(U,t^*)=\partial_{G[A']}U^\dag+(1-2\epsilon)w_G(V\setminus A',S_i\setminus U^\dag)+\frac{\epsilon\tilde\lambda}{10s_0}\vol^W _G(U^\dag).$$ 
 
 \Cref{thm:isolating-cuts-source} guarantees a set $U'\subseteq U\cap V_t$ such that $\partial_HU'\le(1+\epsilon')\partial_HU$. 
 Let $U^*=U'\setminus \{t_i\}$, so that we similarly have $\partial_HU'=\partial_{G[A']}U^*+(1-2\epsilon)w_G(V\setminus A',S_i\setminus U^*)+\frac{\epsilon\tilde\lambda}{10s_0}\vol^W _G(U^*)$. It follows that
\begin{align}
&\partial_{G[A']}U^*+(1-2\epsilon)w_G(V\setminus A',S_i\setminus U^*)
\\\le{}& \partial_{G[A']}U^*+(1-2\epsilon)w_G(V\setminus A',S_i\setminus U^*)+\frac{\epsilon\tilde\lambda}{10s_0}\vol^W _G(U^*) \nonumber
\\={}& \partial_HU' \nonumber
\\\le{}& (1+\epsilon')\partial_HU \nonumber
\\={}&(1+\epsilon')\big(\partial_{G[A']}U^\dag+(1-2\epsilon)w_G(V\setminus A',S_i\setminus U^\dag)+\frac{\epsilon\tilde\lambda}{10s_0}\vol^W _G(U^\dag)\big) .\nonumber
\end{align}
From $U^*\subseteq U^\dag\subseteq S_i$ we obtain that
\[ w_G(V\setminus A',S_i\setminus U^*)-w_G(V\setminus A',S_i\setminus U^\dag)=w_G(V\setminus A',U^\dag\setminus U^*) ,\]
so subtracting $(1-2\epsilon)w_G(V\setminus A',S_i\setminus U^\dag)$ from both sides gives
\begin{align}
&\partial_{G[A']}U^*+(1-2\epsilon)w_G(V\setminus A',U^\dag\setminus U^*)
\\\le{}&(1+\epsilon')\partial_{G[A']}U^\dag+\epsilon' (1-2\epsilon)w_G(V\setminus A',S_i\setminus U^\dag)+\frac{(1+\epsilon')\epsilon\tilde\lambda}{10s_0}\vol^W _G(U^\dag) \nonumber
\\\le{}&(1+\epsilon')\partial_{G[A']}U^\dag+\epsilon'(1-2\epsilon)\vol^W _G(S_i\setminus U^\dag)+\frac{(1+\epsilon')\epsilon\tilde\lambda}{10s_0}\vol^W _G(U^\dag) \nonumber
\\\le{}&(1+\epsilon')\partial_{G[A']}U^\dag+\epsilon'\cdot O(s_0)+\frac{(1+\epsilon')\epsilon\tilde\lambda}{10s_0}\cdot s_0. \nonumber
\end{align}
As long as $\epsilon'\le c\cdot\tilde\lambda/s_0$ for sufficiently small constant $c>0$, we obtain
\begin{gather}
\partial_{G[A']}U^*+(1-2\epsilon)w_G(V\setminus A',U^\dag\setminus U^*) \le (1+\epsilon')\partial_{G[A']}U^\dag+\epsilon\tilde\lambda/2 .\label{eq:isolating-cuts-1}
\end{gather}
If $\partial_{G[A']}U^\dag\ge0.2\lambda$, then we obtain $\partial_{G[A']}U^*+(1-2\epsilon)w_G(V\setminus A',U^\dag\setminus U^*) \le(1+O(\epsilon))\partial_{G[A']}U^\dag$, and dividing by $(1-2\epsilon)$ fulfills property~(\ref{item:isolating-cuts-3}) with $\epsilon$ replaced by $O(\epsilon)$. Otherwise, if $\partial_{G[A']}U^\dag<0.2\lambda$, then we need to show that $U^*=U^\dag$. Suppose for contradiction that $U^*\subsetneq U^\dag$. From Equation~(\ref{eq:isolating-cuts-1}) and recalling that $\epsilon'\le\epsilon \le 0.1$ and $\tilde\lambda\le1.01\lambda$, we obtain
\begin{align}
\partial_{G[A']}U^*+(1-2\epsilon)w_G(V\setminus A',U^\dag\setminus U^*)=(1+\epsilon')\partial_{G[A']}U^\dag+\epsilon\tilde\lambda/2&<0.3\lambda.\label{eq:isolating-cuts-4}
\end{align}
Since $U^\dag\setminus U^*$ is non-empty by assumption,
\begin{align*}
\lambda \le \partial_G(U^\dag\setminus U^*)&=\partial_{G[A']}(U^\dag\setminus U^*)+w_G(V\setminus A',U^\dag\setminus U^*)
\\&\le\partial_{G[A']}U^\dag+\partial_{G[A']}U^*+w_G(V\setminus A',U^\dag\setminus U^*)
\\&\stackrel{\mathclap{(\ref{eq:isolating-cuts-4})}}\le\partial_{G[A']}U^\dag+0.3\lambda+w_G(V\setminus A',U^\dag\setminus U^*)
\\&\le0.5\lambda+w_G(V\setminus A',U^\dag\setminus U^*).
\end{align*}
It follows that $w_G(V\setminus A',U^\dag\setminus U^*)\ge0.5\lambda$, which contradicts Equation~(\ref{eq:isolating-cuts-4}) for $\epsilon\le0.1$. Finally, to satisfy property~(\ref{item:isolating-cuts-3}) with factor $(1+\epsilon)$ instead of $(1+O(\epsilon))$, we can simply reset $\epsilon$ to be a constant factor smaller.
\end{enumerate}
With all three properties established, this concludes the proof.
\end{proof}

\section{Sparsifier construction}\label{sec:sparsifier}

In this section, we outline the construction of the skeleton graph using the iterative graph decomposition from the previous section. Similarly to \cite{Li21}, we separate the construction into an ``unbalanced'' case and a ``balanced'' case.

At a high level, the algorithm derandomizes the sampling procedure where edges are sampled independently with probability proportional to their weights. The challenge is in establishing an efficient ``union bound'' over all cuts, which are exponential in number. To cope with this issue, we exploit the structure of graph cuts guaranteed by \Cref{lem:structure}. Still, the derandomization is only sufficient to preserve ``unbalanced'' cuts, which includes (near-) minimum cuts. To address the remaining ``balanced'' cuts, we construct a separate graph that forces these cuts to be large enough.

\subsection{Unbalanced sparsifier}

In this section, we present our construction of the unbalanced sparsifier $H$, which is an unweighted graph supported on a subset of edges in $G$. To establish the properties we want, we require that a certain list of quantities are approximately preserved. To formally describe these properties, we introduce the graph Laplacian for an algebraic representation of the cuts of a graph.

For a given set $S\subseteq V$, let $\mathds 1_S\in\{0,1\}^V$ be the vector with value $1$ at vertex $v$ if $v\in S$, and value $0$ otherwise. Define the\emph{ Laplacian} $L_G$ of a weighted graph $G$ as the $|V|\times|V|$ matrix indexed by $V\times V$, where each diagonal entry $(v,v)$ is $d^W(v)$ and each off-diagonal entry $(u,v)$ is $-w(u,v)$ if $(u,v)\in E$ and $0$ otherwise. It is well-known that the quadratic form $\mathds 1_S^TL_G\mathds 1_S$ is equal to $\partial_GS$.

Let $L$ be the total number of iterations of Step~\ref{item:decomposition} until $G_L$ is a single vertex.
For each iteration $j$ and each vertex $v\in V_j$, let $\overline v\subseteq V$ denote its ``pre-image under the contraction map'', defined as all original vertices in $V$ that get contracted into $v$ in graph $G_j$. For a set $S\subseteq V_j$, let $\overline S=\bigcup_{v\in D}\overline v$ be the pre-image of $S$.

The unbalanced sparsifier $H$ is an unweighted graph supported on a subset of the edges of $G$.
Let $\nu >0$ and $\eta >0$ be parameters. Our goal is to compute $H$ and an associated weight $W$ so that for all iterations $j,k$ and vertices $u\in V_j,v\in V_k$ with $d^W_{V_j}(u)\le \nu$ and $d^W_{V_k}(v)\le \nu$, the quantity $\mathds 1^T_{\overline u}L_G\mathds 1_{\overline v}$ is preserved to an additive error of $\eta\lambda$: $|W\cdot\mathds 1_{\overline u}^TL_H\mathds 1_{\overline v}-\mathds 1_{\overline u}^TL_G\mathds 1_{\overline v}|\le\eta\lambda$.

We first observe that if $d^W_{V_j}(u)\le \nu$ and $d^W_{V_k}(v)\le \nu$, then the total weight of edges that affect the expression $\mathds 1^T_{\overline u}L_H\mathds 1_{\overline v}$ is at most $\nu$. Assume without loss of generality that $j\le k$; then, either $\overline u\cap\overline v=\emptyset$ or $\overline v\subseteq\overline u$. If the former is true, then $\mathds 1_{\overline u}^TL_G\mathds 1_{\overline v}=-w(\overline u,\overline v)$ only depends on edges with exactly one endpoint in $\overline u$, which is at most $\partial_G(\overline u)=d^W_{V_j}(u)$ total weight of edges. If the latter, then $\mathds 1_{\overline u}^TL_G\mathds 1_{\overline v}=\mathds 1_{\overline u}^TL_G\mathds 1_{\overline u}-\mathds 1_{\overline u}^TL_G\mathds 1_{\overline v\setminus\overline u}=\partial_G(\overline u)-w(\overline u,\overline v\setminus\overline u)$ also only depends on edges with exactly one endpoint in $\overline u$.

Therefore, a simple random sampling procedure gives sufficient concentration. Namely, sampling each edge $e\in E$ with probability proportional to $w(e)$ 
and assigning the same weight $W$ to every sampled edge
preserves the $\eta\lambda$ additive error with high probability by Chernoff bounds. Also, note that $\mathds 1_{\overline u}^TL_H\mathds 1_{\overline v}$ is only nonzero when there is an edge in $G$ between $\overline u$ and $\overline v$. In fact, for each $j,k$, there are $O(1)$ choices of $u\in V_j$ and $v\in V_k$ for which the edge participates in $\mathds 1^T_{\overline u}L_G\mathds 1_{\overline v}$: the endpoints of the edge must be contained in $\overline u$ and $\overline v$. So the total ``representation size'' of the required constraints is $O(L^2|E|)$,  excluding the constraints for which the quantities are zero, which are vacuously satisfied. Since the total representation size of the constraints is only $O(L^2|E|)$, we can derandomize the above sampling procedure with the method of pessimistic estimators, following the approach from Section~3.2.1 of \cite{Li21}. We omit the technical details since the setup is nearly identical and is standard in the literature. The final result can be stated as follows.

\begin{lemma}[Lemma~3.8 of~\cite{Li21}]\label{lem:pessimistic}
For any parameters $\nu,\eta>0$, there is an algorithm with running time $\tilde{O}(L^2 m)$ that computes an unweighted graph $H$ and a weight $W\le O(\frac{\eta\lambda^2}{\nu\ln(Lm)})$ such that for all iterations $j,k$ and vertices $u\in V_j,v\in V_k$ with $d^W_{V_j}(u)\le \nu$ and $d^W_{V_k}(v)\le \nu$, we have $|W\cdot\mathds1_{\overline u}^TL_H\mathds1_{\overline v}-\mathds1_{\overline u}^TL_G\mathds1_{\overline v}|\le\eta\lambda$.
\end{lemma}

We now prove the key property of the unbalanced sparsifier $H$ and associated weight $W$. Intuitively, the lemma below says that if a set $S\subseteq V$ can be represented efficiently by the set difference of sets of the form $\overline v$ with small total degree $\sum_vd^W(v)$, then the cut $\partial S$ is approximately preserved by $H$. Here, the degree $d^W(v)$ of vertex $v\in\cup_jV_j$ equals $d^W_{V_j}(v)$ for the iteration $j$ with $v\in V_j$.
\begin{lemma}\label{lem:structure-unbalanced}
For any cut $S\subseteq V$ in $G$, suppose that there exists a set $D\subseteq\cup_jV_j$ such that $S=\triangle_{v\in D}\overline v$ and $d=\sum_{v\in D}d^W(v)\le\nu$. Then, $|W\cdot\partial_HS-\partial_GS|\le (d/\lambda)^2\eta\lambda$.
\end{lemma}
\begin{proof}
Since the sets $\overline v$ for $v\in\cup_jV_j$ form a laminar family, we can translate the set difference into algebraic form $\mathds 1_S=\sum_{v\in D}\pm\mathds 1_{\overline v}$ by traversing through the sets $V_j$ in decreasing order of $j$ and adding or subtracting vectors $\mathds 1_{\overline v}$ as necessary. We now write
\begin{align*}
W\cdot\partial_HS-\partial_GS&=W\cdot\mathds 1_S^TL_H\mathds 1_S-\mathds 1_S^TL_G\mathds 1_S
\\&= \mathds 1^T_S(W\cdot L_H-L_G)\mathds 1_S
\\&=\bigg(\sum_{v\in D}\pm\mathds 1_{\overline v}\bigg)^T(W\cdot L_H-L_G)\bigg(\sum_{v\in D}\pm\mathds 1_{\overline v}\bigg) .
\end{align*}
Expanding the summation gives $|D|^2$ terms of the form $\mathds 1^T_{\overline u}(W\cdot L_H-L_G)\mathds 1_{\overline v}$, each of which is approximated to additive error $\eta\lambda$ by Lemma~\ref{lem:pessimistic} since $d \le \nu$. It remains to bound the number of terms. Each $v\in D$ has $d^W(v)=\partial_G\overline v\ge\lambda$ unless $\overline v=V$, in which case $\mathds 1^T_{\overline u}(W\cdot L_H-L_G)\mathds 1_{\overline v}=\mathds 1^T_{\overline u}(W\cdot L_H-L_G)\mathds 1_V=0$ for any $u$. So there are at most $d/\lambda$ many vertices, and the additive error of $W\cdot\mathds 1_S^TL_H\mathds 1_S-\mathds 1_S^TL_G\mathds 1_S$ is at most $(d/\lambda)^2\eta\lambda$, as promised.
\end{proof}

We now show that in particular, the minimum cut can be efficiently represented this way.
\begin{lemma}\label{lem:mincut-unbalanced}
For any minimum cut $S\subseteq V$ in $G$, we can represent $S$ as the symmetric difference $\triangle_{v\in D}\overline v$ for a set $D\subseteq\cup_jV_j$ with $\sum_{v\in D}d^W(v)\le O(Ls_0^2/(\epsilon\lambda))$.
\end{lemma}
\begin{proof}
Since $S$ is a minimum cut, we have $\partial_GS=\lambda$. Construct a sequence $S=S_0,S_1,\ldots,S_\ell$ as follows. For each $j$ in increasing order, we ensure by induction that $S\subseteq V_j$ is a cut with value $\partial_{G_j}S\le(1+3\epsilon)^j\lambda$. Since $\epsilon=1/\log^{1.1}|V|$ and $j \le \log |V|$, we have $\partial_{G_j}S\le(1+3\epsilon)^j\lambda\le1.01\lambda$. So $\partial_{G_j}S\le1.01\lambda$, and we can apply \Cref{lem:structure} to obtain a set $S'$, which we index as $S'_j$. If $S'_j=\emptyset$ or $S'_j=V_j$, then we stop. Otherwise, let $S_{j+1}$ be the set $S'_j$ after the contraction from $V_j$ to $V_{j+1}$, which is a cut with value
\[ \partial_{G_{j+1}}S_{j+1}=\partial_{G_j}S'_j\le(1+3\epsilon)\partial_{G_j}S_j\le(1+3\epsilon)^{j+1}\lambda ,\] preserving the inductive statement. 

We define $D'=\cup_{j=0}^{\ell-1}(S_j\triangle S'_j)$, which satisfies either $\triangle_{v\in D'}\overline v=S$ or $\triangle_{v\in D'}\overline v=V\setminus S$ depending on whether $S'_j=\emptyset$ or $S'_j=V_j$ on the last iteration. If the latter, then set $D=D'\cup\{v\}$ for the single vertex $v\in V_L$; otherwise, set $D=D'$. By construction, $S=\triangle_{v\in D}\overline v$. Also, by property~(\ref{item:structure-1c}) of \Cref{lem:structure}, we have $\vol^W(S_j\triangle S'_j)\le O(s_0^2/(\epsilon\lambda))$ for each $j$, so the total degree $\sum_{v\in D}d^W(v)$ is at most $L$ times that, as promised.
\end{proof}

\subsection{Balanced sparsifier}

In \cite{Li21}, the balanced sparsifier is simply a deterministic $n^{o(1)}$-approximate sparsifier, which is itself constructed by overlaying expanders at each level of an expander decomposition hierarchy. For technical reasons, we are unable to extend this approach to our setting. Instead, we will introduce \emph{Steiner} vertices in the balanced sparsifier, which are eventually \emph{split off}.

For each iteration $j\in\{1,2,\ldots,L\}$ and each vertex $v\in V_j$, we introduce a Steiner vertex $x_v^j$. For each original vertex $v\in V$, let $x_v^0$ indicate that vertex. For each iteration $j\in\{0,1,\ldots,L-1\}$ and each vertex $u\in V_j$, let $v\in V_{j+1}$ be the vertex that $u$ contracts to; add an edge $(x_u^j,x_v^{j+1})$ of weight $d^W_{V_j}(u)$. We do not add any edges between original vertices. This concludes the construction of the balanced sparsifier, which we name $X=(V_X,E_X)$. 

For each cut in $G$, we can extend it to a cut in $X$ by adding each Steiner vertex to one side of the cut. The key lemma for the balanced sparsifier is the following. Intuitively, it says that cuts $S$ in $G$ that can be extended to small cuts in $X$ (by siding the Steiner vertices optimally) are precisely the sets $S$ that can be represented by the set difference of sets of the form $\overline v$ with small total degree $\sum_vd^W(v)$.

\begin{lemma}\label{lem:structure-balanced}
For any cut $S\subseteq V$ in $G$, the following two minima are equal:
 \begin{enumerate}
 \item The minimum value $\sum_{v\in D}d^W(v)$ over all sets $D\subseteq\cup_jV_j$ such that the symmetric difference $\triangle_{v\in D}\overline v$ equals $S$. Here, $d^W(v)=d^W_{G_j}(v)$ for the iteration $j$ with $v\in V_j$. \label{item:minimum-1}
 \item The minimum value of $\partial_XS^*$ over all cuts $S^*\subseteq V_X$ with $S^*\cap V=S$.\label{item:minimum-2}
 \end{enumerate}
\end{lemma}
\begin{proof}
To show that the minimum in (\ref{item:minimum-1}) is at most the minimum in (\ref{item:minimum-2}), consider a set $D\subseteq\cup_jV_j$ such that the symmetric difference $\triangle_{v\in D}\overline v$ equals $S$. We construct a cut $S^*$ in $X$ with $S^*\cap V=S$ and total weight $\sum_{v\in D}d^W(v)$ as follows. Initialize $S_0\gets S$, and for each iteration $j$ in increasing order, let $S'_j=S_j\triangle(D\cap V_j)$ and let $S_{j+1}\subseteq V_{j+1}$ be the contraction of the set $S'_j$ into $V_{j+1}$, which is well-defined. The cut $S^*$ is the union of $S\subseteq V$ and $\{x_v^j \mid v,j:v\in S_j\}$. For any edge $(x^j_u,x^{j+1}_v)$ cut by $S^*$ in $X$, where $u\in V_j$ contracts to $v\in V_{j+1}$, we must have either $u\in S_j$ or $v\in S_{j+1}$ but not both. It follows that $u\in D\cap V_j$. Conversely, for any $u\in D\cap V_j$ that contracts to $v\in V_{j+1}$, we have either $x^j_u\in S^*$ or $x^{j+1}_v\in S^*$ but not both, so edge $(x^j_u,x^{j+1}_v)$ cut by $S^*$ is cut by $S^*$. So there is a one-to-one mapping between vertices in $D$ and edges cut by $S^*$ with the same degree and edge weight, respectively. It follows that the minimum in (\ref{item:minimum-1}) is at most the minimum in (\ref{item:minimum-2}), 

To show the other direction, consider a cut $S^*$ in $X$ with $S^*\cap V=S$. Denote by $x^L$ the vertex $x^L_v$ for the single vertex $v\in L$. Without loss of generality, assume that $x^L\notin S^*$. Let $D$ consist of all vertices $u$ for which the edge $(x^j_u,x^{j+1}_v)$ is cut by $S^*$ in $X$. By construction, $\sum_{v\in D}d^W(v)$ is the value of the cut $S^*$ in $X$. It remains to show that $S=\triangle_{v\in D}\overline v$. For any $v\in S$ and iteration $j$, there is exactly one vertex $u\in V_j$ with $v\in\overline u$, and the corresponding vertex $x_u^j$ is on the path from $x_v^0$ to $x^L$. Since $x^L\notin S^*$, for any $v\in S$ there is an odd number of edges cut by $S^*$ on the path from $x_v^0$ to $x^L$, which corresponds to an odd number of vertices $u\in D$ for which $v\in\overline u$. It follows that $v\in\triangle_{u\in D}\overline u$. For any $v\notin S$, replacing ``odd'' with ``even'' in the argument above, we have $v\notin\triangle_{u\in D}\overline u$. Thus $S=\triangle_{v\in D}\overline v$, concluding the proof.
\end{proof}

\subsubsection{Combining the two sparsifiers}

Let $G'$ be the union of the following two graphs: the unbalanced sparsifier $H$ with all edges weighted by $W$, and the graph $X$ with all edge weights multiplied by $Z=(\epsilon\lambda)^2/Ls_0^2$. So $G'$ is supported on vertex set $V_X$. For a cut $S\subseteq V$, we call a cut $S^*\subseteq V_X$ an \emph{extension} of $S$ if $S^*\cap V=S$. We now claim that the minimum cut can be extended to a near-minimum cut in $G'$.

\begin{claim}\label{clm:cut-upperbound}
Assume that $\nu\ge \Omega(Ls_0^2/(\epsilon\lambda))$.
For any minimum cut $S\subseteq V$ in $G$, there exists an extension $S^*\subseteq V_X$ with $\partial_{G'}S^*\le \big(1+O(\epsilon)+O( ( \frac{Ls_0^2}{\epsilon\lambda^2})^2\eta) \big)\lambda$.
\end{claim}
\begin{proof}
By \Cref{lem:mincut-unbalanced}, we can write $S=\triangle_{v\in D}\overline v$ for a set $D\subseteq\cup_jV_j$ with $\sum_{v\in D}d^W(v)\le O(Ls_0^2/(\epsilon\lambda))$. By \Cref{lem:structure-unbalanced} with $d=O(Ls_0^2/(\epsilon\lambda))\le\nu$, we obtain
\[ W\cdot\partial_HS \le \partial_GS + O\bigg( \frac{Ls_0^2}{\epsilon\lambda^2}\bigg)^2 \eta \lambda = \bigg(1+O\bigg( \frac{Ls_0^2}{\epsilon\lambda^2}\bigg)^2\eta \bigg)\lambda .\]
So for any extension $S^*\subseteq S$, the total weight in $\partial_{G'}S^*$ resulting from the unbalanced sparsifier $H$ is at most $\big(1+O( ( \frac{Ls_0^2}{\epsilon\lambda^2})^2\eta) \big)\lambda$.

By \Cref{lem:structure-balanced}, there exists a set $S^*\subseteq V_X$ with $\partial_XS^*=\sum_{v\in D}d^W(v)=O(Ls_0^2/(\epsilon\lambda))$. After dividing every weight by $Z=(\epsilon\lambda)^2/Ls_0^2$, the total weight in $\partial_{G'}S^*$ resulting from the balanced sparsifier $X$ is at most $O(\epsilon\lambda)$. Together with the weight from the unbalanced sparsifier $H$, we obtain the desired bound.
\end{proof}

At the same time, we do not want any extensions of cuts to go below the minimum cut $\lambda$. The claim below ensures this does not happen.
\begin{claim}\label{clm:cut-lowerbound}
Assume that $\nu\ge Ls_0^2/(\epsilon^2\lambda)$. Then, for any cut $S\subseteq V$ in $G$ and any extension $S^*\subseteq V_X$, we have $\partial_{G'}S^*\ge(1-(\epsilon\lambda)^4/(L^2s_0^4)\eta)\lambda$.
\end{claim}
\begin{proof}
Consider the minimum value $\sum_{v\in D}d^W(v)$ over all sets $D\subseteq\cup_jV_j$ such that the symmetric difference $\triangle_{v\in D}\overline v$ equals $S$. If $\sum_{v\in D}d^W(v)\ge\lambda/Z$, then by \Cref{lem:structure-balanced}, any extension $S^*\subseteq V$ must have $\partial_XS^*\ge\lambda/Z$. After scaling by $Z$, the total weight in $\partial_{G'}S^*$ resulting from the balanced sparsifier $X$ is at least $\lambda$, as promised.

Otherwise, we have $\sum_{v\in D}d^W(v)\le\lambda/Z=Ls_0^2/(\epsilon^2\lambda)\le\nu$, so \Cref{lem:structure-unbalanced} guarantees that 
\[ W\cdot\mathds 1_{\overline u}^TL_H\mathds 1_{\overline v} \ge \mathds 1_{\overline u}^TL_G\mathds 1_{\overline v}-(1/Z)^2\eta\lambda=\lambda-(\epsilon\lambda)^4/(Ls_0^2)^2\eta\lambda .\]
So the total weight in $\partial_{G'}S^*$ resulting from the unbalanced sparsifier $H$ is at least $(1-(\epsilon\lambda)^4/(Ls_0^2)^2\eta)\lambda$, as promised.
\end{proof}

\subsubsection{The final construction}
We start by finalizing our parameters for the algorithm:
 \begin{enumerate}
 \item We set $\eta=\textup{poly}(\epsilon,1/L,\lambda/s_0)$ small enough so  that the right-hand side of the guarantees of \Cref{clm:cut-upperbound} and \Cref{clm:cut-lowerbound} are $(1+O(\epsilon))\lambda$ and $(1-O(\epsilon))\lambda$, respectively.
 \item We set $\nu= c \cdot Ls_0^2/(\epsilon^2\lambda)$ 
 for some constant $c \ge 1$ so that \Cref{clm:cut-upperbound} and \Cref{clm:cut-lowerbound} hold.
 \item Given $\eta$ and $\nu$, we set $W\le O(\frac{\eta\lambda^2}{\nu\ln(Lm)})=\textup{poly}(\epsilon,1/L,\lambda/s_0,1/\log n)\cdot\lambda$ so that \Cref{lem:pessimistic} holds.
 \item Recall that we set $Z=(\epsilon\lambda)^2/Ls_0^2=\textup{poly}(\epsilon,1/L,\lambda/s_0)$.
 \end{enumerate}

We thus obtain a graph $G'$ satisfying Claims~\ref{clm:cut-upperbound}~and~\ref{clm:cut-lowerbound}. We want an unweighted sparsifier at the end, so we now ``discretize'' the edges of $G'$. By construction, the balanced sparsifier $X$ consists of edges $(x_u^j,x_v^{j+1})$ of weight $d^W_{V_j}(u)\ge\lambda$, so all edges in $G'$ originating from $X$ have weight at least $\lambda/Z$. By construction, all edges originating from the unbalanced sparsifier $H$ have weight $W$. For each edge $e$ in $G'$, replace it with $\lfloor w(e)/\min\{\epsilon W,\epsilon\lambda/Z\}\rfloor$ many parallel edges of weight $W'=\min\{\epsilon W,\epsilon\lambda/Z\}$ each, which preserves all cuts up to factor $(1+O(\epsilon))$. We now bound
\[ W'=\textup{poly}(\epsilon,1/L,\lambda/s_0,1/\log n)\cdot\lambda=(\epsilon/\log n)^{O(1)}\lambda \]
since $L=O(\log n)$ and $s_0=\lambda\polylog(n)$. Let $H'$ be the unweighted graph consisting of these edges of weight $W'$.

Next, we want to split off the Steiner vertices originating from $X$. Given a graph and a subset of vertices called the \emph{terminals}, the Steiner minimum cut of a graph is the minimum cut separating at least two terminals. Let $V\subseteq V_X$ be the terminals in our setting. We now call the \emph{splitting-off} algorithm of Bhalgat~et~al.~\cite{bhalgat2008fast} on $H'$, which guarantees the following.
\begin{theorem}[\cite{bhalgat2008fast}]\label{thm:splitting-off}
Given an unweighted graph $G=(V_X,E)$ with terminal set $T\subseteq V_X$ and Steiner minimum cut $c$, there is a deterministic, $\tilde{O}(m+nc^2)$ time algorithm that computes a graph $G'$ satisfying the following:
 \begin{enumerate}
 \item Cuts cannot increase in value: for each cut $S\subseteq V_X$, we have $\partial_{G'}S\le\partial_GS$,\label{item:splitting-1}
 \item The Steiner minimum cut of $G$ equals the Steiner minimum cut of $G'$, and\label{item:splitting-2}
 \item Every non-terminal vertex $v\in V_X\setminus T$ is isolated: $d_{G'}(v)=0$.\label{item:splitting-3}
 \end{enumerate}
\end{theorem}
Let $H''$ be the output of applying \Cref{thm:splitting-off} to $H'$. By construction of $H'$, any minimum cut $S\subseteq V$ in $G$ extends to a cut $S^*\subseteq V_X$ in $H'$ with $\partial_{H'}S^*\le(1+O(\epsilon))\lambda/W'$, and property~(\ref{item:splitting-1}) guarantees that $\partial_{H''}S^*\le(1+O(\epsilon))\lambda/W'$ as well. In particular, the Steiner minimum cut of $H'$ is at most $(1+O(\epsilon))\lambda/W'$, and \Cref{thm:splitting-off} takes $\tilde O(m+n(\lambda/W')^2)=\tilde O(m/\epsilon^{O(1)})$ time. On the other hand, any extension $S^*\subseteq V_X$ of a cut $S\subseteq V$ in $G$ satisfies $\partial_{H'}S^*\ge(1-O(\epsilon))\lambda/W'$, so the Steiner minimum cut of $H'$ is at least $(1-O(\epsilon))\lambda/W'$, and property~(\ref{item:splitting-2}) ensures that the Steiner minimum cut of $H''$ is also at least $(1-O(\epsilon))\lambda/W'$. Finally, let $\widetilde H$ be the restriction of $H''$ to the terminal set $V$, which does not affect cut extensions by property~(\ref{item:splitting-3}). We have thus constructed a minimum cut sparsifier, also known as a \emph{skeleton graph}, formalized below (after resetting $\epsilon>0$ to be a constant factor smaller).
\begin{theorem}[Skeleton graph, see Theorem~1.5 of~\cite{Li21}]\label{thm:sparsifier}
For any $0<\epsilon\le1$, we can compute, in deterministic $\tilde{O}(m/\epsilon^{O(1)})$ time, an unweighted graph $\widetilde H$ and some weight $W'=(\epsilon/\log n)^{O(1)}\lambda$ such that
 \begin{enumerate}
 \item For any mincut $\partial S^*$ of $G$, we have $W'\cdot|\partial_{\widetilde H}S^*|\le(1+\epsilon)\lambda$, and
 \item For any cut $\emptyset\subsetneq S\subsetneq V$ of $G$, we have $W'\cdot|\partial_{\widetilde H}S| \ge (1-\epsilon)\lambda$.
 \end{enumerate}
\end{theorem}
Section~1.2.1 of~\cite{Li21} outlines Karger's framework for deterministic minimum cut given the skeleton graph of~\Cref{thm:sparsifier}. The final running time is $\tilde{O}(m)$.

\bibliographystyle{alpha}
\bibliography{ref}
\begin{appendix}
\section{Missing details and analysis of Section~\ref{sec:unit-flow}}
In this section we provide the implementation details and the analysis of the Unit-Flow algorithm from~\Cref{sec:unit-flow}. The algorithm and its analysis largely follows the standard Push-relabel algorithm~\cite{GT88} and the original version of Unit-Flow for unweighted graphs in~\cite{DBLP:journals/siamcomp/HenzingerRW20} with fairly mechanical adaptation, and we include material directly from aforementioned work here for completeness. For readers familiar with Push-relabel and unweighted Unit-Flow, the main difference from the unweighted version of Unit-Flow is that we need to use the dynamic-tree data structure~\cite{ST81} for weighted graphs (as in the standard Push-relabel algorithm) since we want the running time to depend on the number of edges rather than the weighted volume. Similar to the unweighted case, by limiting the vertex labels to be at most $h$ in Unit-Flow, we can bound the total number of relabels per node to be $O(h)$. This is a crucial factor in the running time of Push-relabel, and in particular bounds the number of times we cut or link each edge in the dynamic-tree data structure at $O(h)$. This is the reason that our Unit-Flow runs in time $\tilde{O}(m\cdot h)$ instead of the $\tilde{O}(mn)$ running time of standard Push-relabel where labels can go as high as $O(n)$.
\label{app:unit-flow}
\subsection{Unit-Flow Implementation}
\label{app:unit-flow-details}
{\em Unit-Flow} (Algorithm~\ref{alg:unit-flow})
 takes as input an undirected weighted graph $G=(V,E,w)$ (possibly with self-loops), source function
$\Delta$ and integer $\beta\geq 2$ such that $\forall v:0\leq \Delta(v)\leq \beta d^W(v)$, as well as a capacity parameter $U>0$ on all edges. Each vertex $v$ is a sink of capacity $d^W(v)$. Furthermore, the procedure takes as input an integer $h\ge\ln(\vol^W(G))$ to customize the Push-relabel algorithm, which we describe next.

In the Push-relabel algorithm~\cite{GT88}, each vertex $v$ has a non-negative integer label $l(v)$ which is initially zero. The label of a vertex only increases during the execution of the algorithm and in our modification of the standard Push-relabel the label cannot become larger than $h$. The bound of $h$ on the labels makes the running time of {\em Unit-Flow} linear in $h$, but it may prevent our algorithm from routing all the supply to sinks even when there exists a feasible routing for the flow problem. However, when our algorithm cannot route a feasible flow, allowing labels of value up to $h$ is sufficient to find a cut with low conductance (i.e.,~of value inversely proportional to $h$), which is our primary concern. 

The algorithm maintains a pre-flow and the standard residual network, where each undirected edge $\{v,u\}$ in $G$ corresponds to two directed arcs $(v,u)$ and $(u,v)$ in the residual network, with flow values such that $f(v,u)=-f(u,v)$, and $|f(v,u)|,|f(u,v)|\leq U\cdot w_e$. The residual capacity of an arc $(v,u)$ is $r_f(v,u)=U\cdot w_e-f(v,u)$. We also maintain $f(v)=\Delta(v)+\sum_u f(u,v)$, which will be non-negative for all nodes $v$ during the execution, and the excess $\ex(v)=\max\left(0,f(v)-d^W(v)\right)$.

As in the generic Push-relabel framework, an arc $(v,u)$ is {\em eligible} if $r_f(v,u)>0$ and $l(v) = l(u)+1$. A vertex $v$ is {\em active} if $l(v)<h$ and $\ex(v)>0$. The
algorithm  maintains the property that for any arc $(v,u)$ with
positive residual capacity,  $l(v) \leq l(u) + 1$.
The algorithm maintains a (FIFO) queue $Q$ of active vertices and repeatedly picks an active vertex $v$ from $Q$ to either pushes along an eligible arc incident to
$v$ if there is one, or to raise the label of $v$ by $1$ if $v$ is active but there is no eligible arc out of $v$. In addition, we follow the efficient implementation of Push-relabel using the dynamic-tree data structure, which supports these operations below. The total time for a sequence of $t$ tree operations starting with a collection of single-vertex trees is $O(t\ln n)$ when we have $n$ vertices in total.\\

The dynamic-tree data structure maintains a set of vertex-disjoint rooted trees, where each vertex $v$ has an associated real value $g(v)$, possibly $\pm\infty$. Tree edges are directed toward the root, and every vertex $v$ is considered as both an ancestor and a descendant of itself. In our context, the value $g(v)$ of a vertex $v$ in its dynamic tree is $r_f(v,w)$ if $v$ has a parent $w$ and $g(v)=\infty$ if $v$ is a tree root. Each vertex starts in a single-vertex tree and has $g(v)=\infty$. As we don't aim to optimize the $\ln|C|$ term in the running time, we allow any tree to grow as large as possible (in contrast to the standard Push-relabel implementation). The tree edges in the dynamic trees form a subset of the eligible arcs intersecting the set of current incident edge being considered for each vertex (denoted as $\current(v)$), and the tree operations allow us to push supply along an entire path in a tree in $O(\ln |C|)$ time, causing either a saturating push (i.e. use up the residual capacity of an arc on the path) or moving excess supply from some vertex in the tree all the way to the tree root (See the {\em Send} operation in Algorithm~\ref{alg:unit-flow}). The algorithm keeps the pre-flow $f$ in two different ways; If $(v,w)$ is an edge not in any dynamic tree, $f(v,w)$ is stored explicitly with the edge; Otherwise, if $e=(v,w)$ is a dynamic tree edge, with $w$ being the parent of $v$, then $g(v)=U\cdot w_e - f(v,w)$ is stored implicitly in the dynamic-tree data structure, and will be computed when the tree edge $(v,w)$ is cut. However, note in our Unit-Flow implementation with incremental source supply, we only need to know the $f(v)$'s and $\ex(v)$'s (i.e. supply and excess on nodes) each time Unit-Flow finishes with an interim source supply (to check if we move to the next stage or if there is enough total excess to make a cut). The supply ending on vertices can always be maintained explicitly during the execution of the algorithm, since only the two endpoints' supply will change when we send supply along a tree path. In particular, this means when $\Delta(\cdot)$ increases, we can just continue from the dynamic-tree data structure's state at the time when Unit-Flow finishes with the previous $\Delta(\cdot)$, and we don't need to do anything about the data structure (e.g. explicitly compute the supply routed on each edge) before finishing with the entire aggregate flow problem. To avoid including more redundant material, we refer readers to~\cite{GT88} for further details.\\
\\
\vspace{0.2in}
\fbox{
\parbox{\textwidth}{
{\bf Dynamic-Tree Operations}
\begin{itemize}
    \item $\dtfr(v)$: Find and return the root of the tree containing vertex $v$.
    \item $\dtfs(v)$: Find and return the number of vertices in the tree containing vertex $v$.
    \item $\dtfv(v)$: Compute and return $g(v)$.
    \item $\dtfm(v)$: Find and return the ancestor $w$ of $v$ of minimum value $g(w)$. In case of a tie, choose the vertex $w$ closest to the root.
    \item $\dtcv(v,x)$: Add real number $x$ to g(w) for each ancestor $w$ of $v$. We follow the convention that $\infty+(-\infty)=0$.
    \item $\dtlink(v,w)$: Combine the trees containing vertices $v$ and $w$ by making $w$ the parent of $v$. This operation does nothing if $v$ and $w$ are in the same tree or if $v$ is not a tree root.
    \item $\dtcut(v)$: Break the tree containing $v$ into two trees by deleting the edge from $v$ to its parent. This operation does nothing if $v$ is a tree root.
\end{itemize}}}
\begin{algorithm}
\label{alg:unit-flow}
~\\
\fbox{
\parbox{0.95\textwidth}{
{\em Unit-Flow}($G$,$\Delta$,$U$,$h$,$w$) \\ 
\tab {\bf Initialization:}\\
\tab \tab $\forall \{v,u\}\in E$, $f(u,v)=f(v,u)=0$, $Q=\{v|\Delta(v)>d^W(v)\}$.\\
\tab \tab $\forall v$, $l(v)=0$, and $\current(v)$ is the first edge in its list of incident edges.\\
\tab {\bf While} $Q$ is not empty\\
\tab \tab Remove the vertex $v$ at the front of $Q$.\\
\tab \tab {\bf Repeat}\\
\tab \tab \tab Perform {\em Tree-Push/Relabel}$(v)$. \\
\tab \tab \tab {\bf If} $w$ becomes active during this {\em Tree-Push/Relabel} operation \\
\tab \tab \tab \tab Add $w$ to the rear of $Q$\\
\tab \tab \tab {\bf End If}\\
\tab \tab {\bf Until} $\ex(v)=0$ or $l(v)$ increases.\\
\tab \tab {\bf If} $v$ is still active, add $v$ to the rear of $Q$. \\
\tab {\bf End While}}}
\fbox{
\parbox{0.95\textwidth}{
{\em Tree-Push/Relabel}$(v)$ \\
\tab {\bf Applicability}: $v$ is an active tree root (in the dynamic tree).\\
\tab Let $\{v,w\}$ be $\current(v)$.\\
\tab {\bf If} $l(v)=l(w)+1$ {\bf and } $r_f(v,w)>0$ \tab \tab {\bf Case $(1)$}\\
\tab \tab Make $w$ the parent of $v$ by performing \\
\tab \tab \tab $\dtcv(v,-\infty),\dtcv(v,r_f(v,w))$, and $\dtlink(v,w)$.\\
\tab \tab Push excess from $v$ to $w$ using {\em Send$(v)$}.\\
\tab {\bf Else} (i.e. $l(v)\leq l(w)$ or $r_f(v,w)=0$) \tab \tab {\bf Case $(2)$} \\
\tab \tab {\bf If} $\{v,w\}$ is not the last edge in $v$'s list of edges. \tab \tab {\bf Case $(2a)$}\\
\tab \tab \tab Set $\current(v)$ to be the next edge in $v$'s list of edges.\\
\tab \tab {\bf Else} (i.e. $(v,u)$ is the last edge of $v$) \tab \tab {\bf Case $(2b)$}\\
\tab \tab \tab Set $\current(v)$ to be the first edge of $v$'s list of edges.\\
\tab \tab \tab Perform $\dtcut(u)$ and $\dtcv(u,\infty)$ for every child $u$ of $v$ (in the dynamic tree).\\
\tab \tab \tab Perform {\em Relabel}$(v)$.\\
\tab \tab {\bf End If}\\
\tab {\bf End If}
}}
\fbox{
\parbox{0.95\textwidth}{
{\em Send}$(v)$\\
\tab {\bf Applicability}: $v$ is active.\\
\tab {\bf While} $\dtfr(v)\neq v$ {\bf and} $\ex(v)>0$\\
\tab \tab Send $\delta\leftarrow\min\left(\ex(v),\dtfv\left(\dtfm(v)\right)\right)$ units of supply along\\
\tab \tab \tab the tree path from $v$ by performing $\dtcv(v,-\delta)$\\
\tab \tab {\bf While} $\dtfv\left(\dtfm(v)\right)=0$\\
\tab \tab \tab $u\leftarrow\dtfm(v)$\\
\tab \tab \tab Perform $\dtcut(u)$ followed by $\dtcv(u,\infty)$.\\
\tab \tab {\bf End While}\\
\tab {\bf End While}
}}
\fbox{
\parbox{0.95\textwidth}{
{\em Relabel}$(v)$\\
\tab {\bf Applicability}: $v$ is active, and $\forall w\in V$, $r_f(v,w)>0\implies l(v)\leq l(w)$.\\
\tab $l(v)\leftarrow l(v)+1$.
}}
\end{algorithm}
In the rest of this section, we show how capping labels at $h$ guarantees the running time to be $\tilde{O}(\vol(C)\cdot h)$ if $C$ is the set of nodes with sinks saturated (and $\vol(C)$ is the unweighted volume of $C$). 
\begin{lemma}
\label{lem:flow-runtime} Let $C=\left\{v:f(v)\geq d^W(v)\right\}$ be the set of saturated nodes when Unit-Flow finishes, the running time for {\em Unit-Flow} is $O(\vol(C)\cdot h\cdot \log |C|)$.
\end{lemma}
\begin{proof}
Throughout the execution no dynamic tree will grow to have more than $|C|$ nodes, since we only link a tree rooted at $v$ to another tree when $v$ is active (i.e. has excess supply, thus already saturated, so $v\in C$). Thus the time per dynamic-tree operation is $O(\ln |C|)$.

The running time of Unit-Flow is mostly in the {\em Tree-Push/Relabel} calls. Each {\em Tree-Push/Relabel} invocation takes
\begin{enumerate}[(i)]
    \item $O(1)$ time plus $O(1)$ dynamic-tree
operations;
\item $O(1)$ dynamic-tree operations per $\dtcut$ operation in invocations of {\em Send$(v)$} and
in Case $(2b)$ plus time for {\em Relabel} in Case $(2b)$. 
\end{enumerate}
To bound $(ii)$, the total {\em Relabel} time
is $O(|C|h)$ since we only need to raise the label of nodes in $C$ for at most $h$ times on each node. The total number of $\dtcut$ operations can be bounded as follows. Each $\dtcut$ operation can be attributed to a saturating push or to an edge scan for a {\em Relabel}. If we saturate an arc $(v,w)$ during a push, we know $v\in C$ and $l(v)=l(w)+1$. We won't be able to push along $(v,w)$ (in either direction) until we raise the label of $w$ to be above $l(v)$, so we can charge the $\dtcut$ to the {\em Relabel} of $w$. Similarly, we can charge the $\dtcut$ operations during the edge scan in {\em Relabel$(v)$} to that {\em Relabel}. In total, each node $v\in C$ will be charged $O(d(v))$ per {\em Relabel}, so $O(d(v)h)$ in total. The number of $\dtlink$ operations is at most $|C|-1$ more than the total number of $\dtcut$ operations, and thus can be bounded similarly. Summing over all vertices we get a bound of $O(\vol(C)h\ln|C|)$ for the running time from $(ii)$. 

The running time from $(i)$ is $O(\ln |C|)$ times the number of times we add a vertex to $Q$ (since we only call {\em Tree-Push/Relabel} when we remove a vertex from $Q$). Beyond the (at most $|C|$) vertices we add to $Q$ during initialization (i.e. nodes starting with $\Delta(v)\geq d^W(v)$), a node $v$ can only be added to $Q$ when $l(v)$ increases or as a result of $\ex(v)$ increases above $0$. The former happens at most $|C|h$ times in total, and the latter only happens during Case $(1)$ when we send supply along a tree path. In this case, the number of nodes with excess increasing above $0$ is at most $1$ more than the number of $\dtcut$ operations, and we already bounded the total number of $\dtcut$ operations by $O(\vol(C)h)$ in the previous paragraph. The number of times we get to Case $(1)$ is bounded by the total number of $\dtlink$ operations, which again is bounded by $O(\vol(C)h)$ as in the previous paragraph. In total this bounds the total number of times a vertex is added to $Q$ to be $O(\vol(C)h)$, and thus the running time from $(i)$ is $O(\vol(C)h\ln|C|)$.

The additional running time if we need to find a low conductance cut will be $O(\vol(C))$ since it's just a sweeping cut algorithm using the vertex labels, see details in the proof of~\Cref{thm:unit-flow}.
\end{proof}

\subsection{Analysis of Unit-Flow}
\label{app:unit-flow-analysis}
The {\em Unit-Flow} procedure (Algorithm~\ref{alg:unit-flow}) is a fairly straightforward implementation of the push-relabel framework with a cap on the labels: If a vertex has label $h$, and we cannot push remaining excess from the vertex, instead of relabeling it to $h+1$, we leave the vertex at label $h$ and it's never considered as active from then on. Upon termination, we have a pre-flow $f$ , and labels $l$ on vertices. We make the following observations.
\begin{observation} 
\label{obs:unit-flow}
During the execution of {\em Unit-Flow}, we have
\begin{enumerate}[(1)]
\item If $v$ is active at any point, the final label of $v$ cannot be $0$. The reason is that $v$ will remain active until either $l(v)$ is increased to $h$, or its excess is pushed out of $v$, which is applicable only when $l(v)$ is larger than $0$ at the time of the push.
\item Each vertex $v$ is a sink that can absorb up to $d^W(v)$ units of supply, so we call the $f(v)-\ex(v)=\min(f(v),d^W(v))$ units of supply remaining at $v$ the {\em absorbed} supply. The amount of absorbed supply at $v$ is between $[0,d^W(v)]$, and is non-decreasing. Thus any time after the point that $v$ first becomes active, the amount of absorbed supply is $d^W(v)$. In particular any time the algorithm relabels $v$, there have been $d^W(v)$ units of supply absorbed by $v$.
\end{enumerate}
 Upon termination of {\em Unit-Flow} procedure, we have
\begin{enumerate}[(1)]
\setcounter{enumi}{2}
\item For any edge $\{v,u\}\in E$, if the labels of the two endpoints differ by more than $1$, say $l(v)-l(u)>1$, then arc $(v,u)$ is saturated. This follows directly from a standard property of the push-relabel framework, where $r_f(v,u)>0$ implies $l(v)\leq l(u)+1$. 
\end{enumerate}
\end{observation}
Although {\em Unit-Flow} may terminate with $\ex(v)>0$ for some $v$, we know all such vertices must have label $h$, as the algorithm only stops trying to route $v$'s excess supply to sinks when $v$ reaches level $h$. Thus we have the following lemma. Also note that once some supply is absorbed by a sink or is left as excess on a node at height $h$ when Unit-Flow stops, even if $\Delta(\cdot)$ is incremental and we resume the execution of Unit-Flow, the absorbed supply remains absorbed, and similarly for the excess supply on node with label $h$. Thus it is well-defined to discuss these concepts at the time Unit-Flow terminates with some interim $\Delta(\cdot)$.
\begin{lemma}
\label{lemma:excess}
Upon termination of {\em Unit-Flow} with input $(G,\Delta,h,U,\beta)$, where $\Delta(v)\leq \beta d^W(v)$ for all $v$, the pre-flow and labels satisfy
\begin{enumerate}[(a)]
\item If $l(v)=h$, $(\beta+U) d^W(v)\geq f(v)\geq d^W(v)$; 
\item If $h-1\geq l(v)\geq 1$, $f(v)=d^W(v)$; 
\item If $l(v)=0$, $f(v)\leq d^W(v)$.
\end{enumerate}
\end{lemma}
\begin{proof}
By Observation~\ref{obs:unit-flow}.$(2)$, any vertex with label larger than $0$ must have $f(v)\geq d^W(v)$. The algorithm terminates when there is no active vertices, i.e. $\ex(v)>0\implies l(v)=h$, so all vertices with label below $h$ must have $f(v)\leq d^W(v)$. Moreover, $f(v)\leq wd^W(v)$ since at the beginning $f(v)=\Delta(v)\leq \beta d^W(v)$, and the pre-flow routing can send at most an additional $U\cdot d^W(v)$ units of supply to $v$ since that's the total capacity of edges incident to $v$ for the flow problem.
\end{proof}
Now we can prove the main theorem about {\em Unit-Flow}. Everything other than the running time largely follows the analysis of the unweighted version from~\cite{DBLP:journals/siamcomp/HenzingerRW20}, and the running time follows mechanically from the Push-relabel analysis in~\cite{GT88} using the dynamic-tree data structure from~\cite{ST81}
\unitFlow*
\begin{proof}
We use the labels at the end of {\em Unit-Flow} to divide the vertices into groups
\[
B_i=\{v|l(v)=i\}
\]
If $B_h=\emptyset$, no vertex has positive excess, so all $|\Delta(\cdot)|$ units of supply are absorbed by sinks, and we end up with case $(1)$.

If $B_h\neq \emptyset$, but $B_0=\emptyset$, by Lemma~\ref{lemma:excess} every vertex $v$ has $f(v)\geq d^W(v)$, so we have $\sum_vd^W(v)=\vol^W(G)$ units of supply absorbed by sinks, and we end up with case $(2)$.

{\bf Case $(3)$:} When both $B_h$ and $B_0$ are non-empty, we compute the cut $(A,V\setminus A)$ as follows: Let $S_i=\cup_{j=i}^h B_j$ be the set of vertices with labels at least $i$. We sweep from $h$ to $1$, and let $A$ be the first $i$ such that $\Phi(S_i)\leq 20(\frac{\ln \vol^W(G)}{h}+\frac{\beta}{U})$. The claims in $(3)(a)$ follow directly from $S_h\subseteq A\subseteq S_1$ and Lemma~\ref{lemma:excess}. 

For $(3)(b)$, We will show that there must exists some $S_i$ satisfying the conductance bound. For any $i$, an edge $e=\{v,u\}$ across the cut $S_i$, with $v\in S_i,u\in V\setminus S_i$, must be one of the two groups:
\begin{enumerate}
\item In the residual network, the arc $(v,u)$ has positive residual capacity $r_f(v,u)>0$, so $l(v)\leq l(u)+1$. But we also know $l(v)\geq i > l(u)$ as $v\in S_i,u\in V\setminus S_i$, so we must have $l(v)=i,l(u)=i-1$. In the Push-relabel framework, such edges are called {\em eligible arcs}.
\item In the residual network, if $r_f(v,u)=0$, then $(v,u)$ is a saturated arc sending $U\cdot w(e)$ units of supply from $S_i$ to $V\setminus S_i$.
\end{enumerate}
Suppose the total edge weight of edges in the first group is $z_1(i)$, and the total weight of edges in the second group is $z_2(i)$. By the following region growing argument, we can show there exists some choice of $i=i^*$, such that 
\begin{equation}\label{eq:eligible}
z_1(i^*)\leq \frac{10\min(\vol^W(S_{i^*}),\vol^W(G)-\vol^W(S_{i^*}))\ln \vol^W(G)}{h}
\end{equation}
If $\vol^W(S_{\floor*{h/2}})\leq \vol^W(G)/2$, we start the region growing argument from $i=h$ down to $\floor*{h/2}$. By contradiction, suppose $z_1(i)\geq \frac{10\vol^W(S_i)\ln \vol^W(G)}{h}$ for all $h\geq i\geq \floor*{h/2}$, which implies $\vol^W(S_i)\geq \vol^W(S_{i+1})(1+\frac{10\ln \vol^W(G)}{h})$ for all $h>i\geq \floor*{h/2}$. Since $\vol^W(S_h)=\vol^W(B_h)\geq 1$ and $h\ge \ln \vol^W(G)$, we will have $\vol^W(S_{\floor*{h/2}})\geq(1+\frac{10\ln \vol^W(G)}{h})^{h/2}\gg \vol^W(G)$, which gives a contradiction. The case where $\vol^W(S_{\floor*{h/2}})> \vol^W(G)/2$ is symmetric, and we run the region growing argument from $i=1$ up to $\floor*{h/2}$ instead. 

For any $i$, we can bound $z_2(i)$ as follows. Since the pre-flow pushes $z_2(i)U$ units of supply in total from $S_i$ to $V\setminus S_i$ along the saturated arcs in the second group, $z_2(i)U$ is at most $\sum_{v\in S_i}\Delta(v)+z_1(i)U$, i.e. the sum of the source supply in $S_i$ plus the supply pushed into $S_i$ along the eligible arcs in the first group (and the total capacity of those arcs is $z_1(i)U$). As $\Delta(v)\leq \beta d^W(v)$ for all $v$, we know 
\[
z_2(i)\leq \frac{\beta\vol^W(S_i)}{U}+z_1(i)
\]
On the other hand, $z_2(i)U$ is at most $\sum_{v\in V\setminus S_i}f(v)+z_1(i)U$, as the $z_2(i)U$ units of supply pushed into $V\setminus S_i$ either remain at vertices in $V\setminus S_i$, or back to $S_i$ along the reverse arcs of the eligible arcs in the first group, and the total capacity of these reverse arcs is also $z_1(i)U$. Since any $v\in V\setminus S_i$ is not with label $h$, thus $f(v)\leq d^W(v)$, then we get
\[
z_2(i)\leq \frac{\vol^W(V\setminus S_i)}{U}+z_1(i)
\]
The two upper-bounds of $z_2(i)$ together give
\begin{equation}\label{eq:saturated}
z_2(i)\leq \frac{\beta\min(\vol^W(S_i),\vol^W(G)-\vol^W(S_i))}{U}+z_1(i)
\end{equation}

We know there exists some $i^*$ such that $z_1(i^*)$ is bounded by~\eqref{eq:eligible}, together with~\eqref{eq:saturated}, we have
\[
z_1(i^*)+z_2(i^*)\leq \min(\vol^W(S_{i^*}),2m-\vol^W(S_{i^*}))(\frac{20\ln \vol^W(G)}{h} +\frac{\beta}{U})
\] 
thus $\Phi(S_{i^*})\leq \frac{20 \ln \vol^W(G)}{h}+\frac{\beta}{U}$, which completes the proof of $(3)(b)$.

For $(3)(c)$, we again use the observation that the supply ending in $A$ (i.e. $S_{i^*}$) can only come from either the source supply started in $A$ or sent into $A$ via the reverse arcs of the eligible arcs in the first group. The former is at most $\beta\vol^W(A)$ and the latter is at most $z_1(i^*)U\leq \frac{10U\vol^W(A)\ln\vol^W(G)}{h}$ by~\eqref{eq:eligible}. Moreover, since all nodes in $A$ have their sinks saturated, we know the following about the excess
\[
\ex(T) + \vol^W(A) \leq \beta\vol^W(A)+\frac{10U\vol^W(A)\ln\vol^W(G)}{h},
\]
which gives the bound on $\ex(T)$ in the theorem statement.

The running time claim follows from~\Cref{lem:flow-runtime}.
\end{proof}
\end{appendix}
\end{document}